\documentclass[11pt]{article}

\usepackage[T1]{fontenc}
\usepackage{amsfonts,amsmath,amsthm,amssymb,dsfont,mathtools}
\usepackage{bbm}
\usepackage[margin=1in]{geometry}
\usepackage{fullpage}
\usepackage{enumitem}
\usepackage[linesnumbered,boxed,ruled,vlined]{algorithm2e}
\usepackage[colorlinks,citecolor=blue,linkcolor=blue,urlcolor=red,pagebackref]{hyperref}
\usepackage[capitalise]{cleveref}
\usepackage{url}
\usepackage{thmtools, thm-restate}
\usepackage{mathdots}
\usepackage[normalem]{ulem}
\usepackage{xspace}
\xspaceaddexceptions{]\}}
\usepackage{comment}
\usepackage{tabu}
\usepackage{framed}
\usepackage{float,wrapfig}
\usepackage{xcolor}
\usepackage{graphicx}
\usepackage{tikz}
\usetikzlibrary{decorations.pathreplacing,calligraphy,positioning,shapes}

\allowdisplaybreaks

\newtheorem{theorem}{Theorem}[section]

\newtheorem{lemma}[theorem]{Lemma}
\newtheorem{corollary}[theorem]{Corollary}
\newtheorem{claim}[theorem]{Claim}

\newtheorem{observation}{Observation}
\newtheorem{definition}[theorem]{Definition}

\newtheorem{example}[theorem]{Example}

\newcommand{\todo}[1]{}
\newcommand{\zoya}[1]{}
\newcommand{\manish}[1]{}
\newcommand{\josh}[1]{}
\newcommand{\erik}[1]{}
\newcommand{\eriknew}[1]{}
\newcommand{\cnote}[1]{}

\newcommand{\scratch}{\text{scratch}}

\newcommand{\truemem}{computation graph memory}

\newcommand{\sintree}[0]{starborescence}
\newcommand{\sintrees}[0]{starborescences}
\newcommand{\Sintree}[0]{Starborescence}
\newcommand{\Sintrees}[0]{Starborescences}
\newcommand{\staff}[0]{\varsigma}
\newcommand{\bend}[0]{B}
\newcommand{\bends}[0]{\mathcal{B}}
\newcommand{\guess}[0]{\theta}
\newcommand{\negativesegments}[1][i]{\mathcal{S}_{#1}^-}
\newcommand{\positivesegments}[1][i]{\mathcal{S}_{#1}^+}
\newcommand{\segments}{\mathcal{S}}
\newcommand{\negativesegment}[0]{\text{seg}^-}
\newcommand{\positivesegment}[0]{\text{seg}^+}
\newcommand{\segment}[0]{\text{seg}}
\newcommand{\negativesegmentsstar}[1][i]{\mathcal{S}_{#1}^{\star -}}
\newcommand{\positivesegmentsstar}[1][i]{\mathcal{S}_{#1}^{\star +}}

\newcommand{\outneighbors}[1]{\delta^+(#1)}

\newcommand{\nodesumcost}[1]{\mathrm{mem}_\text{NS}(#1)}

\newcommand{\nodeweightedprofile}[1]{\mathrm{profile}(#1)}
\newcommand{\nodesumprofile}[1]{\mathrm{profile}(#1)}
\newcommand{\profile}[1]{\mathrm{profile}(#1)}

\newcommand{\memprofile}[1]{\mathrm{profile}(#1)}

\newcommand{\internalprofile}[1]{\mathrm{profile}^*(#1)}

\newcommand{\set}{\mathrm{set}}
\newcommand{\ledot}{\mathrel{\ooalign{$\le$\cr
  \hidewidth\raise0.225ex\hbox{$\cdot\mkern0.5mu$}\cr}}}
\newcommand{\cupdot}{\mathbin{\ooalign{$\cup$\cr\hidewidth\raise0.225ex\hbox{$\cdot$}\hidewidth}}}
\newcommand{\capdot}{\mathbin{\ooalign{$\cap$\cr\hidewidth\raise0.05ex\hbox{$\cdot$}\hidewidth}}}

\newcommand{\sz}[1]{\mathit{size({#1})}}
\newcommand{\scr}[1]{\mathit{scratch({#1})}}

\newcommand{\sizefn}{\mathit{size}}
\newcommand{\scratchfn}{\mathit{scratch}}

\newcommand{\ns}[1]{\hat{#1}}

\newcommand{\nsff}[1]{\mathrm{\nu}({#1})}

\newcommand{\leftpeak}[1]{\mathrm{peak}_\leftarrow\left(#1\right)}
\newcommand{\rightpeak}[1]{\mathrm{peak}_\rightarrow\left(#1\right)}
\newcommand{\leftvalley}[1]{\mathrm{valley}_\leftarrow\left(#1\right)}
\newcommand{\rightvalley}[1]{\mathrm{valley}_\rightarrow\left(#1\right)}
\newcommand{\mincut}{\mathrm{mincut}}
\newcommand{\valley}[1]{\mathrm{valley}[#1]}
\newcommand{\peak}[1]{\mathrm{peak}[#1]}

\newcommand{\fancyg}{\mathcal{G}}
\newcommand{\ignore}[1]{}

\newcommand{\abs}[1]{\left| #1 \right|}
\newcommand{\sorted}[1][S]{\text{sorted}(#1)}

\newcommand{\defeq}{:=}

\newcommand{\memduring}[1]{\mathrm{mem_{during}}\left(#1\right)}
\newcommand{\memafter}[1]{\mathrm{mem_{after}}\left(#1\right)}

\newcommand{\startvertex}[1]{#1^{(s)}}
\newcommand{\finishvertex}[1]{#1^{(f)}}
\newcommand{\releasevertex}[1]{#1^{(r)}}

\DeclareMathOperator{\poly}{poly}

\title{New Tools for Peak Memory Scheduling}

\author{Ce Jin, Manish Purohit, Zoya Svitkina, Erik Vee, Joshua R. Wang}

\date{}

\begin{document}

\setcounter{page}{-1} \clearpage
\maketitle
\thispagestyle{empty}
\begin{abstract}

We study scheduling of computation graphs to minimize peak memory consumption, an increasingly critical task due to the surge in popularity of large deep-learning models. This problem corresponds to the weighted version of the classical one-shot black pebbling game.
We propose the notion of a dominant schedule to capture the idea of finding the ``best'' schedule for a subgraph and introduce new tools to compute and utilize dominant schedules. Surprisingly, we show that despite the strong requirements, a dominant schedule exists for any computation graph; and, moreover, that it is possible to compute the dominant schedule efficiently whenever we can find optimal schedules efficiently for a particular class of graphs (under mild technical conditions).

We apply these new tools to analyze trees and series-parallel graphs. We show that the weighted one-shot black pebbling game is strongly NP-complete even when the graph is an out-tree – or simpler still, a pumpkin, one of the simplest series-parallel graphs. On the positive side, we design a fixed-parameter tractable algorithm to find a dominant schedule (hence also a peak memory minimizing schedule) for series-parallel graphs when parameterized by the out-degree. This algorithm runs in time $2^{O(d \log d)} \cdot \poly(n)$ for series-parallel graphs with $n$ nodes and maximum out-degree $d$; for pumpkins, we can improve the dependence on $d$ to $O(2^d \cdot \poly(n))$.

\newpage 
\tableofcontents 

\end{abstract}
\newpage

\section{Introduction}

Scheduling has been a core problem in theoretical computer science for decades. The recent surge in machine learning’s popularity has pushed memory-efficient scheduling to the forefront. Even seemingly minor improvements can have oversized impact and lead to huge resource savings (such as carbon footprint or money) in real-life applications. It is clear that larger machine learning (ML) models have better quality. But many simply will not fit in memory unless properly optimized, particularly when we expect them to run on handheld devices. 

We tackle this challenging area, producing practical algorithms by developing stronger theoretical tools. We study the problem of scheduling directed acyclic graphs (DAGs) to minimize the peak memory of the schedule. This problem can also be viewed as a weighted version of the classical one-shot black pebbling game. Consider a (directed and acyclic) computation graph where nodes of the graph correspond to operations, and an edge from node $u$ to node $v$ indicates that the output produced by operation $u$ is used by operation $v$. This output must remain in memory until $v$ runs. The goal is to find a schedule, i.e.\ a topological ordering of the nodes, that minimizes the maximum memory consumed over the course of the schedule. (In the language of pebbling games, the schedule is our pebbling strategy, i.e.\ the order in which we place pebbles, and we wish to minimize the maximum total weight of pebbles used at any timestep.)

Unfortunately, one-shot black pebbling (even with unit weights) is NP-hard to approximate within any constant factor (assuming the Small Set Expansion conjecture) \cite{wu2014inapproximability}, so finding peak-minimizing schedules is hard. The good news is that the computation graphs that arise from ML training and inference tasks often have simple, recognizable structures that can be exploited to make optimal scheduling more efficient. For example, \cite{kumar2019efficient}  observe that ML computation graphs typically have low treewidth.  Large ML models are often comprised of simpler building blocks; efficient scheduling algorithms can exploit this structure to first find a sub-schedule for each individual building block, and then to extend those to obtain a complete schedule. In this paper, we introduce tools for peak memory scheduling to exploit such strategies.

We use these tools to provide optimal polynomial-time algorithms for series-parallel (SP) graphs with bounded out-degree.\footnote{In practice, computation graphs arising from ML inference tasks are typically \emph{nearly} (in a non-technical sense) series-parallel, with at most a few nodes having large degree; although there are heuristics to deal with this, such issues are beyond the scope of this paper.} Specifically, we give an exact algorithm running in time $2^{O(d \log d)} \cdot \poly(n)$ for SP graphs of out-degree $d$ and $n$ nodes; for the simpler pumpkin graph (defined in \cref{sec:pumpkin}), we can reduce the time to $O(2^d \cdot \poly(n))$. 
As a complementary result, we show that solving SP graphs optimally in general is strongly NP-hard, implying that such an exponential dependence is necessary. Our hardness results hold even for out-trees and pumpkin graphs.

We note that the peak memory scheduling problem that we study can be approximated within a factor of $O(\log^{3/2} n)$ by using the $O(\sqrt{\log n})$ approximation algorithm for directed balanced separator \cite{ACMM05} and a standard recursive cut technique. We omit the details from this paper.

\subsection{Comparison of Models}\label{sec:comparison}
Scheduling is a well-studied area spanning decades of both theoretical and compiler-focused research. As such, there have been many proposed formalizations. In this paper, we focus primarily on the computation graph model (that we define in \cref{subsec:computation_graphs}) and weighted one-shot black-pebbling (defined in \cref{sec:pebbling}). These models are the same in spirit: both use the same underlying graph and produce schedules to minimize peak memory. The key difference is that computation graphs allow the notion of \emph{scratch} memory, which models the fact that additional memory may be needed while computing a given node. Because of this, computation graphs are somewhat more expressive than the classical pebbling formulation. When scratch memory is set to zero, the two models become equivalent.

Due to this additional expressive power, and because of its familiarity for compiler-focused researchers, we state our positive results and algorithms in terms of computation graphs. (Note that since computation graphs generalize pebbling, our results apply to pebbling as well.) On the other hand, we present our lower bounds in terms of the weighted one-shot black pebbling game. Again, since computation graphs slightly generalize the pebbling game we study, our hardness results immediately apply to computation graphs as well.

For purely technical reasons, our main theoretical results are proven in a third model, node-sum, defined in \cref{sec:preliminaries}. This model serves to reduce the difficulty in reasoning about memory use in graphs, but at first glance might seem slightly unintuitive. In the node sum model, nodes are allowed to have positive or negative weights. The negative weights model the reduction in memory that happens when data is released from memory (or a pebble is removed from the graph). The advantage of this is that the memory at any given time $t$ is simply the sum of nodes scheduled up to time $t$, and thus each node only results in one memory usage change, not the two changes that occur in the computation graph model and weighted one-shot black pebbling (one when allocated and one when deallocated). Furthermore, problems in the other two models can be reduced to the node sum model via graph transformations (which involve splitting an original node into an allocation node, a deallocation node, and possibly another node to handle scratch memory).

The disadvantage of the node sum model is that we must change the underlying graph structure in order to reduce a scheduling instance on a computation graph to it. In particular, additional nodes (with negative weight) are introduced to model deallocation of data from memory. This means that SP computation graphs lose the SP property when converted into the equivalent node sum formulation. However, our results on dominance and its use (``linearization'') hold for general DAGs. Thus, these results are presented in the node sum model (to simplify the proofs), but they also apply to the \truemem\ model.

\subsection{Pebbling Games}\label{sec:pebbling}

There has been a long line of theoretical work on pebbling games on DAGs, dating back to the problem's proposal in the 1970s \cite{PH70, sethi1973complete}.  Although these games have seen several variations emerge over the years, one of their primary use cases is as a useful theoretical model for how much space is necessary to perform a computation. The basic setup of a pebbling game is as follows. A computation involving operations and dependencies between them is encoded as a DAG. To play the pebbling game on this graph, we begin in a state where no nodes have pebbles on them and repeatedly perform one of the following moves.
\begin{itemize}
  \item A pebble can only be added to a node if all of its immediate predecessors have pebbles.
  \item A pebble can be removed from a node at any time.
\end{itemize}
The typical goal of this game is to manage to place a pebble on a specified ``output'' node. 
Given a pebbling strategy, also known as a schedule, its cost is the maximum number of pebbles used at any point in time.

With just this ruleset, we arrive at the simplest variant of the problem, now referred to as the ``black'' pebbling game. Since each pebble incurs the same cost, this variant nicely models the task of register allocation, since we think of each register as holding a singular value. For more complex computer architectures and theoretical applications, the game can be made appropriately more complicated as well. For example, ``black-white'' pebbling games introduce white pebbles to model non-deterministic guesses that must be confirmed \cite{cook1974storage}, and ``red-blue'' pebbling games use red pebbles to represent fast memory and blue pebbles to represent slow memory \cite{jia1981complexity}.

We study the weighted version of one-shot black pebbling, where each node has a non-negative weight. The cost at any time step is then the total weight of pebbled nodes, and the objective is to minimize the maximum cost over all time steps.

\subsection{Technical Overview and Our Contributions}

\paragraph{Dominance.}
Our main technical contribution is the framework of dominance. In \cref{sec:dominant-def}, we give the formal definition of dominance and prove the existence of a (most) dominant schedule in general DAGs. This proof is constructive and shows how dominant schedules can be found using an oracle that can compute peak memory minimizing schedules.

Dominance is a careful strengthening over simply comparing the peak memory of two schedules to decide which is better. Informally, one schedule dominates another if it is possible to align the two schedules, by adding delays, in such a way that when executed simultaneously, the memory consumption of the dominating schedule is at most that of the dominated schedule, at all points in time. This not only implies that the dominating schedule has a smaller peak memory cost, but it is a useful notion when trying to design a recursive algorithm.

The main technical hurdle that we overcome is demonstrating that there always exists a dominant schedule (one that dominates all others). To help illustrate why this result is surprising, let us focus our attention on two key measures of a particular schedule's quality: (i) its \emph{peak} memory, i.e.\ the highest amount of memory it ever uses, and (ii) its \emph{valley} memory, i.e.\ the lowest amount of memory it uses at any time after starting the first operation and before finishing the last operation. When one schedule dominates another, it has both lower peak memory \emph{and} lower valley memory.
So the existence of a dominant schedule implies that there is no trade-off between these two measures. In fact, we could continue to recursively define additional peak and valley measures (which is discussed in \Cref{sec:segments}), and the dominant schedule must simultaneously optimize all of these measures.

\begin{restatable*}{theorem}{existsdominant}
\label{thm:existsdominant}
  For any directed acyclic graph $G = (V, E)$ with node weights $w_v \in \mathbb{R}$, there exists a schedule $\sigma$ of $G$ that dominates all schedules of $G$.
\end{restatable*}

As we mentioned, our proof of this result is constructive. Note that since a dominant schedule is also a peak memory minimizing schedule, it must be hard to find since finding the latter is NP-hard. To get around this issue, we show that if we have an oracle that returns a peak memory minimizing schedule, then a linear number of calls to it suffices to compute a dominant schedule.

\paragraph{Linearization.}
The main consequence of our dominance framework is the ability for recursive algorithms to perform linearization, which is the process of simplifying a graph by replacing a (complicated) subgraph with a simple directed path that corresponds to a schedule of that subgraph.
Intuitively, for linearization to work, the schedule that is used for it must have low peaks (in order not to contribute too much to the overall peak), but it also must have low valleys. The reason that low valleys are useful is that they are good places to pause the execution of this path while other portions of the surrounding graph run. As dominant schedules have both good peaks and valleys, they can be used for linearization.
In \cref{sec:linearization}, we describe the technical condition (\emph{isolation}) under which linearization can be applied. We then prove that replacing an isolated subgraph with the directed path induced by its dominant schedule results in a (simpler) graph whose optimal peak memory is identical to the original. 

\paragraph{Series-parallel graphs.}
We demonstrate the power of the dominance framework by showing how to compute dominant, and hence also peak memory minimizing, schedules for SP computation graphs. SP graphs (formally defined in \cref{def:sp-graph}), an important family of graphs, appear often in the context of computation graphs. Common ML operations like sharding or scatter-gather naturally produce SP graphs and subgraphs.

Series-parallel graphs have a useful recursive structure. This allows us to apply linearization repeatedly, yielding an efficient algorithm when the graph has bounded degree.
\begin{restatable*}{theorem}{spmain}
  Let $G=(V, E)$ be a series-parallel graph with $n$ nodes and let $d$ denote the maximum out-degree of any node in $G$. Then there is an algorithm to find the dominant schedule for $G$ (in the computation graph memory model) in time $2^{O(d \log d)}\cdot \poly(n)$.
\end{restatable*}

In \cref{sec:pumpkin}, we present simpler and faster algorithms for a restricted class of SP graphs and for out-trees.

\paragraph{Hardness.} Finally, we show that the exponential dependence on the out-degree $d$ in SP graphs is unavoidable. In \cref{sec:hardness}, we show that finding an optimal strategy for the weighted one-shot black pebbling game on a SP graph is strongly NP-hard, even on the simplest class of SP graphs (called \emph{pumpkins}, formally defined in \cref{def:pumpkin}) and also on out-trees with a slight tweak to the construction. As a trivial corollary, it is also strongly NP-hard to minimize peak memory in the \truemem\ model, even on pumpkin graphs or out-trees.

\begin{restatable*}{theorem}{pumpkinhardness}
\label{thm:pumpkin-hardness}
It is strongly NP-hard to minimize the pebbling cost in the weighted one-shot black pebbling game on pumpkin graphs.
\end{restatable*}

\begin{restatable*}{corollary}{treehardness}
It is strongly NP-hard to minimize the pebbling cost in the weighted one-shot black pebbling game on out-trees.
\end{restatable*}

\subsection{Related work}

Graph pebbling has a rich history \cite{PH70,sethi1973complete,lengauer1981black,kirousis1986searching} in the theoretical computer science literature, which has investigated the problem as well as key variants such as black-white pebbling (capturing nondeterminism) \cite{cook1974storage,hertel2010pspace} and red-blue pebbling (modeling multi-level memory) \cite{jia1981complexity,demaine2018red}.

Even the simplest variants of this problem tend to be difficult to solve optimally or to approximate well. For example, the black pebbling game is known to be PSPACE-hard~\cite{gilbert1979pebbling}. This version of pebbling permits strategies that use a value, throw it out, then recompute it again later; in compiler literature, this is referred to as \emph{rematerialization}. Here, we study the \emph{one-shot} black pebbling game, which does not allow such recomputation. As noted, this problem is NP-hard, as well as NP-hard even to approximate within any constant factor assuming the Small Set Expansion conjecture \cite{wu2014inapproximability}.

Peak memory scheduling is also closely related to the Minimum-Cut Linear Arrangement (MCLA) problem. Given an edge-weighted (directed or undirected) graph $G$, the MCLA problem asks for an ordering of vertices of $G$ so that the maximum total weight of edges crossing any cut induced by the ordering is minimized. The two problems are reducible to each other, but without preserving the graph structure. On general DAGs, the best known approximation algorithm for MCLA yields an $O(\log^{3/2} n)$ approximation~\cite{leighton1999multicommodity,arora2009expander}. 
When the underlying graph is a tree, \cite{liu1987application} provides an optimal algorithm in this edge-weighted model. For SP graphs, \cite{kayaaslan2018scheduling} introduces a generalization of MCLA  that has weights on both nodes and edges and provides an optimal polynomial-time algorithm. In their memory model, the output of an operation is only consumed by one other operation.
Although that work provides elegant insights into such graphs, this simplified memory model isn't directly applicable to pebbling games or real-world applications. The transformation between the two models does not preserve the SP structure, and the results no longer hold.
\section{Model and Preliminaries}
\label{sec:preliminaries}

We first consider scheduling to minimize peak memory consumption in a simple, idealized memory model, referred to as the \emph{node sum memory model}. In this model, each node either allocates or deallocates memory and edges between nodes specify precedence constraints. This model was also studied by Kayaaslan et al.~\cite{kayaaslan2018scheduling} under the name \emph{cumulative weight} model.

Let $G = (V, E, w)$ be a node-weighted DAG where $w_v \in \mathbb{R}$ (possibly negative) denotes the weight of node $v \in V$. Let $n = |V|$ be the number of nodes in $G$. A schedule $\sigma = (\sigma_1, \sigma_2, \ldots, \sigma_n)$ is a topological ordering of $V$, that is, an ordering of the $n$ nodes such that $(\sigma_j, \sigma_i) \notin E$ for all $1 \leq i < j \leq n$.

For integers $i, j$, we use $i:j$ to denote the set $\{i, i+1, \cdots, j\}$ (which is the empty set if $j < i$). For any sequence $\sigma$, we use subscripts to extract single elements or a contiguous subsequence of elements (when combined with $i:j$ notation). In particular, $\sigma_i$ is the $i$th element of $\sigma$ and $\sigma_{i:j}$ is the subsequence $(\sigma_i, \ldots, \sigma_j)$. Again, we allow $i > j$ (for example, $\sigma_{1:0}$), in which case $\sigma_{i:j}$ is the empty sequence.  For any sequence $\sigma$, let $\set(\sigma) := \{\sigma_1,\sigma_{2},\dots,\sigma_{|\sigma|}\}$ denote the set of its elements. Throughout the paper we use such notation for sequences of vertices (schedules), and also for sequences of memory consumption (called {\em profiles}).
  
\begin{definition}[Node Sum Memory Model]
\label{defn:nodesummodel}
  Let $G = (V, E, w)$ be a node-weighted DAG. 
  For any set $S\subseteq V$, we define
    $\nodesumcost{S} := \sum_{u \in S} w_u$.
  Let $\sigma = (\sigma_1, \sigma_2, \ldots, \sigma_n)$ be an arbitrary schedule for $G$, and for any index $1\leq i \leq n$, let $S_i := \set(\sigma_{1:i})$ be the set of the first $i$ nodes in the schedule.
  Then, the memory used by schedule $\sigma$ at time $i$ is defined to be $\nodesumcost{S_i}$, and we define the
  memory profile of schedule $\sigma$ to be the sequence induced by the above for $i$ going from 0 to $n$:
  \begin{align*}
    \nodesumprofile{\sigma} &:= \left(0 = \nodesumcost{S_0}, 
        \nodesumcost{S_1}, 
        \nodesumcost{S_2}, 
        \dots, 
        \nodesumcost{S_n} \right).
  \end{align*}
 The peak memory of schedule $\sigma$ is defined as expected:
$$\peak{\sigma} = \max_{0 \leq i \leq n} \{\nodesumcost{S_i}\}.$$
\end{definition}

We note that profiles always start with zero (since not scheduling any node always costs zero), but in general, a profile can have negative values. The \emph{Peak Memory Scheduling Problem} is to find a schedule $\sigma$ that minimizes the $\peak{\sigma}$.

\subsection{Dominant Schedules}
\label{sec:dominant-def}

The memory profile of a schedule $\sigma$ indicates the amount of memory consumed at each time point when the underlying tasks are executed in the order dictated by $\sigma$. Informally, we would like an ``ideal'' schedule for a graph $G$ to be one that consumes the least possible amount of memory at all points in time, not just at its peak memory. Such an ideal schedule would allow us to realize a recursive scheduling strategy where we locate an induced subgraph $G_1 \subset G$ (suitably isolated from the rest of $G$), find the ideal schedule $\sigma_1$ for $G_1$, and then only consider schedules for $G$ that are consistent with $\sigma_1$ on $G_1$. The following example demonstrates how a peak memory minimizing schedule for $G_1$ is insufficient.

\begin{example}[Peak Memory Minimizing Schedules are Not Ideal]
\label{ex:not-ideal}
  
  \begin{figure}
\centering
\scalebox{0.8}{
\begin{tikzpicture}[%
  auto,
  scale=0.8,
  hollownode/.style={
    circle,
    draw=black,
    inner sep=1pt,
    minimum size=40pt,
  },
  redblock/.style={
    rectangle,
    solid,
    draw=red,
    fill=white,
    text=black,
    align=center,
  },
  ]
  \node[hollownode] (v1) at (0, 0)     {$v_1: 0$};
  \node[hollownode] (v2) at (-2, -2.5) {$v_2: +10$};
  \node[hollownode] (v3) at (-2, -5)   {$v_3: -9$};
  \node[hollownode] (v4) at (-4, -7.5) {$v_4: +4$};
  \node[hollownode] (v5) at (-4, -10)   {$v_5: -3$};
  \node[hollownode] (v6) at (0, -7.5) {$v_6: +4$};
  \node[hollownode] (v7) at (0, -10)   {$v_7: -3$};
  \node[hollownode] (v8) at (3, -2.5)  {$v_8: +9$};
  \node[hollownode] (v9) at (3, -5)    {$v_9: -5$};

  \draw[->] (v1) -- (v2);
  \draw[->] (v2) -- (v3);
  \draw[->] (v3) -- (v4);
  \draw[->] (v3) -- (v6);
  \draw[->] (v4) -- (v5);
  \draw[->] (v6) -- (v7);
  \draw[->] (v1) -- (v8);
  \draw[->] (v8) -- (v9);

  \draw[red, thick, dotted] (-5.3, -1.2) rectangle (1.3, -11.3) node[redblock] {$G_1$};
\end{tikzpicture}
}
\caption{Example graph to illustrate that not all peak memory minimizing schedules are equal. Nodes are labeled with their name $v_i$ and weight $w_{v_i}$. See \cref{ex:not-ideal}.}
\label{fig:tikz-not-ideal}
\end{figure}
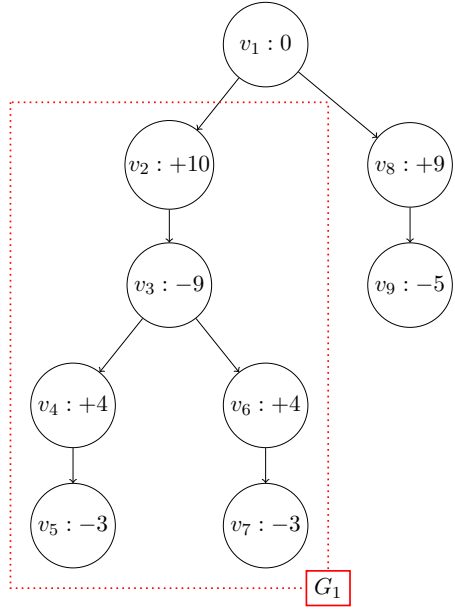
    Consider the graph $G$ shown in Figure \ref{fig:tikz-not-ideal}. 
  In particular, let us consider schedules for the subgraph $G_1$ induced by nodes $\{v_i\}_{i=2}^7$. Note that $\sigma_1 = (v_2, v_3, v_4, v_6, v_5, v_7)$ which has a memory profile of $(0, 10, 1, 5, 9, 6, 3)$ is a peak memory minimizing schedule for $G_1$. However, $\sigma_1^* = (v_2, v_3, v_4, v_5, v_6, v_7)$ with a memory profile of $(0, 10, 1, 5, 2, 6, 3)$ is an even better schedule for $G_1$. To see why, we now zoom back out to the larger graph $G = (V, E)$.
  
  Clearly, any schedule for $G$ must begin with $v_1$, but now $v_8$ and $v_9$ can be interleaved anywhere in the subschedule for $G_1$. 
  Using the better schedule $\sigma_1^*$ for $G_1$, we can extend it into the schedule $\sigma^* = (v_1, v_2, v_3, v_8, v_9, v_4, v_6, v_5, v_7)$ that has a memory profile of $(0, 10, 1, 10, 5, 9, 6, 10, 7)$ and maintains the peak memory of ten. However, there is no way to extend the schedule $\sigma_1$ to maintain the peak memory of 10.
  To see why, observe that $v_8, v_9$ might as well be inserted together, so we can benefit from the $-5$ as soon as possible after incurring the $+9$. However, this pair cannot be inserted before $v_2$, otherwise $v_2$ will create a sum of $14$. It cannot be inserted between $v_3$ and $v_4$ since then (the improperly scheduled) $v_4, v_6$ will create a sum of $13$. Finally, it cannot be inserted at the end, since then $v_8$ will create a sum of $12$.
\end{example}

\Cref{ex:not-ideal} demonstrates a case where it did not suffice just to get the peak memory as small as possible; the rest of the schedule can matter too when extending to the surrounding graph. We will now present a formal definition of one schedule dominating another which is able to guarantee that the dominant schedule will be better when extending to a surrounding graph. Of course, we will also need to impose some conditions on the subgraph being ``isolated'' (formally defined later) from the surrounding graph; otherwise, we could pick any two nodes and deduce which comes first in the optimal schedule. A similar relation was defined by Kayaaslan et al.\ \cite{kayaaslan2018scheduling} and Liu et al.\ \cite{liu1987application} when dealing with memory minimization in the edge-weighted model for series-parallel graphs and trees respectively. In a sense, our definition can be seen as the two-sided generalization (both forward and backward from the minimum memory usage) of their one-sided definition (only forward from the minimum memory usage).

Actually, we present two definitions of dominance and prove them equivalent. Our definitions of dominance will be written in terms of sequences of numbers, but we extend this to two schedules $\sigma$ and $\tau$ by saying $\sigma$ dominates $\tau$ (denoted $\sigma \ledot \tau$) if and only if $\memprofile{\sigma}$ dominates $\memprofile{\tau}$. Finally, we say that a schedule $\sigma$ is \emph{dominant} for a DAG $G$ if and only if $\sigma$ dominates all feasible schedules of $G$.

This first definition makes the connection to extending subschedules clearer. Intuitively, it says that one can maintain two pointers initially at the beginning of each memory profile. At each step, we are allowed to move one of the pointers forward until both pointers arrive at the end of their sequences. Then one memory profile dominates the other if and only if we can maintain at all times, its pointer refers to a value that is at most the value referred to by the other pointer. When applied to extending subschedules, this sequence of pointers inform us how to appropriately replace a dominated schedule with appropriately staggered pieces from the dominant schedule. We sometimes refer to this definition in later proofs as the ``pointer advancement'' argument.

\begin{definition}[Dominance (Pointer Advancement)]
\label{defn:pointerdef}
  Let $A$ and $B$ denote two arbitrary sequences of numbers.
  Sequence $A$ dominates sequence $B$ (denoted $A \ledot B$) if and only if there exists a sequence of pairs of indices $(a_1, b_1), ..., (a_{|A| + |B| - 1}, b_{|A| + |B| - 1})$ such that:
  \begin{itemize}[parsep=0pt,partopsep=0pt]
    \item $(a_1, b_1) = (1, 1)$,
    \item $(a_{|A| + |B| - 1}, b_{|A| + |B| - 1}) = (|A|, |B|)$,
    \item For each $i>1$, either (i) $a_i = a_{i-1} + 1$ and $b_i = b_{i-1}$, or (ii) $a_i = a_{i-1}$ and $b_i = b_{i-1}+1$,
    \item $A[a_i] \le B[b_i] \quad \forall i$.
  \end{itemize}
\end{definition}

The second definition of dominance that we present is more algebraic and hence useful for deducing properties of the dominance relation. Intuitively, it describes a transformation where we take each memory profile and generate a sequence of equal length but with four-dimensional vectors whose coordinates correspond to the maximum/minimum looking forward/backward from that point. Each vector generated from the dominated memory profile must be coordinate-wise dominated by a vector generated from the dominant memory profile.

\begin{definition}[Dominance (Algebraic)]
\label{defn:algdef}
 We say that sequence $A$ \emph{dominates} sequence $B$ (denoted $A\ledot B$) if and only if the following is true.
For all $i \in [|A|]$, there exists a $j \in [|B|]$ such that:
  \begin{itemize}[parsep=0pt,partopsep=0pt]
    \item $\max_{i' \ge i} A[i'] \le \max_{j' \ge j} B[j']$
    \item $\min_{i' \ge i} A[i'] \le \min_{j' \ge j} B[j']$
    \item $\max_{i' \le i} A[i'] \le \max_{j' \le j} B[j']$
    \item $\min_{i' \le i} A[i'] \le \min_{j' \le j} B[j']$ 
  \end{itemize}
\end{definition}

As alluded to by their names, the two definitions are equivalent, but we will defer the proof of \cref{defn:cost-seq-superior} to \cref{sec:missing-proofs}.

\begin{restatable}{lemma}{pointerlemma}
  \label{defn:cost-seq-superior}
  Definitions \ref{defn:pointerdef} and \ref{defn:algdef} are equivalent.
\end{restatable}

Dominance is a very strong property, and it can easily be observed from either definition that if $\sigma \ledot \tau$, then $\sigma$ enjoys weakly lower peak memory consumption than $\tau$ (however the converse is not true, as \cref{ex:not-ideal} shows). In particular, it is not a priori clear that a dominant schedule always exists for any graph. One of our main technical contributions is to demonstrate the existence of a dominant schedule for any graph.
The following statements are easy consequences of Definitions \ref{defn:pointerdef} and \ref{defn:algdef} respectively.

\begin{lemma}[Concatenation preserves dominance]
\label{lem:concatendominance}
If $A_1 \ledot B_1$ and $A_2 \ledot B_2$, then $A_1 \circ A_2 \ledot B_1 \circ B_2$ where $\circ$ denotes concatenation.
\end{lemma}

\begin{lemma}[Dominance is reflexive and transitive]
\label{lem:superiority-reflexive-transitive}
  Let $A,B,C$ be sequences. Then $A\ledot A$. If $A\ledot B$ and $B\ledot C$, then $A\ledot C$.
\end{lemma} 
\section{Every Graph has a Dominant Schedule}
\label{sec:dominantexists}
In this section, we show the following theorem, which is one of our main technical contributions. It is proved in subsections \ref{subsec:dominance-prelim} -- \ref{subsec:remove-nonneg}. Our proof is constructive, i.e., we show how to construct a dominant schedule by making at most $O(n)$ calls to an oracle that returns a peak memory minimizing schedule. In subsection \ref{subsec:reduction}, we formalize a notion of \emph{hereditary} graphs to show that if one can find a peak memory minimizing schedule on these graphs, then one can also find a dominant schedule efficiently.

\existsdominant

\subsection{Preliminaries}
\label{subsec:dominance-prelim}
We give a few definitions, introduce some additional notation and state a couple of simple observations that will be used throughout the proof.

\begin{definition}[Topological Cut]
\label{defn:topcut}
  Given a DAG $G = (V, E)$, we say that $(S, T)$ is a {\em topological cut} if $S, T$ partitions the nodes of $V$ so that no edge goes from $T$ to $S$, i.e. there is no $(u,v) \in E$ such that $u\in T$ and $v\in S$. 
  Because $T = V \setminus S$ is implied by $S$, we will often refer to a topological cut $(S,T)$ simply using $S$, its ``lefthand'' side.
\end{definition}

  A schedule $\sigma$ of a DAG induces a sequence of topological cuts, $\set(\sigma_{1:0})$, $\set(\sigma_{1:1})$, \ldots, $\set(\sigma_{1:n})$.
        First we have the following simple observation.
\begin{lemma}
    \label{lemma:capcuptopo}
If $X, Y$ are topological cuts of a DAG $G$, then $X \cap  Y$ and $X \cup Y$ are both topological cuts of $G$.
\end{lemma}
\begin{proof}
Let $v$ be an arbitrary node in $X \cap Y$. Let $\delta^-(v)$ be the set of in-neighbors of $v$. Since $X$ and $Y$ are both topological cuts, we have $\delta^-(v) \subseteq X \cap Y$, and hence $X \cap Y$ is also a topological cut. Similarly, let $u$ be an arbitrary node in $X \cup Y$. Then we have $\delta^-(u) \subseteq X$ or $\delta^-(u) \subseteq Y$. Consequently, $\delta^-(u) \subseteq X \cup Y$ and hence $X \cup Y$ is also a topological cut.
\end{proof}

Our proof works with partial schedules, which are simply prefixes of schedules.
\begin{definition}[Partial schedules]
   In a DAG $G=(V,E)$, for a sequence $\sigma$ 
   of distinct nodes $\sigma_1, \sigma_2, \dots, \sigma_k$, we say the sets $S_i = \{\sigma_1, \sigma_2, \dots, \sigma_i\}$ (where $0\le i\le k$) are \emph{induced by}  the sequence $\sigma$.
   We say the sequence $\sigma$ is 
   a \emph{partial schedule} (or sometimes, a valid partial schedule, for emphasis), if for every $0\le i\le k$, the induced set $S_i$ is a topological cut of $G$. 
\end{definition}

We need the following operations on partial schedules. 
Although we use symbols $\cupdot$ and $\capdot$ in analogy to union and intersection for topological cuts, we remark these operations are not commutative, i.e., it may hold that $\sigma \capdot \tau \neq \tau\capdot \sigma$, and likewise, $\sigma \cupdot \tau \neq \tau \cupdot \sigma$. 
\begin{definition}[$\sigma \protect\cupdot \tau$ and $\sigma\protect\capdot \tau$]
Let $\sigma= (\sigma_1, \sigma_2, \dots, \sigma_p)$ be a sequence of distinct nodes.
\begin{itemize}
    \item Given another sequence $\tau = (\tau_1,\tau_2,\dots,\tau_q)$ of distinct nodes, remove from $\tau$ all nodes appearing in $\sigma$, and denote the subsequence of the remaining nodes in $\tau$ (with order unchanged) by $\tau'$. Then, $\sigma\cupdot \tau$ is defined as the concatenation $\sigma \circ \tau'$.
    
    Note that $\sigma\cupdot \tau$ contains distinct nodes.
    \item Given a set $T = \{\tau_1,\tau_2,\dots,\tau_q\}$, remove from $\sigma$ all nodes not appearing in $T$, and the subsequence of the remaining nodes in $\sigma$ (with order unchanged) is denoted by $\sigma\capdot T$.
    
    We extend this definition to sequence $\tau = (\tau_1,\tau_2,\dots,\tau_q)$, and write $\sigma \capdot \tau$ as a shorthand for $\sigma \capdot \set(\tau)$.
\end{itemize}
\end{definition}

\begin{lemma} Let $\sigma$ be a partial schedule of a DAG $G$.
\label{lem:isvalidpartialschedule}
\begin{itemize}
    \item Let $T$ be a topological cut of $G$. Then $\sigma \capdot T$ is a valid partial schedule of $G$.
    \item Let $\tau$ be another partial schedule of $G$. Then $\sigma \cupdot \tau$ is a valid partial schedule of $G$.
\end{itemize}
    \end{lemma}
    \begin{proof}
        Let $\sigma = (\sigma_1, \sigma_2, \dots, \sigma_p)$, and let the induced sets be  $S_i = \{\sigma_1, \sigma_2, \dots, \sigma_i\}$ for each $i$.
        
        Notice that the sets induced by the sequence $\sigma \capdot T$ correspond precisely to (possibly with repetition) $S_1 \cap T, S_2 \cap T, \dots, S_p \cap T$, which are all topological cuts, by \cref{lemma:capcuptopo}.
         So $\sigma \capdot T$ is a valid partial schedule.

For the second statement, let  $\tau = (\tau_1, \tau_2, \dots, \tau_q)$,  and $T_i = \{\tau_1, \tau_2, \dots, \tau_i\}$.
        Likewise, consider the sets induced by $S \cupdot T$. These are $S_1, S_2, \dots, S_p$, followed by $S \cup T_1, S \cup T_2,  \dots, S \cup T_q$. Again, by \cref{lemma:capcuptopo}, these are all topological cuts.
         So $\sigma \cupdot \tau$ is a valid partial schedule.
    \end{proof}

    \begin{definition}[Subschedules]
        For two partial schedules $\sigma,\tau$ in a DAG $G$, we say  $\sigma$ is a \emph{subschedule} of $\tau$, if 
        every node in $\sigma$ appears in $\tau$, and the ordering of those nodes agrees with their ordering in $\tau$.
        In that case, $\sigma=\sigma\capdot \tau = \tau \capdot \sigma$. 
    \end{definition}

We note that the definition of dominance (\cref{defn:algdef}) naturally carries over to partial schedules as well. For two partial schedules $\sigma,\tau$ in a DAG $G$,  we write $\sigma \ledot \tau$ if $\sigma$ dominates $\tau$.

\subsection{Constructing a Dominant Schedule}

In this section, we show that a dominant schedule always exists, in a constructive way: we give an algorithm that constructs a schedule for a DAG $G$, and prove that this schedule indeed dominates all possible schedules of $G$.
However, this algorithm does not necessarily have a polynomial time implementation: one key step of the algorithm asks for a   schedule that minimizes the peak (but is not necessarily dominant), which we know must exist but do not know how to compute efficiently.
For some special graph classes, there exists an efficient implementation; we will discuss this matter further in \cref{subsec:reduction}.

For now, we assume that all valid topological cuts of the graph $G$ have nonnegative total weight. This assumption will be relaxed later in \cref{subsec:remove-nonneg}, using an argument similar to \cite{kayaaslan2018scheduling}. We refer to this assumption as \emph{cut-nonnegativity}.

\begin{definition}[Cut-nonnegativity]
We say a DAG $G=(V,E)$ with node weights $w_v$ is \emph{cut-nonnegative} if every topological cut $S\subseteq V$ has nonnegative cost, $\nodesumcost{S}\ge 0$.
\end{definition}

Our procedure $\textsc{FindSchedule}(G)$ that finds a dominant schedule for a cut-nonnegative DAG $G$ is given in \cref{alg:awesomizer}.
\begin{algorithm}
\DontPrintSemicolon
\caption{Constructing a dominant schedule for cut-nonnegative DAGs in the node sum model}
\label{alg:awesomizer}
$\textsc{FindSchedule}(G)$: \tcp*[f]{$G=(V,E)$ is a cut-nonnegative DAG} \\
\Begin{
\lIf{$w_v=0$ for all $v\in V$}{\Return{an arbitrary schedule} \label{line:trivial}\tcp*[f]{trivial degenerate case}}
$\alpha :=$ a schedule of $G$ with minimum $\peak{\alpha}$ \label{line:defnalpha} \tcp*[r]{a peak-minimizing schedule}
$\hat c:= \peak{\alpha}$\label{line:definecost}\tcp*[r]{peak cost}
$h :=$ min $h\in \{0,1,\dots,|V|\}$ such that $\nodesumcost{A_h} = \hat c $\label{line:defn-i}\tcp*[r]{the leftmost peak of $\alpha$}
$L := $ a topological cut of $G$ containing $A_h$ with minimum $\nodesumcost{L}$ \label{line:defn-mincut-L}\\
$\lambda := \alpha \capdot  L$ \label{line:lambda}\tcp*[r]{arrange the nodes of $L$ in the same order as $\alpha$} 
\Return{$\sigma:= \lambda \circ \textsc{FindSchedule}(G[V\setminus L])$} \label{line:recursion}\tcp*[r]{$\circ$ denotes concatenation}
}
\end{algorithm}

Our goal is to prove the following main lemma about the correctness of  $\textsc{FindSchedule}(G)$.
\begin{lemma}
Let $G$ be a cut-nonnegative DAG in the node sum model.
Then, $\textsc{FindSchedule}(G)$ (\cref{alg:awesomizer}) returns a schedule $\sigma$ of $G$ that dominates all schedules of $G$.
\label{lem:awesomizer-correct}
\end{lemma}

To begin with, note that \cref{alg:awesomizer} recursively invokes itself on a subgraph at \cref{line:recursion}. To
show that this is a valid recursion that terminates correctly, we prove that it only recurses on a cut-nonnegative proper subgraph, and returns a valid schedule of $G$.
\begin{lemma}
At \cref{line:recursion} of \cref{alg:awesomizer}, the induced subgraph $G[V\setminus L]$ is strictly smaller than $G$, and is a cut-nonnegative DAG.
Further, the returned sequence $\sigma$ is a valid schedule of $G$.
\end{lemma}
\begin{proof}
Assume $w_v$ is not constant zero for all $v$; otherwise we would return a schedule immediately at \cref{line:trivial}. 
Then, we must have $\peak{\alpha}>0$. Otherwise, by the cut-nonnegativity of $G$, we can only have $\nodesumcost{A_i} = 0$ for all $0\le i\le |V|$, which would imply $w_{\alpha_i} = \nodesumcost{A_i} -\nodesumcost{A_{i-1}} = 0$ for all vertices $\alpha_i$, a contradiction.

In particular, we have $\nodesumcost{A_0} = 0 \neq \peak{\alpha} $, and hence $h\neq 0$ (\cref{line:defn-i}).
Note $A_h \subseteq L$ from the definition of $L$ at \cref{line:defn-mincut-L}. Hence,  $|L|\ge h\ge 1$, which means $G[V\setminus L]$ is a proper subgraph of $G$.

Now, suppose for contradiction that $G[V\setminus L]$ is not cut-nonnegative, and hence it has a topological cut $L' \subseteq V\setminus L$ such that $\nodesumcost{L'} <0$. Then, observe that $L\cup L'$ is also a topological cut of $G$ such that $A_h \subseteq L \subseteq L\cup L'$,  and has smaller cost as $\nodesumcost{L\cup L'} = \nodesumcost{L} + \nodesumcost{L'}< \nodesumcost{L}$, contradicting the optimality of $L$ defined at \cref{line:defn-mincut-L}.

To show that $\sigma$ is a valid schedule of $G$, we first note that $\lambda = \alpha \capdot L$ is a valid partial schedule of $G$, due to \cref{lem:isvalidpartialschedule}. Since $L$ is a topological cut, there are no edges directed from $V\setminus L$ to $L$. Hence, by induction, the concatenation $\sigma = \lambda \circ \textsc{FindSchedule}(G[V\setminus L])$ is a schedule of $G$.
\end{proof}

The proof of \cref{lem:awesomizer-correct} is based on the following two lemmas. The symbols $h$ and $L$ in the lemma statements are defined in \cref{line:defn-i} and \cref{line:defn-mincut-L} in \cref{alg:awesomizer} respectively.

\begin{lemma}[Extend common prefix to peak]
\label{lem:extendtopeak}
On a cut-nonnegative DAG $G$, suppose $\sigma = \textsc{FindSchedule}(G)$ does not dominate some schedule $\tau$.
Then, there exists a schedule $\tau'$ such that $\tau'_{1:h}=\sigma_{1:h}$, and $\sigma$ does not dominate $\tau'$.
\end{lemma}

\begin{lemma}[Extend common prefix to valley]
\label{lem:extendtovalley}
On a cut-nonnegative DAG $G$, suppose $\sigma = \textsc{FindSchedule}(G)$ does not dominate some schedule $\tau$ that satisfies $\tau_{1:h}=\sigma_{1:h}$.
Then, there exists a schedule $\tau'$ such that $\tau'_{1:|L|}=\sigma_{1:|L|}$, and $\sigma$ does not dominate $\tau'$.
\end{lemma}

We defer the proofs of lemmas \cref{lem:extendtopeak} and \cref{lem:extendtovalley} to \cref{subsec:prefixextension}.
Next we show how they imply that $\sigma= \textsc{FindSchedule}(G)$ dominates all schedules of $G$ (\cref{lem:awesomizer-correct}).

\begin{proof}[Proof of \cref{lem:awesomizer-correct}]
To prove the correctness of $\textsc{FindSchedule}(G)$, we use
induction on the size of $G$, and hence assume as inductive hypothesis that $\textsc{FindSchedule}(G')$ is correct on all cut-nonnegative DAG $G'$ with strictly fewer nodes than $G$. In particular, 
$\textsc{FindSchedule}(G[V\setminus L])$ at \cref{line:recursion} is correct.
We also assume the input $G$ is not in the trivial case defined at \cref{line:trivial}; otherwise, all schedules have constant $0$ memory profiles, so it is valid to return an arbitrary one.

Suppose for contradiction that $G$ has a schedule $\tau$, such that $\sigma =\textsc{FindSchedule}(G)$ does not dominate  $\tau$. Over all such schedules, choose $\tau$ to have the longest prefix in common with $\sigma$, and let $\sigma_{1:k}=\tau_{1:k}$ be that common prefix. 

Recall from \cref{line:recursion} that $\sigma = \lambda \circ \textsc{FindSchedule}(G[V\setminus L])$, where $\set(\lambda)=L$.  We consider the following three cases.
\begin{itemize}
\item \textbf{Case $k\ge |L|$:}

 Then, $\set(\tau_{1:|L|}) = \set(\sigma_{1:|L|}) =  L $, and hence the suffix $\tau_{|L|+1:|V|}$ is a schedule for $G[V\setminus L]$. By inductive hypothesis, $\sigma_{|L|+1:|V|} = \textsc{FindSchedule}(G[V\setminus L])$ is correct, and hence dominates $\tau_{|L|+1: |V|}$.
By prepending a common prefix $\tau_{1:|L|} = \sigma_{1:|L|}$, this implies $\sigma$ dominates $\tau$ (by \cref{lem:concatendominance} and \cref{lem:superiority-reflexive-transitive}), a contradiction. 
    \item \textbf{Case $h\le k< |L|$:}
    
    Then, $\tau_{1:h}=\sigma_{1:h}$, and we can invoke \cref{lem:extendtovalley} to obtain another schedule $\tau'$ that shares a strictly longer prefix $\tau'_{1:|L|} = \sigma_{1:|L|}$, such that $\sigma$ does not dominate $\tau'$. This contradicts the definition of $\tau$.
    \item \textbf{Case $0\le k< h$:}
    
     We can invoke \cref{lem:extendtopeak} to obtain another schedule $\tau'$ that shares a strictly longer prefix $\tau'_{1:h} = \sigma_{1:h}$, such that $\sigma$ does not dominate $\tau'$. This contradicts the definition of $\tau$.
\end{itemize}

We have reached contradictions in both cases. Hence, $\sigma$ dominates all schedules of $G$.
\end{proof}

\subsection{Removing the Cut-Nonnegativity Assumption}
\label{subsec:remove-nonneg}

Finally, we explain how to relax the cut-nonnegativity assumption in \cref{lem:awesomizer-correct}, using an argument similar to the earlier work \cite{kayaaslan2018scheduling}.
In the general case, we break the given DAG $G$ into two parts using a minimum topological cut. Then, the right part and the (reversed and negated) left part are both cut-nonnegative, and can be scheduled using \cref{lem:awesomizer-correct}. Finally we concatenate them into a schedule of $G$. The pseudocode is given in \cref{alg:awesomizergeneral}.
\begin{algorithm}
\DontPrintSemicolon
\caption{Constructing a dominant schedule for general DAGs in the node sum model}
\label{alg:awesomizergeneral}
$\textsc{FindScheduleGeneral}(G)$: \tcp*[f]{$G=(V,E)$ is a DAG} \\
\Begin{
$L:=$ a minimum topological cut of $G$\\
$R:= V\setminus L$\\
$\rho := \textsc{FindSchedule}(G[R])$\\
$G^{\mathrm{rev}}:= \textsc{NegateCostAndReverseEdges}(G)$ \tcp*[r]{reverse the direction of edges in $G$, and negate all node costs $w_v$}
$\lambda:= \textsc{Reverse}(\textsc{FindSchedule}(G^{\mathrm{rev}}[L]))$\\
\Return{$\lambda \circ \rho$}
}
\end{algorithm}

To prove the correctness of \cref{alg:awesomizergeneral}, a key claim is the following (proved in \cref{subsec:mincut}).
\begin{restatable}[A dominant schedule passes through min-cut]{lemma}{mincutlemma}
Let $\tau$ be any schedule for DAG $G$, and let $S$ be a minimum topological cut of $G$.  Define partial schedule $\sigma = \tau \capdot S$. Then, $\sigma \cupdot \tau \ledot \tau$.
\label{lem:mintopodominates}
\end{restatable}

The following statement implies our main \cref{thm:existsdominant}.
\begin{theorem}
Let $G$ be a DAG in the node sum model.
Then, $\textsc{FindScheduleGeneral}(G)$ (\cref{alg:awesomizergeneral}) returns a schedule $\sigma$ of $G$ that dominates all schedules of $G$.
\label{thm:awesomizer-correct-bothdirections}
\end{theorem}
\begin{proof}
We first show that $G[R]$ and $G^{\mathrm{rev}}[L]$ are indeed cut-nonnegative.  
\begin{itemize}
    \item Suppose for contradiction that $G[R]$ has a topological cut $S\subseteq R$ with $\nodesumcost{S}<0$. Then, $L\cup S$ is a topological cut of $G$ with $\nodesumcost{L\cup S} = \nodesumcost{L} + \nodesumcost{S}<\nodesumcost{L}$, contradicting the fact that $L$ is a minimum topological cut of $G$. Hence $G[R]$ is cut-nonnegative.
    \item Similarly, suppose for contradiction that $G^{\mathrm{rev}}[L]$ has a topological cut $T\subseteq L$ with negative cost. Since we negated the node costs in $G^{\mathrm{rev}}$, the total cost of $T$ in the original graph $G$ is positive, $\nodesumcost{T}>0$. Since we also reversed the edge directions in $G^{\mathrm{rev}}$, we know that $L\setminus T$ is a topological cut in $G$, and has total cost $\nodesumcost{L\setminus T} = \nodesumcost{L} - \nodesumcost{T} < \nodesumcost{L}$, a contradiction. Hence $G^{\mathrm{rev}}[L]$ is cut-nonnegative.
\end{itemize} 
Hence, by \cref{lem:awesomizer-correct}, $\lambda= \textsc{Reverse}(\textsc{FindSchedule}(G^{\mathrm{rev}}[L]))$ is a dominant schedule of $G[L]$, and $\rho = \textsc{FindSchedule}(G[R])$ is a dominant schedule of $G[R]$. To see the former claim,  simply notice that for $0\le i\le |L|$, $\nodesumcost{L_i}  = \nodesumcost{L} - \sum_{j=i+1}^{|L|} w(\lambda_j)= \nodesumcost{L} + \sum_{j=1}^{|L|-i} (-w(\textsc{Reverse}(\lambda)_j))$.

It remains to  show $\lambda\circ \rho \ledot\tau$  for any given schedule $\tau$ of $G$. By  \cref{lem:mintopodominates}, we can without loss of generality assume $\tau$ passes through the minimum topological cut $L$, i.e., $\set(\tau_{1:|L|}) = L$ and $\set(\tau_{|L|+1:|V|}) = R$. Then, we use the following pointer-advancing argument to show $\lambda\circ \rho \ledot\tau$: by $\lambda \ledot \tau_{1:|L|}$, we can move both pointers to $L = \set(\tau_{1:|L|})$, and they have the same cost at this point. Then, by $\rho \ledot \tau_{|L|+1:|V|}$, we can continue to move both  pointers to the end.
\end{proof}

\subsection{Finding a Dominant Schedule Efficiently}
\label{subsec:reduction}
We have presented our algorithm (\cref{alg:awesomizer} and \cref{alg:awesomizergeneral}) for constructing a dominant schedule for a DAG. 
For general DAGs, this algorithm may not have a polynomial-time implementation, due to the bottleneck of finding a peak-minimizing schedule at \cref{line:defnalpha} of \cref{alg:awesomizer}. However, for some special graph classes of interest, we can design such an efficient peak-minimizing algorithm, and then it will imply an efficient algorithm for finding a dominant schedule. In the following we describe this reduction more formally.

We restrict our attention to \emph{single-source single-sink DAGs} $G=(V,E)$, where there exists a unique source node $s\in V$ such that all $v\in V$ are reachable from $s$, and a unique sink node $t\in V$ such that all $v\in V$ can reach $t$ (we allow $s=t$, in which case $|V|=1$).
Every DAG $G$ can be transformed to a single-source single-sink DAG, by the following operations:
\begin{definition}[Adding a source/sink]
   Given a DAG $G=(V,E)$, define another DAG $G' = \textsc{AddSource}(G)$ as follows:
   \begin{itemize}
    \item Let $G'=(V',E')$, where $V' = \{s\} \cup V$ contains original nodes and an additional source $s$, and \[E' = E \cup \bigcup_{v\in V: \delta^-(v)=\varnothing}\{(s,v)\}\]
    consists of original edges and additional edges from the new source to every old source in $G$.
   \end{itemize}

We symmetrically define $\textsc{AddSink}(G)$, which adds a sink node $t$ to $G$.
\end{definition}

\begin{definition}[Hereditary class]
    \label{defn:hereditary}
A class $\mathcal{G}$ of single-source single-sink DAGs is called \emph{right-hereditary} (resp.\ \emph{left-hereditary}), if for every graph $G=(V,E)$ in $\mathcal{G}$ and every  topological cut $U\subseteq V$ of $G$,  the DAG $G'=\textsc{AddSource}(G[V\setminus U])$  (resp.\ $G'=\textsc{AddSink}(G[U])$) also belongs to $\mathcal{G}$.

We say $\mathcal{G}$ is hereditary if $\mathcal{G}$ is both left- and right-hereditary.
\end{definition}

Now our reduction (in the node sum memory model) can be summarized in the following theorem.

\begin{restatable}{theorem}{reductiontopeakminimization}
\label{thm:reductiontopeakminimization}
For a hereditary class $\mathcal{G}$ of single-source single-sink DAGs,  suppose there is an algorithm $A$ that, given a DAG $G\in \mathcal{G}$ on $n$ nodes and $m$ edges in the node sum model, finds a schedule $\sigma$ of $G$ that minimizes peak memory $\nodesumcost{\sigma}$ in $T(n,m)$ time.

Then, there is another algorithm that can find a dominant schedule of $G\in \mathcal{G}$ in the node sum model, in $O(n\cdot (T(n+1,m+n) + m^{1+o(1)}))$ time.
\end{restatable}
\begin{proof}
   Given $G\in \mathcal{G}$, we first invoke $\textsc{FindScheduleGeneral}(G)$ in \cref{alg:awesomizergeneral}, which uses the minimum topological cut of $G$  to decompose the problem.
  Here, the minimum topological cut can be computed in $m^{1+o(1)}$ time by simply reducing to the minimum directed $s't'$-cut problem (which can be solved by computing max-flow \cite{maxflow}) with added source $s'$ and sink $t'$: for every $v\in V(G)$ we encode its weight $w_v$ by connecting an edge from $v$ to $t'$ with capacity $w_v$ if $w_v\ge 0$, or connecting an edge from $s'$ to $v$ with capacity $-w_v$ if $w_v<0$, and for every $(u,v)\in E(G)$ we connect an edge from $u$ to $v$ with capacity $+\infty$ to force the topological requirement.
  
   The found minimum topological cut in \cref{alg:awesomizergeneral} then decomposes the original problem into two subproblems, $\textsc{FindSchedule}(G[R])$ and $\textsc{FindSchedule}(G^{\mathrm{rev}}[L])$. We focus on the former one (the latter one can be analyzed symmetrically).

   Since $\mathcal{G}$ is right-hereditary, we have that $\textsc{AddSource}(G[R]) \in \mathcal{G}$. In the recursive algorithm $\textsc{FindSchedule}(G[R])$ (\cref{alg:awesomizer}),  we maintain the invariant that the input $G'$ to $\textsc{FindSchedule}(\cdot )$ is always the induced subgraph of the right-part of some topological cut of the original input graph $G$, and hence satisfies 
$G'\in \mathcal{G}$.
   At \cref{line:defnalpha}  we can find a peak-minimizing schedule of $G'$ as follows: invoke algorithm $A$ on $\textsc{AddSource}(G')$ which is a graph with at most $n+1$ nodes and $m+n$ edges (setting the added source to have weight $0$), and remove the added source from the computed schedule. The correctness is obvious since the added source has to be the first one in any schedule of $\textsc{AddSource}(G')$. \cref{line:defn-mincut-L} of \cref{alg:awesomizer} can also be solved in $m^{1+o(1)}$ time \cite{maxflow}.  The other steps can be implemented in linear time. 

   Since the recursion terminates after at most $n$ steps, the total time complexity of\\ $\textsc{FindScheduleGeneral}(G)$ is $O(n\cdot (T(n+1,m+n) + m^{1+o(1)}))$.
\end{proof} 
\section{Application: Computation Graphs}
\label{subsec:computation_graphs}

Despite the idealism of the node sum model we have been studying so far, it is still powerful enough to capture other models of computation. In this section, we present a more realistic model for computation graphs and provide a reduction to node sum that allows our results to transfer. We also note that it generalizes one-shot black pebbling.

A {\em computation graph} is a triple $\fancyg = (G, \sizefn, \scratchfn)$ where $G = (V,E)$ is an (unweighted) DAG and $\sizefn, \scratchfn$ are both functions mapping $V \rightarrow \mathbb{R}$. Each vertex $v \in V$ is meant to denote an operation, and an edge $(u, v) \in E$ denotes that operation $v$ uses the output of operation $u$. For simplicity (and consistent with real-life implementations), we have assumed that each operation produces an indivisible output.\footnote{Although this assumption can be relaxed, it significantly increases the complexity of the needed structures without fundamentally changing the expressive power of the model.} The function $\sizefn$ maps vertices to the size of their output, whereas $\scratchfn$ maps vertices to additional working memory (``scratch'' memory) needed while running. We constrain $\sizefn(v)$ to be nonnegative. On the other hand, we will permit $\scratchfn(v)$ to be negative, but only to a limited degree: 
\begin{equation}
\label{eqn:scratchcondition}
  \scratch(v) \ge \max \left\{
    -\sizefn(v), \sum_{u \in V: \delta^+(u)= \{v\}} -\sizefn(u)
  \right\}
\end{equation}
The first part of this constraint, $\scratchfn(v) \ge - \sizefn(v)$, means that starting operation $v$ (and hence allocating space for its scratch and output) cannot reduce memory usage. The second part of the constraint, $\scratchfn(v) \ge \sum_{u \in V : \delta^+(u) = \{v\}} -\sizefn(u)$, is meant as a matching condition on the ending of operation $v$. No matter how we choose to schedule the graph, we know that all $u$ whose outputs are only used by $v$ can have their outputs freed after we finish operation $v$. If we pick a schedule which only frees this minimal set of nodes, then ending the operation deallocates the scratch and all of these outputs; the condition says that such an end should not cause memory to increase.

One practical situtation where such a need for negative scratch memory arises is when an operation can safely edit its inputs in-place. For example, if we have a tensor $A$ and want to double it entry-wise to get $2A$, this can be safely done in-place without allocating an extra copy of $A$. However, this is only safe if no other operations want to use $A$, which is the reason we wound up with a condition like $\delta^+(u) = \{v\}$. In this case, operation $u$ generates $A$ and operation $v$ uses $A$ to generate $2A$, and we can model $v$'s in-place construction of $2A$ by setting $\scratchfn(v) = -\sizefn(v)$ (which also equals $-\sizefn(u)$).

There are also some technical reasons why it is convenient for our scratch memory to be negative but not too negative. We utilize negative scratch for the following natural conversion we want to perform: given a memory profile, construct a computation graph which is a path and whose only schedule generates that memory profile. In particular, this is used to let our lower bound constructions focus on memory profiles and it lets us use dominant schedules to replace subgraphs with paths. Unfortunately, without negative scratch, some memory profiles would not permit this conversion. On the other hand, we utilize the fact that scratch is not too negative when reducing computation graphs to node sum graphs, since it lets us argue that some of the generated node sum nodes have negative weights and hence might as well be executed as soon as possible.

Putting all these pieces regarding output sizes, scratch memory, and dependencies together, the memory consumed while executing a schedule behaves (informally) as follows. Suppose we want to run some operation $u$ of our schedule (still a topological ordering of $G$). Each operation $v$ that has been executed prior to $u$ has produced some output that normally occupies $\sizefn(v)$ units of memory. However, not all of these are still in memory; only the ones depended on by future operations (including $u$ and onwards) must still be occupying memory.

In addition to its inputs, the operation $u$ itself must allocate space for its output, which takes $\sizefn(u)$ units of memory. It must also allocate additional scratch memory, taking $\scratchfn(u)$ units of memory. We can then finally run $u$. Afterwards, we might be able to release some of $u$'s input (those that are not used by later operations), as well as its temporary scratch memory. With this sequence of events in mind, we are now ready to formally define how to compute the memory profile of a particular schedule.

\begin{definition}[Computation Graph Memory Model]
\label{def:irl}
Let $\fancyg = (G, \sizefn, \scratchfn)$ be a computation graph, and let $\sigma$ be a schedule for $G$. Let $S_i$ be the set of nodes in $\sigma_{1:i}$, i.e.\ the first $i$ nodes, and let $T_i$ be the set of nodes in $\sigma_{i:n}$. For each $i$, define
\begin{align*}
  \memduring{\sigma, \sigma_i} &\defeq
    \left( \sum_{v \in S_{i-1} : \outneighbors{v} \cap T_i \ne \varnothing} \sz{v} \right)+ \sz{\sigma_i} + \scr{\sigma_i}, \\
  \memafter{\sigma, \sigma_i} &\defeq
    \sum_{v \in S_i : \outneighbors{v} \cap T_{i+1} \ne \varnothing} \sz{v}
  \text{.}
\end{align*}
Note that the above are implicitly functions of the entire computation graph, though we suppress the $\fancyg$ for readability. We sometimes abuse notation and suppress the dependence on $\sigma$ to say $\memduring{\sigma_i} = \memduring{\sigma, \sigma_i}$ when the schedule $\sigma$ is clear from context.

We define the \emph{memory profile} of a schedule $\sigma$ to be the sequence induced by the above:
{\small
  \begin{align*}
    \nodeweightedprofile{\sigma} &\defeq
      \left( 0,
        \memduring{\sigma_1},
        \memafter{\sigma_1},
        \memduring{\sigma_2},
        \memafter{\sigma_2},
        \ldots,
        \memduring{\sigma_n},
        \memafter{\sigma_n}
      \right)
    \text{.}
  \end{align*}
}

We define the \emph{peak memory} to be the maximum memory over the memory profile. Due to \cref{eqn:scratchcondition}, we know that $\memduring{\sigma_i} \ge \memafter{\sigma_i}, \forall i$ and hence peak memory occurs during some operation:
\begin{align*}
  \peak{\sigma} &\defeq \max_i \{ \memduring{\sigma_i} \}
  \text{.}
\end{align*}
\end{definition}

Our formal problem of interest using this model is, as before, to minimize this peak memory:
\begin{definition}[Peak Memory Scheduling Problem]
\label{def:peak_memory_scheduling} 
\ \\
Given a computation graph $\fancyg = (G, \sizefn, \scratchfn)$, the {\em Peak Memory Scheduling Problem} is to find a schedule $\sigma$ for $G$ with the smallest possible value of $\peak{\sigma}$.
\end{definition}

\begin{observation}\nonumber
The weighted one-shot black pebbling game is equivalent to the peak memory scheduling problem where the scratch of every node in the computation graph is set to 0.
\end{observation}
In particular, starting to compute a node corresponds to placing a new pebble (resulting in the cost of $\mathrm{mem_{during}}$); finishing a node corresponds to removing pebbles that are no longer needed (resulting in the cost of $\mathrm{mem_{after}}$).

\subsection{Reducing Computation Graphs to Node Sum}
\label{subsec:equivnodesumnw}

We now present a reduction from \truemem\ model to the node sum memory model, which will allow us to leverage our node sum results (i.e., dominance) for the {\em Peak Memory Scheduling Problem} on computation graphs. It is important to note that our reduction is not guaranteed to preserve graph structure. For example, even if our computation graph is series-parallel, its corresponding node sum graph may not be so.

Perhaps the most salient mismatch between the two models is that computation graphs have to consider memory during and after a node, whereas node sum only generates a single memory number for each node. Hence, in order to generate profiles of similar complexity, our reduction will need to split each computation graph node into multiple nodes. At first glance, we need at least two nodes: one to handle ``starting'' the operation by allocating for its scratch and output and one to handle ``finishing'' the operation by deallocating its scratch and any inputs that are no longer needed. However, determining whether an input is no longer needed depends on all consumers of that input, so we also introduce a ``release'' node that can determine when all these consumers finish. We refer to the ``finish'' and ``release'' nodes in this reduction as \emph{virtual} nodes.

There is a small issue with this plan to make a ``start'', ``finish'', and ``release'' node for each operation. If both our sizes and scratch were nonnegative, then we would know that the start node has nonnegative weight and the finish and release nodes have nonpositive weight, so we can argue that any time a schedule executes the start node, it might as well immediately execute the finish node and any available release nodes. However, when the scratch function is negative, then the end node will have positive weight and schedules have incentive to not immediately execute the finish node. To solve this, we take advantage of \cref{eqn:scratchcondition} and observe that we can match its negativity with the release nodes of operations that depend only on this operation. Hence for an operation with only one out-neighbor, we fold its release node into the finish node of its out-neighbor.

\begin{figure}
\centering
\begin{tikzpicture}[%
  auto,
  scale=0.8,
  hollowcluster/.style={
    ellipse,
    draw=black,
    inner sep=1pt,
    minimum height=100pt,
    minimum width=25pt,
  },
  shadednode/.style={
    circle,
    draw=black,
    fill=black!10,
    inner sep=1pt,
    minimum size=25pt,
  },
  ]
  \node[hollowcluster] (u1) at (0, 0) {$\delta^-(v)$};
  \node[shadednode] (v1) at (2, 0) {$v$};

  \draw[->] (u1) -- (v1);
  \draw[->] (u1.north east) -- (v1);
  \draw[->] (u1.south east) -- (v1);

  \draw[dashed] (6, -3) -- (6, 3);

  \node[hollowcluster] (u2) at (8, 0) {$\finishvertex{\left[\delta^-(v)\right]}$};
  \node[shadednode] (v2) at (11, 0) {$\startvertex{v}$};
  \node[shadednode] (v3) at (13, 0) {$\finishvertex{v}$};

  \draw[->] (u2) -- (v2);
  \draw[->] (u2.north east) -- (v2);
  \draw[->] (u2.south east) -- (v2);
  \draw[->] (v2) -- (v3);
\end{tikzpicture}
\caption{When $|\delta^+(v)| = 0$, $v$ gets a start node and a finish node. Its would-be release node gets folded into its finish node. As a result, the start node has weight $w_{\startvertex{v}} = \scratchfn(v) + \sizefn(v)$ and the finish node has weight $w_{\finishvertex{v}} = - \scratchfn(v) - \sizefn(v) - \sum_{u : \delta^+(u) = \{v\}} \sizefn(u)$.}
\label{fig:tikz-reduction0}
\end{figure}
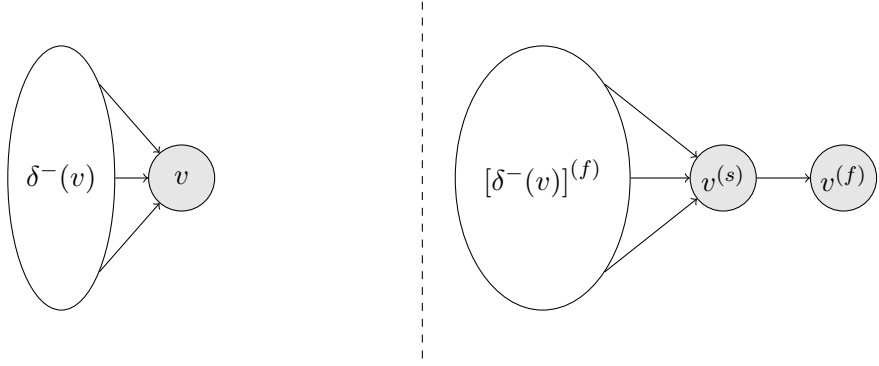
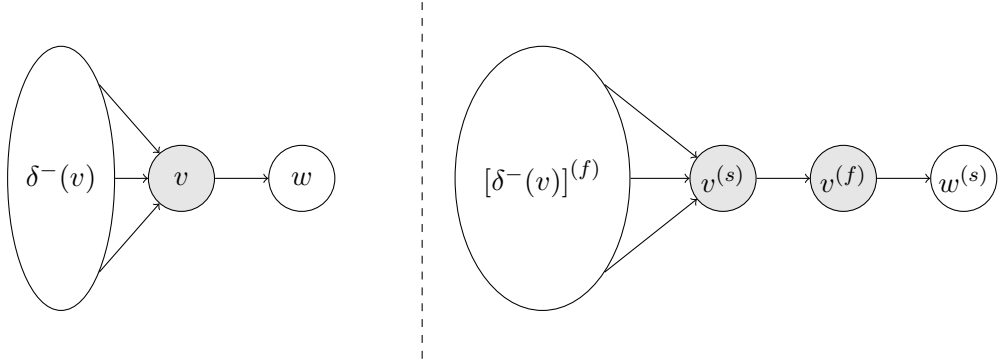
\begin{figure}
\centering
\begin{tikzpicture}[%
  auto,
  scale=0.8,
  hollownode/.style={
    circle,
    draw=black,
    inner sep=1pt,
    minimum size=25pt,
  },
  hollowcluster/.style={
    ellipse,
    draw=black,
    inner sep=1pt,
    minimum height=100pt,
    minimum width=25pt,
  },
  shadednode/.style={
    circle,
    draw=black,
    fill=black!10,
    inner sep=1pt,
    minimum size=25pt,
  },
  ]
  \node[hollowcluster] (u1) at (0, 0) {$\delta^-(v)$};
  \node[shadednode] (v1) at (2, 0) {$v$};
  \node[hollownode] (w1) at (4, 0) {$w$};

  \draw[->] (u1) -- (v1);
  \draw[->] (u1.north east) -- (v1);
  \draw[->] (u1.south east) -- (v1);
  \draw[->] (v1) -- (w1);

  \draw[dashed] (6, -3) -- (6, 3);

  \node[hollowcluster] (u2) at (8, 0) {$\finishvertex{\left[\delta^-(v)\right]}$};
  \node[shadednode] (v2) at (11, 0) {$\startvertex{v}$};
  \node[shadednode] (v3) at (13, 0) {$\finishvertex{v}$};
  \node[hollownode] (w2) at (15, 0) {$\startvertex{w}$};
  
  \draw[->] (u2) -- (v2);
  \draw[->] (u2.north east) -- (v2);
  \draw[->] (u2.south east) -- (v2);
  \draw[->] (v2) -- (v3);
  \draw[->] (v3) -- (w2);
\end{tikzpicture}
\caption{When $|\delta^+(v)| = 1$, $v$ gets a start node and a finish node. Its would-be release node gets folded into the finish node of its unique consumer. As a result, the start node has weight $w_{\startvertex{v}} = \scratchfn(v) + \sizefn(v)$ and the finish node has weight $w_{\finishvertex{v}} = - \scratchfn(v) - \sum_{u : \delta^+(u) = \{v\}} \sizefn(u)$.}
\label{fig:tikz-reduction1}
\end{figure}
\begin{figure}
\centering
\begin{tikzpicture}[%
  auto,
  scale=0.8,
  hollownode/.style={
    circle,
    draw=black,
    inner sep=1pt,
    minimum size=25pt,
  },
  hollowcluster/.style={
    ellipse,
    draw=black,
    inner sep=1pt,
    minimum height=100pt,
    minimum width=25pt,
  },
  shadednode/.style={
    circle,
    draw=black,
    fill=black!10,
    inner sep=1pt,
    minimum size=25pt,
  },
  ]
  \node[hollowcluster] (u1) at (0, 0) {$\delta^-(v)$};
  \node[shadednode] (v1) at (2, 0) {$v$};
  \node[hollowcluster] (w1) at (4, 0) {$\delta^+(v)$};

  \draw[->] (u1) -- (v1);
  \draw[->] (u1.north east) -- (v1);
  \draw[->] (u1.south east) -- (v1);
  \draw[->] (v1) -- (w1);
  \draw[->] (v1) -- (w1.north west);
  \draw[->] (v1) -- (w1.south west);

  \draw[dashed] (6, -3) -- (6, 3);

  \node[hollowcluster] (u2) at (8, 0) {$\finishvertex{\left[\delta^-(v)\right]}$};
  \node[shadednode] (v2) at (11, 0) {$\startvertex{v}$};
  \node[shadednode] (v3) at (13, 0) {$\finishvertex{v}$};
  \node[hollowcluster] (w2) at (16, 0) {$\startvertex{\left[\delta^+(v)\right]}$};
  \node[shadednode] (v4) at (19, 0) {$\releasevertex{v}$};

  \draw[->] (u2) -- (v2);
  \draw[->] (u2.north east) -- (v2);
  \draw[->] (u2.south east) -- (v2);
  \draw[->] (v2) -- (v3);
  \draw[->] (v3) -- (w2);
  \draw[->] (v3) -- (w2.north west);
  \draw[->] (v3) -- (w2.south west);
  \draw[->] (w2) -- (v4);
  \draw[->] (w2.north east) -- (v4);
  \draw[->] (w2.south east) -- (v4);
\end{tikzpicture}
\caption{When $|\delta^+(v)| \ge 2$, $v$ gets a start node, a finish node, and a release node. As a result, the start node has weight $w_{\startvertex{v}} = \scratchfn(v) + \sizefn(v)$, the finish node has weight $w_{\finishvertex{v}} = - \scratchfn(v) - \sum_{u : \delta^+(u) = \{v\}} \sizefn(u)$, and the release node has weight $w_{\releasevertex{v}} = -\sizefn(v)$.}
\label{fig:tikz-reduction2}
\end{figure}
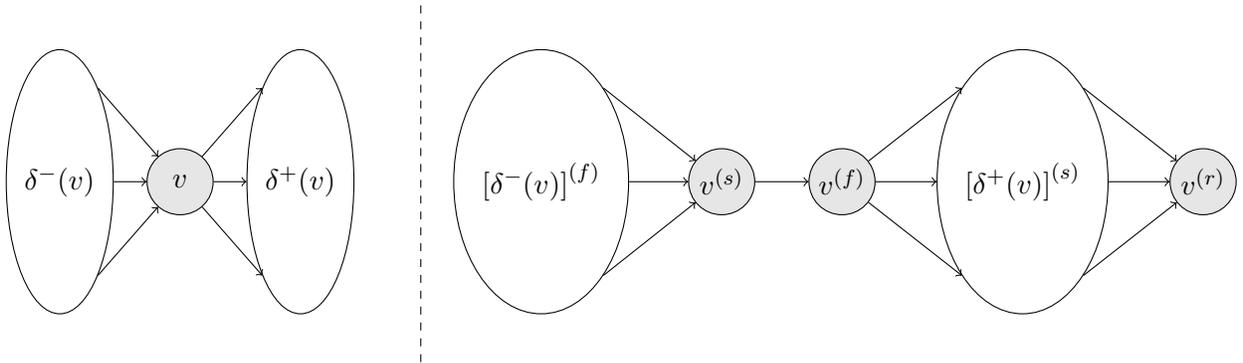

We are now ready to formally present our reduction. Figures \ref{fig:tikz-reduction0}, \ref{fig:tikz-reduction1}, and \ref{fig:tikz-reduction2} illustrate the key cases of this transformation. Given a computation graph $\fancyg = (G, \sizefn, \scratchfn)$, we generate a weighted node sum graph $\nsff{\fancyg}$ using \Cref{alg:reduction}.

For a node sum graph, define an \emph{eager} schedule as one in which each node with nonpositive weight is scheduled immediately after its last in-neighbor. (If multiple nonpositive nodes share the same last in-neighbor, they can be arranged in arbitrary order immediately following it.) To analyze the reduction, we start with two auxiliary lemmas, and then present the main one (\cref{lem:model-reduction}).

\begin{algorithm}
  \caption{Reduction $\nsff{\fancyg}$} \label{alg:reduction}
  \KwData{Computation graph $\fancyg = (G, \sizefn, \scratchfn)$}
  \KwResult{Weighted node sum graph $\ns{G} = (\ns{V}, \ns{E}, \ns{w})$}
  Initialize $\ns{V} \leftarrow \emptyset, \ns{E} \leftarrow \emptyset$\;
  \For{computation graph vertex $v \in V$} {
    Add vertices $\ns{V} \leftarrow \ns{V} \cup \{\startvertex{v}, \finishvertex{v}\}$\;
    Set weight $w_{\startvertex{v}} \leftarrow \scratchfn(v) + \sizefn(v)$\;
    Set weight $w_{\finishvertex{v}} \leftarrow -\scratchfn(v) - \sum_{u : \delta^+(u) = \{v\}} \sizefn(u)$\; \label{line:virtual1}
    Add edge $\ns{E} \leftarrow \ns{E} \cup \{(\startvertex{v}, \finishvertex{v})\}$\;
    \For{in-neighbor $u \in \delta^-(v)$} {
      Add edge $\ns{E} \leftarrow \ns{E} \cup \{(\finishvertex{u}, \startvertex{v})\}$\;
    }
    \Switch{$|\delta^+(v)|$} {
      \Case{$|\delta^+(v)| = 0$} {
        Update weight $w_{\finishvertex{v}} \leftarrow -\scratchfn(v) - \sizefn(v) - \sum_{u : \delta^+(u) = \{v\}} \sizefn(u)$\;  \label{line:virtual2}
      }
      \Case{$|\delta^+(v)| = 1$} {
        Perform no additional updates for this case\;
      }
      \Case{$|\delta^+(v)| \ge 2$} {
        Add additional vertex $\ns{V} \leftarrow \ns{V} \cup \{\releasevertex{v}\}$\;
        Set weight $w_{\releasevertex{v}} \leftarrow -\sizefn(v)$\; \label{line:virtual3}
        \For{out-neighbor $x \in \delta^+(v)$} {
          Add edge $\ns{E} \leftarrow \ns{E} \cup \{(\startvertex{x}, \releasevertex{v})\}$\;
        }
      }
    }
  }
  \KwRet{$(\ns{V}, \ns{E}, \ns{w})$}
\end{algorithm}

\begin{lemma}\label{lem:eager}
    Any schedule for a node sum graph can be transformed into an eager schedule that dominates the original.
\end{lemma}
\begin{proof}
  Since dominance is transitive (Lemma \ref{lem:superiority-reflexive-transitive}), it suffices to show that moving a nonpositive weight node earlier in the schedule (respecting precedence to keep the schedule valid) results in a schedule that dominates the original. In fact, it suffices to move the node a single node earlier in the schedule, since we can just chain such moves to achieve a larger move.
  
  So, for a schedule $\ns{\sigma}$, consider moving a single nonpositive node $\ns{\sigma}_{i+1}$ one node earlier in the schedule. According to the node sum memory model (\Cref{defn:nodesummodel}), the memory usage at time $i$ is just the sum of the weights of the first $i$ nodes in the schedule $\ns{\sigma}$. So the only memory usage term that changes is $\nodesumcost{S_i}$ (which has the term $w_{\ns{\sigma}_i}$ replaced with $w_{\ns{\sigma}_{i+1}}$). We cut out three memory usage terms of interest: $\nodesumcost{S_{i-1}}, \nodesumcost{S_i}, \nodesumcost{S_{i+1}}$.
  Previously, they are equal to $\nodesumcost{S_{i-1}}$, $\nodesumcost{S_{i-1}} + w_{\ns{\sigma}_i}$, $\nodesumcost{S_{i-1}} + w_{\ns{\sigma}_i} + w_{\ns{\sigma}_{i+1}}$. They will now be equal to $\nodesumcost{S_{i-1}}$, $\nodesumcost{S_{i-1}} + w_{\ns{\sigma}_{i+1}}$, $\nodesumcost{S_{i-1}} + w_{\ns{\sigma}_i} + w_{\ns{\sigma}_{i+1}}$. We will show that the latter sequence dominates the former using the pointer advancement definition (\Cref{defn:pointerdef}).

  Both pointers will begin at the start of their respective sequences, which is safe because the sequences start with the same term: $\nodesumcost{S_{i-1}} \le \nodesumcost{S_{i-1}}$. We then increment the pointer for the latter sequence, which is safe because $w_{\ns{\sigma}_{i+1}} \le 0$ implies $\nodesumcost{S_{i-1}} + w_{\ns{\sigma}_{i+1}} \le \nodesumcost{S_{i-1}}$. The next change to both sequences is now the same (adding $w_{\ns{\sigma}_i}$), and we want to increment both pointers. We increment the former sequence's pointer first if this quantity is positive and second otherwise, which guarantees the intermediate comparison works out. Finally, it is safe to increment the pointer for the former sequence because the sequences end with the same term: $\nodesumcost{S_{i-1}} + w_{\ns{\sigma}_i} + w_{\ns{\sigma}_{i+1}} \le \nodesumcost{S_{i-1}} + w_{\ns{\sigma}_i} + w_{\ns{\sigma}_{i+1}}$.

  To finish the proof, we observe that all sequence terms before and after the three terms we just focused on are the same. Hence the prefix dominates by reflexivity as does the suffix (\Cref{lem:superiority-reflexive-transitive}). We can then combine the prefix with the three terms with the suffix by concatenating and invoking \Cref{lem:concatendominance}.
\end{proof}

\begin{lemma}\label{lem:reduction-dominance}
  Suppose we have a computation graph $\fancyg$ and the corresponding node sum graph $\nsff{\fancyg}$ from the reduction in \Cref{alg:reduction}. Any eager schedule of the node sum graph $\nsff{\fancyg}$ both dominates and is dominated by the (computation graph) schedule of $\fancyg$ that orders nodes $v$ according to the order of their corresponding start nodes.
\end{lemma}
\begin{proof}
  We argue this by comparing subsequences of the memory profiles and then concatenating them all. Specifically, we group every start node $\startvertex{v}$ in the node sum schedule with the virtual nodes that run immediately after (due to the eagerness, this is the finish node $\finishvertex{v}$ and all release nodes $\releasevertex{u}$ for which this is the last consumer). Each of these nodes has an associated memory usage in the node sum model, $\nodesumcost{S_i}$, and since these nodes are contiguous, these memory usages are as well. We compare this with the corresponding node $v$ in the original computation graph and its two associated memory costs, $\memduring{\sigma, v}, \memafter{\sigma, v}$. We want to show that each node sum subsequence both dominates and is dominated by its two corresponding computation graph terms. Importantly, the concatenation we perform to arrive at the entire memory profiles agrees on the order because we are considering the unique computation graph schedule that agrees on this order.

  We claim by induction over $v$ in $\sigma$ order that the first term of our node sum memory subsequence is $\memduring{\sigma, v}$ and the last term is $\memafter{\sigma, v}$. The first term matches because $w_{\startvertex{v}} = \scratchfn(v) + \sizefn(v)$ is the correct update to get from the previous $\memafter{\sigma, \sigma_i}$ to the current $\memduring{\sigma, \sigma_{i+1}}$. To get the last term to match, we want to show that all the virtual nodes involved always subtract off the sizes of released nodes as well as $\scratchfn(v)$. Our finish node always includes $\scratchfn(v)$. Any released nodes $u$ that have $|\delta^+(u)| = 0$ (this could only be $v$ itself) are handled by subtracting the size from the weight of the finish node. Any released nodes $u$ that have $|\delta^+(u)| = 1$ are handled by the $\sum_{u : \delta^+(u) = \{v\}} \sizefn(u)$ term in the finish node. Any released nodes $u$ that have $|\delta^+(u)| \ge 1$ are handled by their release node being a virtual node that is in this group.

  Hence the two sequences we want to compare agree on their first and last terms. The node sum subsequence additionally has intermediate terms, but they are at most the first term and at least the last term because we have already shown that virtual nodes have nonpositive weight, and each successive term is the weight of a virtual node plus the previous term. It is now easy to argue dominance in both directions using the pointer advancement definition of dominance (\Cref{defn:pointerdef}). To show that the node sum subsequence is dominant, completely increment its pointer first. To show that it is dominated, increment the computation graph pointer first.
\end{proof}

\begin{lemma}
\label{lem:model-reduction}
  Suppose we have a computation graph $\fancyg$ and the corresponding node sum graph $\nsff{\fancyg}$ from the reduction in \Cref{alg:reduction}. If we have a schedule $\sigma$ for $\fancyg$ then (in linear time) we can find a schedule $\ns{\sigma}$ for $\nsff{\fancyg}$ such that $\nodesumprofile{\ns{\sigma}} \ledot \nodeweightedprofile{\sigma}$ (the memory profiles are computed respective to their relative graphs). Additionally, if we have a schedule $\ns{\sigma}$ for $\nsff{\fancyg}$ then (in linear time) we can find a schedule $\sigma$ for $\fancyg$ such that $\nodeweightedprofile{\sigma} \ledot \nodesumprofile{\ns{\sigma}}$ (again, memory profiles are computed respective to their relative graphs).
\end{lemma}

\begin{proof}
  Given a computation graph schedule $\sigma$ for $\fancyg$, the (eager) node sum schedule $\ns{\sigma}$ for $\nsff{\fancyg}$ is constructed as follows. For each node $v$ in $\sigma$, append $\startvertex{v}$, then $\finishvertex{v}$, then any release nodes all of whose in-neighbors have already been scheduled, to $\ns{\sigma}$.  Conversely, given a schedule $\ns{\sigma}$ for $\nsff{\fancyg}$, we produce the schedule $\sigma$ by arranging the nodes of $\fancyg$ in the order of their start nodes in $\ns{\sigma}$. It can be checked that these schedules respect all the edges of their corresponding graphs. The conversions can be performed in linear time.

  To analyze memory profiles, we first show that virtual nodes are assigned nonpositive weight. To do so, we focus our attention on the lines of the algorithm which assign weight to virtual nodes, namely lines \ref{line:virtual1}, \ref{line:virtual2}, and \ref{line:virtual3}. Regarding line \ref{line:virtual1}, which sets $w_{\finishvertex{v}} \leftarrow -\scratchfn(v) - \sum_{u:\delta^+(u) = \{v\}}\sizefn(u)$, the nonpositivity follows by subtracting $\scratchfn(v)$ from both sides of  \cref{eqn:scratchcondition}. The nonpositivity for line \ref{line:virtual2} follows by observing that this line chooses almost the same quantity as line \ref{line:virtual1}, just less $\sizefn(v)$ which is nonnegative by assumption. Line \ref{line:virtual3} is trivial, because it sets $w_{\releasevertex{v}} \leftarrow -\sizefn(v)$ and again $\sizefn(v)$ is nonnegative by assumption.
  
  Now, given a computation graph schedule $\sigma$, our construction of the node sum one produces an eager schedule $\ns{\sigma}$ with start nodes in the same order, which dominates $\sigma$ by \cref{lem:reduction-dominance}. For the other direction, given a node sum schedule $\ns{\sigma}$ for $\nsff{\fancyg}$, we can analytically convert it to an eager schedule $\ns{\sigma}'$ and then convert $\ns{\sigma}'$ into a computation graph schedule $\sigma$ by restricting it to start nodes. (Algorithmically, we can just directly restrict $\ns{\sigma}$ to start nodes). Then $\sigma$ dominates $\ns{\sigma}'$ by \cref{lem:reduction-dominance} and $\ns{\sigma}'$ dominates $\ns{\sigma}$ by \cref{lem:eager}. Transitivity of dominance (\Cref{lem:superiority-reflexive-transitive}) then gives the result.
\end{proof}
\section{Linearization}
\label{sec:linearization}

When dealing with recursive graph structures such as SP graphs or even trees, it is tempting to first optimally schedule subgraphs recursively and then merge these schedules to obtain the final solution. For example, if a graph $G$ is obtained via a parallel composition of two subgraphs $G_1$ and $G_2$ (formally defined in Definition \ref{def:sp-graph}), then one approach to schedule $G$ is to first find optimal schedules $\sigma$ and $\tau$ for $G_1$ and $G_2$, respectively, and then find an optimal way to merge $\sigma$ and $\tau$ to find the optimal schedule for $G$. Unfortunately, such a strategy fails due to the nuances of the computation graph memory model.
Indeed, it is possible that no optimal schedule for $G$ includes $\sigma$ or $\tau$ as a subsequence. 

\subsection{Linearization}
\label{subsec:linearization}

In this section, we introduce the algorithmic technique termed \emph{linearization}, that allows us to simplify the structure of an input DAG by replacing a subgraph of the DAG by a simple path. The key idea is to find a \emph{dominant} schedule of the subgraph and then add a path graph corresponding to this dominant schedule. We call this process \emph{linearization} and show that it preserves dominant schedules for the input DAG provided that the subgraph being replaced is \emph{isolated} from the rest of the graph. Linearization is a very useful tool since it allows one to effectively design a recursive algorithm for peak memory scheduling (especially for recursive graph structures). It also leads to effective heuristics in practice.

We first define the notion of isolated subgraphs. 
\begin{definition}[Isolated subgraph]
\label{defn:isolated}
Let $G=(V,E)$ be a DAG. For $U\subseteq V$, we say the induced subgraph $G[U]$ is an \emph{isolated subgraph} if and only if
\begin{itemize}
    \item $G[U]$ is a DAG with a unique source $s$, and a unique sink $t$, where $s\neq t$.
    \item Any path in $G$ starting from some $x \in V \setminus U$ and ending at some $y\in U$ must either (i) go through $s$, or (ii) contain $t 
 = y$ as the unique node in $U$. 
    \item Any path in $G$ starting from some $x \in U$ and ending at some $y \in V \setminus U$ must go through $t$.
\end{itemize}
\end{definition}

Linearization replaces an isolated subgraph by a path corresponding to its dominant schedule. However, we need to be a little careful about the definition of dominance here. First, note that since an isolated subgraph has a unique source node $s$ and a unique sink node $t$, all valid schedules must start with $s$ and end with $t$. 

\begin{definition}[Internal Memory Profile]
\label{defn:internalmemoryprofile}
Consider any graph $G'$ that has a single source $s$ and single sink $t$, and let $\sigma = (s = \sigma_1, \sigma_2, \ldots, \sigma_{|U|} = t)$ be any valid schedule for $G'$. Recall that 
{\small
\[\nodeweightedprofile{\sigma} := \left( 0, 
        \mathrm{mem_{during}}(\sigma_1),
        \mathrm{mem_{after}}(\sigma_1),
        \mathrm{mem_{during}}(\sigma_2),
        \mathrm{mem_{after}}(\sigma_2),
        \ldots,
        \mathrm{mem_{during}}(\sigma_n),
        \mathrm{mem_{after}}(\sigma_n)
    \right)\] 
}
denotes the sequence of memory consumptions incurred by schedule $\sigma$.
Then let \[\internalprofile{\sigma} :=  \left( 
        \sz{s} = \mathrm{mem_{after}}(\sigma_1),
        \mathrm{mem_{during}}(\sigma_2),
        \mathrm{mem_{after}}(\sigma_2),
        \ldots,
        \mathrm{mem_{during}}(\sigma_n),
        \sz{t}
    \right)\] be the memory profile corresponding to $\sigma$ that begins after the source $s$ has already been scheduled and ends after $t$ has been scheduled but still held in memory. We refer to $\internalprofile{\sigma}$ as the \emph{internal memory profile} of $\sigma$.
\end{definition}

Now we are ready to state the following exchange lemma.
\begin{restatable}[Exchange lemma]{lemma}{exchangelemma}
\label{lem:exchangelinearize}
Let $G=(V,E)$ be a computation graph with a single source and a single sink, and let $G[U]$ be an isolated subgraph of $G$. 
Let $\pi$ and $\pi'$ be two schedules of $G[U]$, such that $\internalprofile{\pi'}\ledot \internalprofile{\pi}$.  
Then, for any schedule $\sigma$ of $G$ that contains $\pi$ as a subschedule, there exists another schedule $\sigma'$ of $G$ that contains $\pi'$ as a subschedule, such that $\internalprofile{\sigma'}\ledot \internalprofile{\sigma}$.
\end{restatable}

Using \cref{lem:exchangelinearize}, one can show the following linearization result. Figure \ref{fig:linearization} gives an illustration of linearization. 
\begin{figure}[htbp]
    \centering
    \includegraphics[width=0.9\textwidth]{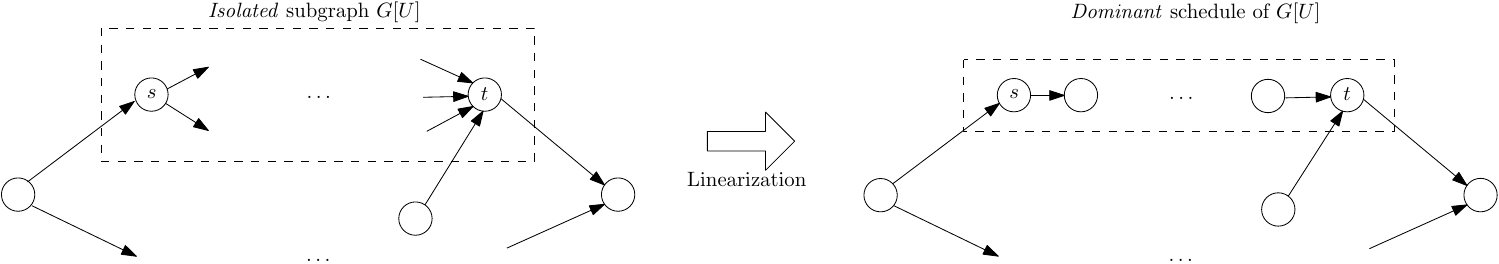}
    \caption{Linearization of an isolated subgraph by its dominant schedule}
    \label{fig:linearization}
\end{figure}

\begin{restatable}[Linearization]{lemma}{linearization}
\label{lem:linearization}
Let $G=(V,E)$ be a computation graph in the \truemem\ model, let $G[U]$ be an isolated subgraph of $G$ and let $\pi$ denote a dominant schedule for $G[U]$. Then we may replace the subgraph $G[U]$ with a path graph on $|U|$ nodes (with suitable node weights) corresponding to the schedule $\pi$, thus transforming $G$ into another (simpler) graph $G'$ such that: 
\begin{itemize}
    \item any valid schedule in $G'$ is also a valid schedule in $G$.
    \item the dominant schedule of $G'$ is a dominant schedule of $G$.
\end{itemize}
\end{restatable}
The proofs of \cref{lem:exchangelinearize} and \cref{lem:linearization} rely on the ``pointer advancement'' definition of dominance and crucially exploit the property 
that one can replace any valid schedule for a subgraph with a dominant schedule while always maintaining a lower memory cost. We defer the full proofs to \cref{sec:linearization-proofs}.
\section{Anatomy of a Dominant Schedule}
\label{sec:segments}

After invoking linearization via \cref{lem:linearization}, we get to replace an isolated subgraph with a directed path. There is only one order to schedule the nodes in this directed path, but how do we interleave these nodes with the surrounding graph? For example, if we had multiple such paths, how would we combine their unique schedules? This question was previously studied by Liu \cite{liu1987application} for in-trees. We generalize the merging strategy in that work to apply to any subschedules, not just subschedules that begin with zero memory usage.

\subsection{Isolated Subgraphs for Node Sum}

We previously defined isolated subgraphs for the \truemem\  model in \cref{defn:isolated}. Since we are now working with the node sum model, the definition can be a bit looser:
\begin{definition}[Isolated Subgraph for Node Sum]
\label{def:isolated-node-sum}
  Let $G=(V,E)$ be a DAG. For $U\subseteq V$, we say the induced subgraph $G[U]$ is an \emph{isolated subgraph for node sum} if and only if
  \begin{itemize}
    \item $G[U]$ is a DAG with a unique source $s$ and a unique sink $t$, where $s\neq t$.
    \item Any path in $G$ starting from some $x \in V \setminus U$ and ending at some $y\in U$ must either (i) go through $s$, or (ii) contain $t = y$ as the unique node in $U$. 
    \item Any path in $G$ starting from some $x \in U$ and ending at some $y \in V \setminus U$ must either (i) go through $t$, or (ii) contain $s = x$ as the unique node in $U$.
  \end{itemize}
\end{definition}

The key difference is that we permit edges to go from our subgraph to the surrounding graph if they originate from \emph{either} $s$ or $t$. For \truemem\ model, this was not allowed, because such edges would affect when $s$'s output could be released. However, in node sum, the effect of $s$ on memory is restricted to the moment it occurs.

\subsection{Segmentation}

We introduce the notion of a canonical segmentation of the unique schedule. We show that when choosing a schedule for the surrounding graph, each segment we construct can be assumed to have its nodes appear consecutively, i.e.\ treated as a single scheduling unit.

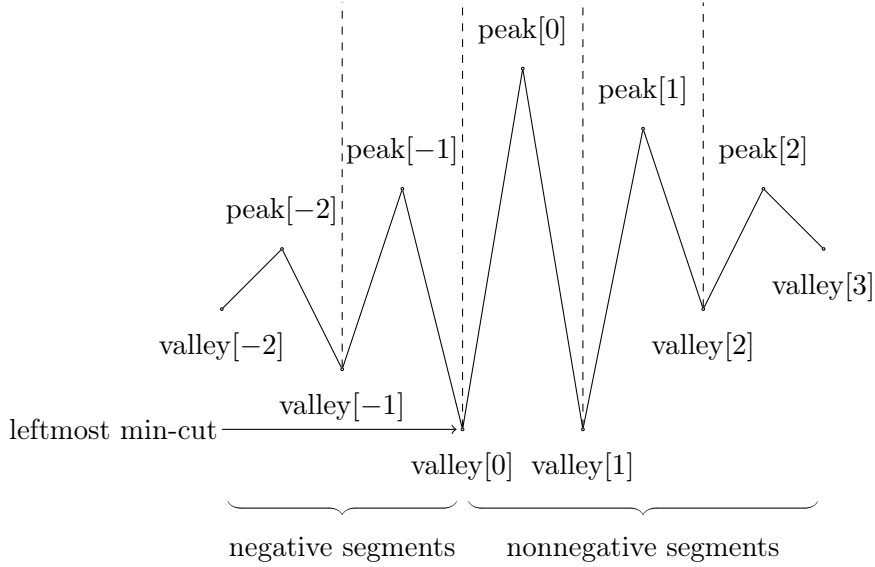
\begin{figure}
\centering
\begin{tikzpicture}[%
  auto,
  scale=0.8,
  void/.style={
    circle,
    draw=black,
    inner sep=0pt,
  },
  ]
  \node[void] (v-2) at (-5, 2) {}; \node [below=0.15cm of v-2] {$\valley{-2}$};
  \node[void] (p-2) at (-4, 3) {}; \node [above=0.15cm of p-2] {$\peak{-2}$};
  \node[void] (v-1) at (-3, 1) {}; \node [below=0.15cm of v-1] {$\valley{-1}$};
  \node[void] (p-1) at (-2, 4) {}; \node [above=0.15cm of p-1] {$\peak{-1}$};
  \node[void] (v+0) at (-1, 0) {}; \node [below=0.15cm of v+0] {$\valley{ 0}$};
  \node[void] (p+0) at ( 0, 6) {}; \node [above=0.15cm of p+0] {$\peak{ 0}$};
  \node[void] (v+1) at ( 1, 0) {}; \node [below=0.15cm of v+1] {$\valley{ 1}$};
  \node[void] (p+1) at ( 2, 5) {}; \node [above=0.15cm of p+1] {$\peak{ 1}$};
  \node[void] (v+2) at ( 3, 2) {}; \node [below=0.15cm of v+2] {$\valley{ 2}$};
  \node[void] (p+2) at ( 4, 4) {}; \node [above=0.15cm of p+2] {$\peak{ 2}$};
  \node[void] (v+3) at ( 5, 3) {}; \node [below=0.15cm of v+3] {$\valley{ 3}$};
  
  \draw (v-2) -- (p-2) -- (v-1) -- (p-1) -- (v+0) -- (p+0) -- (v+1) -- (p+1) -- (v+2) -- (p+2) -- (v+3);
  
  \draw [decorate, decoration = {calligraphic brace, amplitude=5pt}] (-1.1, -1.2) --  (-4.9, -1.2);
  \node (negative) at (-3, -2) {negative segments};
  \draw [decorate, decoration = {calligraphic brace, amplitude=5pt}] (4.9, -1.2) --  (-0.9, -1.2);
  \node (negative) at (2, -2) {nonnegative segments};
  
  \draw[dashed] (v-1) -- (-3, 7.1);
  \draw[dashed] (v+0) -- (-1, 7.1);
  \draw[dashed] (v+1) -- (1, 7.1);
  \draw[dashed] (v+2) -- (3, 7.1);
  
  \draw[->] (-5, 0) -- (-1.1, 0);
  \node (mincut) at (-6.8, 0) {leftmost min-cut};
\end{tikzpicture}
\caption{High-level view of a segment decomposition as given by \cref{defn:segments}. The decomposition begin at the leftmost min-cut, then recursively finds peaks and valleys going in both directions. Boundaries between segments are denoted with dashed vertical lines. In general, a memory profile might not transition in perfectly straight lines between peaks and valleys, but it must stay between the closest peak and valley due to \cref{lem:segment-peak-valley}.}
\label{fig:tikz-segments}
\end{figure}

Let $G = (V, E)$ be a weighted DAG in the node sum memory model. Suppose the induced subgraph $G[U]$ for some $U \subseteq V$ is an isolated subgraph for node sum and is a directed path $\pi = (\pi_1, \pi_2, \dots, \pi_k)$.

\begin{definition}[Directional Peaks and Valleys]
\label{defn:directional-peaks-valleys}
  Suppose we have a weighted directed path $\pi = (\pi_1, \pi_2, \cdots, \pi_k)$ in the node sum memory model.

  For $1 \le \ell \le r \le k-1$, the left-peak over the subpath $\pi_{\ell:r}$ is the smallest index $i \in \ell:r$ achieving peak memory: $\leftpeak{\pi_{\ell:r}} := \min \arg \max_{i \in \ell:r} \nodesumcost{\pi_{1:i}}$. Similarly, the right-peak is the largest index: $\rightpeak{\pi_{\ell:r}} := \max \arg \max_{i \in \ell:r} \nodesumcost{\pi_{1:i}}$.
  
  For $1 \le \ell \le r \le k$, the left-valley over the subpath $\pi_{\ell:r}$ is the smallest index $i \in \ell:r$ achieving minimum memory: $\leftvalley{\pi_{\ell:r}} := \min \arg \min_{i \in \ell:r} \nodesumcost{\pi_{1:i}}$. Similarly, the right-valley is the largest index: $\rightvalley{\pi_{\ell:r}} := \max \arg \min_{i \in \ell:r} \nodesumcost{\pi_{1:i}}$.
\end{definition}

Let $\mincut = \leftvalley{\pi_{1:(k-1)}}$. \zoya{why does it go until k-1 and not k?} We will define two groups of segments: \emph{negative segments} that partition $(\pi_1, ..., \pi_{\mincut})$, and \emph{nonnegative segments} that partition $(\pi_{\mincut+1}, ..., \pi_k)$\footnote{To contrast with previous work on in-trees \cite{liu1987application}, the minimum cut in in-trees always occurrs at the beginning of the schedule (with zero memory usage) and hence only nonnegative segments are constructed.}.

Segments are defined recursively. We start with the (earliest) point having minimum memory, then proceed from it in both directions. Going right, find the highest peak and then the lowest valley after it. The segment goes from the min-cut to that valley. To define the next segment, find the highest peak after the current valley, and the lowest valley after it. This defines the nonnegative segments, whose peaks get lower and valleys get higher, the farther we get from the min-cut toward the last node $\pi_k$. Thus, such a nonnegative segment $(i, j)$ between two valleys has $\nodesumcost{\pi_{1:j}} \ge \nodesumcost{\pi_{1:i}}$, hence the name ``nonnegative''\footnote{Only the first nonnegative segment can have equality and hence be a ``zero'' segment; this is an artifact of starting from the earliest point with minimum memory.}. A symmetrical process happens going left from the min-cut, toward the first node $\pi_1$. In that case a segment $(i, j)$ between two valleys has $\nodesumcost{\pi_{1:j}} < \nodesumcost{\pi_{1:i}}$, so it is called ``negative''.

\newcommand{\firstvalley}{a}
\newcommand{\lastvalley}{z}

Formally, these segments are defined as follows.
\begin{definition}[Segments]
\label{defn:segments}
  We first define nonnegative segments. Let $\valley{0} := \mincut = \leftvalley{\pi_{1:(k-1)}}$. Then define:
  \begin{align*}
    \peak{i} &:= \rightpeak{\pi_{ \valley{i} : (k-1)}},\\
    \valley{i+1} &:= \rightvalley{\pi_{ \peak{i} : (k-1)}},
  \end{align*}
  until some $\valley{\lastvalley} = k-1$. The nonnegative segments are:
  \begin{align*}
    \positivesegments[] &:= \left\{
      (\valley{0} + 1):\valley{1},
      (\valley{1} + 1):\valley{2},
      \cdots,
      (\valley{\lastvalley - 1} + 1):\valley{\lastvalley}
    \right\}
  \end{align*}
  
  We next define negative segments. We keep $\valley{0} := \mincut$ and define:
  \begin{align*}
    \peak{i-1} &:= \leftpeak{\pi_{ 1:\valley{i} }},\\
    \valley{i-1} &:= \leftvalley{\pi_{ 1:\peak{i-1} }},
  \end{align*}
  until some $\valley{\firstvalley} = 1$. The negative segments are:
  \begin{align*}
    \negativesegments[] &:= \left\{
      (\valley{\firstvalley} + 1):\valley{\firstvalley + 1},
      \cdots,
      (\valley{-2} + 1): \valley{-1},
      (\valley{-1} + 1): \valley{0}
    \right\}
  \end{align*}
\end{definition}

By this definition, each segment $\segment$ is of the form $(\valley{i}+1):\valley{i+1}$ and contains $\peak{i}$. We next state a technical lemma that helps guarantee the memory usage along a segment stays bounded between its two valleys and its peak. The proof of this lemma is deferred to \cref{sec:appendix:segments}.

\begin{restatable}{lemma}{segmentpeakvalley}
\label{lem:segment-peak-valley}
  For all $i \in (\firstvalley:\lastvalley-1)$: 
  \begin{align*}
    \valley{i}   &= \leftvalley{\pi_{\valley{i}:\peak{i}}} \\
    \peak{i}     &= \rightpeak{\pi_{\valley{i}:\peak{i}}} \\
    \peak{i}     &= \leftpeak{\pi_{\peak{i}:\valley{i+1}}} \\
    \valley{i+1} &= \rightvalley{\pi_{\peak{i}:\valley{i+1}}} \\
  \end{align*}
\end{restatable}

The following lemma states that schedules for the surrounding graph should keep our segments intact.
The proof of this lemma is deferred to \cref{sec:appendix:segments}.

\begin{restatable}{theorem}{segmentconsecutive}
    
\label{thm:segment-consecutive}
  Let $G = (V, E)$ be a weighted DAG in the node sum memory model. Suppose the induced subgraph $G[U]$ is an isolated subgraph for node sum (\cref{def:isolated-node-sum}), for some $U \subseteq V$. Additionally it is a directed path $\pi = (\pi_1, \pi_2,\dots , \pi_k)$. Then the partition of this path into segments, as defined above, satisfies the following: for any schedule $\sigma$ of $G$, there is a schedule $\sigma'$ of $G$ such that $\sigma' \ledot \sigma$ and each segment appears consecutively in $\sigma'$.
\end{restatable}

\subsection{Merging Segments}

\newcommand{\sortvalue}[1]{\mathrm{value}_\mathrm{sort}\left(#1\right)}

We now describe how to combine two (or more) sets of segments. We claim that they should be sorted according to the following criteria.
\begin{definition}[Segment Sort Value]
  Suppose we have a nonnegative segment $\positivesegment = \valley{i}:\valley{i+1}$ with peak $\peak{i}$. Its sort value is:
  \begin{align*}
    \sortvalue{\positivesegment}
      &:= \left(
                    1 + \nodesumcost{\pi_{1:\peak{i}}} -
                    \nodesumcost{\pi_{1:\valley{i+1}}}
                  \right)^{-1}
  \end{align*}
  Suppose we have a negative segment $\negativesegment = \valley{i}:\valley{i+1}$ with peak $\peak{i}$. Its sort value is:
  \begin{align*}
    \sortvalue{\negativesegment}
      &:= \left(
                    -1 - \nodesumcost{\pi_{1:\peak{i}}} +
                    \nodesumcost{\pi_{1:\valley{i}}}
                  \right)^{-1}
  \end{align*}
\end{definition}

We claim that in the absence of any precedence constraints, a (merged) set of segments should be sorted in nondecreasing order. \zoya{Would it simplify things to get rid of the $^{-1}$ and sort non-increasing?} Intuitively, this definition of sort value has the following characteristics:
\begin{itemize}
  \item Since the peak of a segment uses more memory than either valley, nonnegative segments get a positive sort value and negative segments get a negative sort value. Hence all negative segments occur before all nonnegative segments.
  \item Among nonnegative segments, segments with a larger gap between their peak and their right valley occur before segments with a smaller such gap.
  \item Among negative segments, segments with a smaller gap between their peak and their left valley occur before segments with a larger such gap.
\end{itemize}

\begin{lemma}
\label{lem:segmentation-is-sorted}
  Segments produced by \cref{defn:segments} have strictly increasing sort values.
\end{lemma}

\begin{proof}
  By construction, all negative segments come before all nonnegative segments, so we only need to compare within each type.
  
  For nonnegative segments, the peaks have strictly decreasing memory usage because each peak was chosen to be the largest index possible that had that memory usage. Additionally, the valleys have weakly increasing memory usage (in fact, only the first two can tie; the rest are also chosen to be the largest index possible). As a result, the gaps between peak and right valley must be strictly decreasing. Since the sort values invert these gaps, they are strictly increasing.
  
  For negative segments, the peaks have strictly decreasing memory usage \emph{going from high index to low index} (remember negative segments are defined by a ``reverse'' process). The valleys have strictly increasing memory usage \emph{going from high index to low index}. As a result, the gaps between peak and left valley must be strictly decreasing \emph{going from high index to low index}, and strictly increasing in the normal order. Since the sort values both negate and invert these gaps, they are strictly increasing as well.
  
  This completes the proof.
\end{proof}

\newcommand{\sigmasort}{\sigma_{\mathrm{sort}}}

Finally we state the main theorem of this section. The proof of this theorem builds on the lemmas established in this section, and we defer the full proof to \cref{sec:appendix:segments}.

\begin{restatable}{theorem}{segmentmerge}
\label{thm:segment-merge}
   Given a set of directed paths and subsets of their associated segments, $\{(\pi, \segments)\}$, consider any schedule $\sigmasort$ obtained by arranging them in nondecreasing sort order (breaking ties arbitrarily). Let $\sigma$ be any schedule which arranges them in any order (but keeps nodes from the same segment in order). Then $\sigmasort \ledot \sigma$.
\end{restatable}

\section{Series-Parallel Graphs}
\label{sec:sp}

In this section, we present our algorithm for series-parallel (SP) computation graphs. Although there is known to be a polynomial-time algorithm in the edge-weighted model \cite{kayaaslan2018scheduling},\footnote{As we mention earlier, results for the edge-weighted model do not apply here. In particular, pebbling and computation graphs that are series-parallel are not series-parallel in the edge-weighted model.} it is NP-hard to compute a peak-minimizing schedule for SP graphs in the computation graph model (which we show in Section~\ref{sec:hardness}). As a consequence, our algorithm presented in this section will have an exponential dependence on $d$, the maximum out-degree of the SP graph.
We first formally define (two-terminal) SP directed graphs~\cite{DBLP:journals/siamcomp/ValdesTL82}. 

\begin{definition}[Series-Parallel Graph]
\label{def:sp-graph}
A \emph{series-parallel} graph $G$ with source $s$ and sink $t$ is a DAG, recursively defined to be one of the following:
\begin{itemize}
  \item \textbf{Base case:} A graph with two vertices $s$ and $t$ and a directed edge $(s, t)$.
  \item \textbf{Series composition:} The series composition of two SP graphs, $G_1$ with source $s_1$ and sink $t_1$, and $G_2$ with source $s_2$ and sink $t_2$, is formed by identifying $t_1 = s_2$, and has source $s = s_1$ and sink $t = t_2$.
  \item \textbf{Parallel composition:} The parallel composition of two SP graphs, $G_1$ with source $s_1$ and sink $t_1$, and $G_2$ with source $s_2$ and sink $t_2$, is formed by identifying the source $s = s_1 = s_2$ and the sink $t = t_1 = t_2$. If any parallel edges are formed as a result, delete one of the two parallel edges.
\end{itemize}
\end{definition}

A \emph{series-parallel decomposition} of an SP graph is a tree demonstrating how \cref{def:sp-graph} applies to this graph. Leaves of the tree correspond to the edges of the SP graph (representing the base case), and internal tree nodes correspond to the intermediate SP graphs obtained by a series or a parallel composition of their children, with the root of the tree corresponding to the final graph.

\begin{restatable}{theorem}{spmain}
\label{thm:spmain}
    Let $G=(V, E)$ be a series-parallel graph with $n$ nodes and let $d$ denote the maximum out-degree of any node in $G$. Then there is an algorithm to find the dominant schedule for $G$ (in the computation graph memory model) in time $2^{O(d \log d)}\cdot poly(n)$.
\end{restatable}

Our algorithmic plan is to perform two reductions. Note that the first reduction stays in the \truemem\ model, but the second one transitions to the node sum model.  First, using our framework of linearization, we show that it suffices to schedule several instances of special cases of SP graphs, which we call \sintrees{}, in the \truemem\  model. Next, for each \sintree{}, we consider $2^{O(d \log d)}$ guesses for the order in which a set of its key nodes appears in the optimal schedule. For each such guess, we reduce the \sintree{} instance in the \truemem\ model to a \emph{caduceus} graph instance in the node sum model. Then we solve this caduceus instance using a greedy algorithm. \cref{thm:spmain} will follow from \cref{thm:starborescencetosp}, \cref{thm:caduceustosb}, and \cref{cor:caduceus}, which formalize these reductions below.

The two new graph classes are formally defined as follows. 
\begin{definition}[\Sintree{}]
   A \sintree{} is an SP graph $G = (V, E)$, with source $s$ and sink $t$, such that removing $s$, i.e. $G[V \setminus \{s\}]$, yields an in-tree\footnote{An in-tree is a directed graph such that the underlying undirected graph forms a tree and each node has at most one out-neighbor. Strictly speaking, in-trees are usually referred to as anti-arborescences, but that does not portmanteau as well with the star emanating from $s$.} rooted at $t$. 
\end{definition}

\begin{definition}[Caduceus\protect\footnotemark \ Graph]\footnotetext{The caduceus is a symbol of Hermes in Greek mythology and depicts a short staff that is entwined by two serpents.}
  A graph $G = (V, E)$ is caduceus if it can be written in the following form. The core of the graph is the staff $\staff$, which is a directed path $\staff = (v_1, v_2, \dots, v_p)$. The remainder of the graph consists of $b$ additional bends $\bends$; each bend $\bend \in \bends$ is a directed path that starts and ends on the staff: $\bend = (v_\ell, u_1, \dots, u_k, v_r)$ for some $\ell < r$. Each bend has at least one internal bend node and these internal bend nodes $u_1, \dots, u_k$ are distinct over all bends and are distinct from the nodes of the staff.
\end{definition}

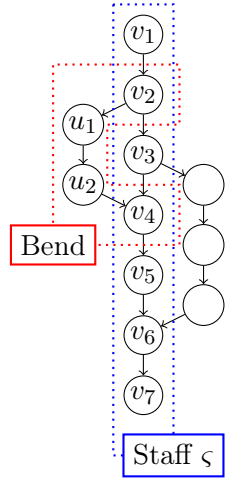
\begin{figure}
\centering
\begin{tikzpicture}[%
  auto,
  scale=0.8,
  hollownode/.style={
    circle,
    draw=black,
    inner sep=1pt,
  },
  blueblock/.style={
    rectangle,
    solid,
    draw=blue,
    fill=white,
    text=black,
    align=center,
  },
  redblock/.style={
    rectangle,
    solid,
    draw=red,
    fill=white,
    text=black,
    align=center,
  },
  ]
  \node[hollownode] (v1) at ( 0, 6) {$v_1$};
  \node[hollownode] (v2) at ( 0, 5) {$v_2$};
  \node[hollownode] (v3) at ( 0, 4) {$v_3$};
  \node[hollownode] (v4) at ( 0, 3) {$v_4$};
  \node[hollownode] (v5) at ( 0, 2) {$v_5$};
  \node[hollownode] (v6) at ( 0, 1) {$v_6$};
  \node[hollownode] (v7) at ( 0, 0) {$v_7$};
  
  \node[hollownode] (u1) at (-1, 4.5) {$u_1$};
  \node[hollownode] (u2) at (-1, 3.5) {$u_2$};
  
  \node[hollownode] (uu1) at (1, 3.5) {\phantom{$u_1$}};
  \node[hollownode] (uu2) at (1, 2.5) {\phantom{$u_2$}};
  \node[hollownode] (uu3) at (1, 1.5) {\phantom{$u_3$}};
  
  \draw[->] (v1) -- (v2);
  \draw[->] (v2) -- (v3);
  \draw[->] (v3) -- (v4);
  \draw[->] (v4) -- (v5);
  \draw[->] (v5) -- (v6);
  \draw[->] (v6) -- (v7);
  
  \draw[->] (v2) -- (u1);
  \draw[->] (u1) -- (u2);
  \draw[->] (u2) -- (v4);
  
  \draw[->] (v3) -- (uu1);
  \draw[->] (uu1) -- (uu2);
  \draw[->] (uu2) -- (uu3);
  \draw[->] (uu3) -- (v6);
  
  \draw[blue, thick, dotted] (-0.5, 6.5) rectangle (0.5, -1) node[blueblock] {Staff $\staff$};
  \draw[red, thick, dotted] (-1.5, 2.5) node[redblock] {Bend} -- (-1.5, 5.5) -- (0.6, 5.5) -- (0.6, 4.5) -- (-0.6, 4.5) -- (-0.6, 3.5) -- (0.6, 3.5) -- (0.6, 2.5) -- cycle;
\end{tikzpicture}
\caption{A caduceus graph with a seven-node staff $\staff$ (marked in blue) and two bends $\bends$ (one bend $\bend \in \bends$ is marked in red).}
\label{fig:tikz-caduceus}
\end{figure}
A caduceus graph is depicted in \cref{fig:tikz-caduceus}.

\subsection{Reducing SP Graphs to \Sintrees{}}

We now present the first reduction. Staying in the \truemem\ model, it uses linearization to reduce an SP graph problem to a set of $O(n)$ \sintree{} problems.

\begin{restatable}{theorem}{starborescencetosp}
\label{thm:starborescencetosp}
  Suppose we have an algorithm that, given a \sintree{} with $n$ nodes and maximum out-degree $d$ in the \truemem\  model, finds a dominant schedule in $f(n,d)$ time. Then there is an algorithm that, given an SP graph with $n$ nodes and maximum out-degree $d$ in the \truemem\ model, finds a dominant schedule in $O(n^2 + n \cdot f(n, d))$ time.
\end{restatable}

\begin{figure}
\centering
\begin{tikzpicture}[%
  auto,
  scale=0.8,
  hollownode/.style={
    circle,
    draw=black,
    inner sep=1pt,
  },
  shadednode/.style={
    circle,
    draw=black,
    fill=yellow,
    inner sep=1pt,
  }
  ]
  \node[hollownode] (v0) at ( 0, 4) {0};
  \node[hollownode] (v1) at (-1, 3) {1};
  \node[shadednode] (v2) at ( 1, 3) {2};
  \node[hollownode] (v3) at (-2, 2) {3};
  \node[shadednode] (v4) at ( 0, 2) {4};
  \node[shadednode] (v5) at ( 1, 2) {5};
  \node[shadednode] (v6) at ( 2, 2) {6};
  \node[hollownode] (v7) at (-1, 1) {7};
  \node[shadednode] (v8) at ( 1, 1) {8};
  \node[hollownode] (v9) at ( 0, 0) {9};
  
  \draw[->] (v0) -- (v1);
  \draw[->] (v0) -- (v2);
  \draw[->] (v1) -- (v3);
  \draw[->] (v1) -- (v7);
  \draw[->] (v2) -- (v4);
  \draw[->] (v2) -- (v5);
  \draw[->] (v2) -- (v6);
  \draw[->] (v3) -- (v7);
  \draw[->] (v4) -- (v8);
  \draw[->] (v5) -- (v8);
  \draw[->] (v6) -- (v8);
  \draw[->] (v7) -- (v9);
  \draw[->] (v8) -- (v9);
\end{tikzpicture}
\caption{Locating a \sintree{} in an SP graph. The nodes are labeled with their order from an initial topological sort. Node $2$ is the last node with out-degree at least two. Algorithm~\ref{alg:sintree} identifies the set of yellow nodes.}
\label{fig:tikz-sintree}
\end{figure}
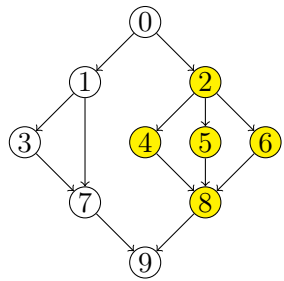

\begin{proof}
  We repeatedly perform the following steps, until the whole graph becomes a path:
  \begin{enumerate}
    \item identify a \sintree{}, which is an isolated subgraph, in the SP graph,
    \item apply the assumed subroutine on it to get a dominant schedule,
    \item use Linearization (Lemma \ref{lem:linearization}) to replace the \sintree{} with a path where the nodes are ordered according to the dominant schedule.
  \end{enumerate}
  In each iteration, we remove at least one parallel composition from our SP graph. Since our original graph can involve at most $O(n)$ parallel compositions, we only need to perform this sequence of steps $O(n)$ times. It remains to explain how to perform each of these computations in $O(n + f(n, d))$ time.
  
  \begin{algorithm}[H]
  \SetAlgoLined
  \KwData{SP-Graph $G = (V, E)$ which is not entirely series (i.e., not a path)}
  \KwResult{A subset of vertices $U$ such that $G[U]$ is a \sintree{} and is an isolated subgraph (\cref{defn:isolated})}
  
  Let $\sigma$ be a topological sort of $V$\;
  Let $s$ be the last node in $\sigma$ with $\abs{\outneighbors{s}} \ge 2$\;
  Compute a series-parallel decomposition of $G$\;
  \Return the set of nodes $U$ of the highest-level parallel composition node in the decomposition tree that has $s$ as a source\; \label{algline:source}
  \caption{$\textsc{LocateStarborescence}(G)$}
  \label{alg:sintree}
\end{algorithm}
  In order to identify a \sintree{}, we run \cref{alg:sintree}. We first prove the correctness of this algorithm. We begin doing so by explaining why \cref{algline:source} is guaranteed that $s$ is a source of some composition node. This is because $s$ has out-degree at least two, the base case for SP graphs only has nodes with out-degree at most one, and only parallel composition increases the out-degree of nodes (and only the source).
  
  We now show that $G[U]$ is a \sintree{}. The key observation is that, since $s$ is the source of $G[U]$ and $s$ was chosen to be the last node in a topological sort with out-degree at least two, all other nodes of $U$ must have out-degree at most one. Hence they must form an in-tree whose sink $t$ is the same as that of the composition identified by the algorithm.
  
  Next, we show that $G[U]$ is an isolated subgraph (\cref{defn:isolated}). Recall that $s$ and $t$ are the source and sink of the highest-level composition identified by the algorithm. From the point of view of one of the subgraphs, compositions only add edges via identifying two nodes together, and the only nodes that are so identified are sources and sinks. Hence any path in $G$ starting from some $x \in V \setminus U$ and ending at some $y \in U$ must go through $s$ or contain $t = y$ as the unique node in $U$.
  
  We have almost proven the third property in \cref{defn:isolated} as well; we already know that a path starting from $U$ must go through either $s$ or $t$ and want to show that it only goes through $t$. Since $s$ is a source, it gains in-edges when involved in a series composition and out-edges when involved in a parallel composition. However, it cannot be involved in any parallel compositions because it would remain the source and our algorithm chose the highest-level composition that had $s$ as a source. Hence $s$ cannot gain out-edges and only $t$ can. We have shown that $G[U]$ is an isolated subgraph.
  
  Next, we consider the runtime of \cref{alg:sintree}. We first remark that for SP graphs, the number of edges $m$ is $O(n)$ (this can be observed by undirecting their edges and either noting that the resulting graph is planar or has constant treewidth).
  With this in mind, our topological sort and scanning for $s$ take $O(n)$ time. Next, computing a series-parallel decomposition can be done in linear time as well \cite{DBLP:journals/siamcomp/ValdesTL82}. Since such decompositions only involve $O(n)$ compositions, we can find the highest-level composition with $s$ as a source in $O(n)$ time as well.
  
  The other two steps are easy to do accounting for; applying the subroutine requires $f(n, d)$ time by definition, and Linearization (\cref{lem:linearization}) takes $O(m + n) = O(n)$ time. This completes the proof.
\end{proof}

\subsection{Reducing \Sintrees\ to Caduceus Graphs}

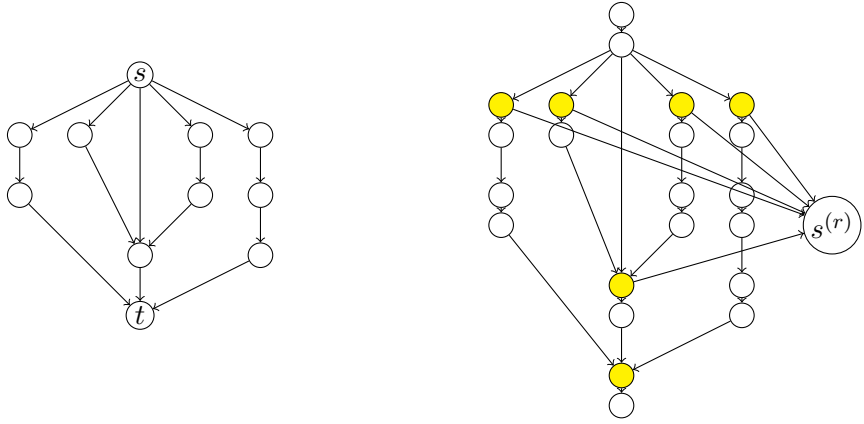
\begin{figure}
\centering
\begin{tikzpicture}[%
  auto,
  scale=0.8,
  hollownode/.style={
    circle,
    draw=black,
    inner sep=1pt,
  },
  shadednode/.style={
    circle,
    draw=black,
    fill=yellow,
    inner sep=1pt,
  }
  ]

  \node[hollownode]  (s) at ( 0, 5) {$s$};
  \node[hollownode] (v0) at (-2, 4) {\vphantom{0}};
  \node[hollownode] (v1) at (-1, 4) {\vphantom{0}};
  \node[hollownode] (v2) at ( 1, 4) {\vphantom{0}};
  \node[hollownode] (v3) at ( 2, 4) {\vphantom{0}};
  \node[hollownode] (v4) at (-2, 3) {\vphantom{0}};
  \node[hollownode] (v5) at ( 1, 3) {\vphantom{0}};
  \node[hollownode] (v6) at ( 2, 3) {\vphantom{0}};
  \node[hollownode] (v7) at ( 0, 2) {\vphantom{0}};
  \node[hollownode] (v8) at ( 2, 2) {\vphantom{0}};
  \node[hollownode] (t) at ( 0, 1) {$t$};
  
  \draw[->]  (s) -- (v0);
  \draw[->]  (s) -- (v1);
  \draw[->]  (s) -- (v2);
  \draw[->]  (s) -- (v3);
  \draw[->]  (s) -- (v7);
  \draw[->] (v0) -- (v4);
  \draw[->] (v1) -- (v7);
  \draw[->] (v2) -- (v5);
  \draw[->] (v3) -- (v6);
  \draw[->] (v4) -- (t);
  \draw[->] (v5) -- (v7);
  \draw[->] (v6) -- (v8);
  \draw[->] (v7) -- (t);
  \draw[->] (v8) -- (t);

  \node[hollownode]  (ss) at (8, 6) {\vphantom{0}};
  \node[hollownode]  (sf) at (8, 5.5) {\vphantom{0}};
  \node[hollownode] (sr) at (11.5, 2.5) {$\releasevertex{s}$};
  \node[shadednode] (v0s) at (6, 4.5) {\vphantom{0}};
  \node[hollownode] (v0f) at (6, 4) {\vphantom{0}};
  \node[shadednode] (v1s) at (7, 4.5) {\vphantom{0}};
  \node[hollownode] (v1f) at (7, 4) {\vphantom{0}};
  \node[shadednode] (v2s) at (9, 4.5) {\vphantom{0}};
  \node[hollownode] (v2f) at (9, 4) {\vphantom{0}};
  \node[shadednode] (v3s) at (10, 4.5) {\vphantom{0}};
  \node[hollownode] (v3f) at (10, 4) {\vphantom{0}};
  \node[hollownode] (v4s) at (6, 3) {\vphantom{0}};
  \node[hollownode] (v4f) at (6, 2.5) {\vphantom{0}};
  \node[hollownode] (v5s) at (9, 3) {\vphantom{0}};
  \node[hollownode] (v5f) at (9, 2.5) {\vphantom{0}};
  \node[hollownode] (v6s) at (10, 3) {\vphantom{0}};
  \node[hollownode] (v6f) at (10, 2.5) {\vphantom{0}};
  \node[shadednode] (v7s) at (8, 1.5) {\vphantom{0}};
  \node[hollownode] (v7f) at (8, 1) {\vphantom{0}};
  \node[hollownode] (v8s) at (10, 1.5) {\vphantom{0}};
  \node[hollownode] (v8f) at (10, 1) {\vphantom{0}};
  \node[shadednode] (ts) at (8, 0) {\vphantom{0}};
  \node[hollownode] (tf) at (8, -0.5) {\vphantom{0}};

  \draw[->] (ss) -- (sf);
  \draw[->] (v0s) -- (v0f);
  \draw[->] (v1s) -- (v1f);
  \draw[->] (v2s) -- (v2f);
  \draw[->] (v3s) -- (v3f);
  \draw[->] (v4s) -- (v4f);
  \draw[->] (v5s) -- (v5f);
  \draw[->] (v6s) -- (v6f);
  \draw[->] (v7s) -- (v7f);
  \draw[->] (v8s) -- (v8f);
  \draw[->] (ts) -- (tf);
  
  \draw[->]  (sf) -- (v0s);
  \draw[->]  (sf) -- (v1s);
  \draw[->]  (sf) -- (v2s);
  \draw[->]  (sf) -- (v3s);
  \draw[->]  (sf) -- (v7s);
  \draw[->] (v0f) -- (v4s);
  \draw[->] (v1f) -- (v7s);
  \draw[->] (v2f) -- (v5s);
  \draw[->] (v3f) -- (v6s);
  \draw[->] (v4f) -- (ts);
  \draw[->] (v5f) -- (v7s);
  \draw[->] (v6f) -- (v8s);
  \draw[->] (v7f) -- (ts);
  \draw[->] (v8f) -- (ts);

  \draw[->] (v0s) -- (sr);
  \draw[->] (v1s) -- (sr);
  \draw[->] (v2s) -- (sr);
  \draw[->] (v3s) -- (sr);
  \draw[->] (v7s) -- (sr);
        
\end{tikzpicture}
\caption{A \sintree{} in the \truemem\ model (left) and its transformation to the node sum model (right). The key nodes of the node sum graph $G$ are highlighted.}
\label{fig:tikz-key-nodes}
\end{figure}

\begin{restatable}{theorem}{caduceustosb}
\label{thm:caduceustosb}
  Suppose we have an algorithm that given a caduceus graph, finds a dominant schedule in the node sum model in $f(n)$ time. Then there is an algorithm that given a \sintree{} with maximum out-degree $d$, finds a dominant schedule in the \truemem\   model in $2^{O(d \log d)} \cdot O(n^2 + f(n))$ time.
\end{restatable}

\begin{proof}
  We start by transforming the \sintree{} from the \truemem\  model into the node sum model using \Cref{alg:reduction}. Let $G$ be the resulting graph (see \cref{fig:tikz-key-nodes}). Note that in $G$, only $s$ will have a release node, as it is the only node with out-degree greater than one. This release node $\releasevertex{s}$ will have start nodes of $s$'s out-neighbors as its in-neighbors. Aside from this one release node, the structure of the graph stays similar to what it was, only with each node $v$ replaced by an edge $(\startvertex{v}, \finishvertex{v})$.

  Next, we designate a subset of the nodes of $G$ as \emph{key nodes}. In particular, the following two types of nodes are key nodes: (i) any out-neighbor of $\finishvertex{s}$ and (ii) any node with an in-degree of at least two, except for the release node $\releasevertex{s}$.
  We observe that there are at most $O(d)$ key nodes. The out-degree of $\finishvertex{s}$, which is the same as the out-degree of the original root node $s$ in the \sintree{}, is $d$, so there are $d$ key nodes of the first type. To bound the number of key nodes of the second type, recall that removing $s$ from the \sintree{} forms an in-tree with at most $d$ leaves. Such a tree can have at most $d-1$ nodes with in-degree of at least two. 
  As the key nodes of the second type correspond to such in-tree nodes, there are also at most $d-1$ of them. Hence, there are at most $O(d)$ key nodes total.

  Next, we ``guess'' the order in which the key nodes appear in the dominant schedule for the node sum graph $G$. Specifically, we try all possible topologically-consistent orderings of the key nodes. Since there are $O(d)$ key nodes, there are $2^{O(d \log d)}$ possible orderings of them (and even fewer topologically consistent ones). We also insert $\releasevertex{s}$ into the guessed ordering, immediately following its last in-neighbor. (Since all its in-neighbors are key nodes of the first type, they are all included in the ordering. Furthermore, by \cref{lem:eager}, inserting $\releasevertex{s}$ as soon as possible preserves a dominant schedule.)

  For each such guessed ordering, we add edges to a copy of $G$ that enforce this ordering. (Recall that in the node sum model, edges don't influence memory consumption, they only enforce precedence constraints.) We then find the dominant schedule on each such modified graph (as described below). Since each guess is topologically consistent, all these schedules are feasible for $G$. At the end, we compare all these schedules to find the overall dominant one. 
  To verify correctness, consider a dominant schedule of $G$, which exists by \cref{thm:existsdominant}. That dominant schedule agrees with one of our guesses of ordering key nodes. Hence the schedule we generate for that guess dominates this overall-dominant schedule, which in turn dominates all schedules. Hence at least one of the compared schedules dominates the rest. 

  The remaining task is, given a guessed ordering $\sigma$, to find a dominant schedule of $G$ which is consistent with $\sigma$. We do this by a reduction to a caduceus graph. To get the staff of the caduceus graph, we include $\startvertex{s}$, then $\finishvertex{s}$, then the nodes of $\sigma$ (including $\releasevertex{s}$, which immediately follows its last in-neighbor), then any chain that may follow the last key node in $G$. If any consecutive nodes of this staff don't have an edge, we add that edge. Next, we observe that all the nodes of $G$ which are not part of the staff are on disjoint directed paths between different staff nodes. Thus, the resulting graph is a caduceus. We find a dominant schedule for it in $f(n)$ time using the hypothesized algorithm.

  In terms of running time, in addition to calling our hypothesized algorithm $2^{O(d \log d)}$ times, we need to transform the graphs as described above, which involves $O(n)$ work per guess. Finally, we need to compare schedules to figure out which one is dominant. For the sake of simplicity, we invoke \cref{defn:algdef} to do so in quadratic time, but this can be optimized to linear time using \cref{sec:segments}.
\end{proof}

\subsection{Solving Caduceus Graphs}

We have now transformed our SP graph problem in the \truemem\ model into many caduceus graph problems in the node sum model. We complete our SP graph algorithm by proving the following theorem in this subsection.
\begin{restatable}{theorem}{caduceus}
\label{thm:caduceus}
  There is an algorithm that given a caduceus graph $(\staff, \bends)$ in the node sum model, finds a memory-optimal schedule in $O(np \log n)$ time, where $p$ is the length of the staff $\staff$. 
\end{restatable}

Recall that an algorithm that computes memory-optimal schedules can be used as a black-box to compute a dominant schedule via \cref{thm:reductiontopeakminimization}. This requires the hereditary property (\cref{defn:hereditary}), which we verify now for caduceus graphs.
\begin{claim}
  \label{claim:caduceushereditary}
The caduceus graphs form a hereditary class of single-source single-sink DAGs.
\end{claim}
\begin{proof}
  By symmetry, we only need to verify left-heredity. Given a caduceus graph $G=(V,E)$ and a topological cut $U$ of $G$, the induced subgraph $G[U]$ contains a prefix of the staff, some bends, and some prefixes of bends with a dangling end. Then, $\textsc{AddSink}(G[U])$ appends a sink at  the end of the staff, which also connects to the dangling nodes and turns them into bends. So $\textsc{AddSink}(G[U])$ is a caduceus graph.
\end{proof}

The following is a corollary of \cref{thm:caduceus}, \cref{thm:reductiontopeakminimization} (proved in \cref{subsec:reduction}),
and \cref{claim:caduceushereditary}. Combined with  \cref{thm:starborescencetosp} (reduction from SP graphs to \sintrees{}) and \cref{thm:caduceustosb} (reduction from \sintrees{} to caduceus graphs), it completes the proof of \cref{thm:spmain} (algorithm for SP graphs).
\begin{corollary}
\label{cor:caduceus}
  There is an algorithm that given a caduceus graph $(\staff, \bends)$ in the node sum model, finds a dominant schedule in $O(n^2 p \log n)$ time, where $p$ is the length of the staff $\staff$.
\end{corollary}

The rest of the section discusses how to find memory-optimal schedules for caduceus graphs in the node sum model, proving \cref{thm:caduceus}. 

\begin{proof}[Proof of \cref{thm:caduceus}]

We know that any schedule $\sigma$, when restricted to the nodes of the staff $\staff$, must contain them in the same order. Hence the remaining decision to make is, for each bend $\bend \in \bends$, how are its non-staff nodes $u_1, \ldots, u_k$ interleaved between $v_\ell$ and $v_r$?
 
We simplify the problem by dividing the non-staff nodes of each bend $\bend$ into segments using \cref{defn:segments}.  In particular, note that each bend is an isolated subgraph for node sum (\cref{def:isolated-node-sum}) since its nonstaff nodes are unique to that bend. Due to \cref{thm:segment-consecutive}, we can assume that the optimal schedule keeps these segments intact, so we are now trying to decide how these segments get interleaved between $v_\ell$ and $v_r$. It suffices to determine which segments occur between each pair of adjacent (on the staff) staff nodes. This is because we can also assume without loss of generality that our optimal schedule has the segments between each pair of adjacent staff nodes sorted according to the merge criteria in \cref{thm:segment-merge}. Importantly, this will not violate any precedence constraints, because two segments from the same bend will be in their original order if sorted due to \cref{lem:segmentation-is-sorted}.
So suppose that we have determined, for each $i$, that a set of bends' negative segments  $\negativesegments[i]$, as well as a set of their nonnegative segments $\positivesegments[i]$, should go between staff nodes $v_i$ and $v_{i+1}$. Then the resulting schedule will be $(v_1, \negativesegments[1], \positivesegments[1], v_2, \negativesegments[2], \positivesegments[2], \dots)$, where the segments in each $\negativesegments, \positivesegments$ are sorted according to \cref{thm:segment-merge}. \cref{alg:caduceus} allocates all the segments of all the bends into such sets.

\begin{algorithm}[htbp]
  \SetAlgoLined
  \KwData{Caduceus Graph $G = (V, E)$ with staff $\staff$, bends $\bends$, node-sum weights $w_v \in \mathbb{R}$, and memory threshold $\guess$.}
  \KwResult{A schedule $\sigma$ that uses at most $\guess$ memory if one exists; otherwise, ``no schedule''.}
  
  For $i$ from $1$ to $p - 1$, we initialize the set of negative segments between $v_i$ and $v_{i+1}$ to be $\negativesegments = \varnothing$ and the set of nonnegative segments between them to be $\positivesegments = \varnothing$\;
  \For{bend $\bend \in \bends$} {
    Unpack the bend as $\bend = (v_\ell, u_1, \cdots, u_k, v_r)$\;
    Segment the nonstaff nodes $(u_1, \cdots, u_k)$ using \cref{defn:segments}\;
    Add all produced negative segments to $\negativesegments[\ell]$\;
    Add all produced nonnegative segments to $\positivesegments[r-1]$\;
  }
  \tcc{For the rest of the algorithm, our current schedule is assumed to be $(v_1, \negativesegments[1], \positivesegments[1], v_2, \negativesegments[2], \positivesegments[2], \dots)$ where the segments in each $\negativesegments, \positivesegments$ are sorted according to \cref{thm:segment-merge}.}
  \While{the peak of the current schedule exceeds threshold $\guess$} {
  \label{algline:while}
    Choose an arbitrary peak\;
    \If{the peak occurs at a staff node $v_i$} {
      \label{algline:staff}
      \Return ``no schedule''\;
    }
    \If{the peak occurs in a negative segment $\negativesegment$ in some $\negativesegments$} {
      Attempt to move $\negativesegment$ and any subsequent negative segments from the same bend in $\negativesegments$ from $\negativesegments$ to $\negativesegments[i+1]$. If this is not possible (either because the bend ends at $r = i + 1$ or the bend has a nonnegative segment in $\positivesegments$), then return ``no schedule''\; \label{algline:move1}
    }
    \ElseIf{the peak occurs in a nonnegative segment $\positivesegment$ in some $\positivesegments$} {
      Attempt to move $\positivesegment$ and any preceding nonnegative segments from the same bend in $\positivesegments$ from $\positivesegments$ to $\positivesegments[i-1]$. If this is not possible (either because the bend begins at $\ell = i$ or the bend has a negative segment in $\negativesegments$), then return ``no schedule''\; \label{algline:move2}
    }
  }
  \Return the schedule induced by the $\negativesegments, \positivesegments$\;
  \caption{$\textsc{OptimalCaduceusHelper}(G, w, \guess)$}
  \label{alg:caduceus}
\end{algorithm}
We argue that \cref{alg:caduceus} correctly determines how to place segments between adjacent staff nodes in at most $\guess{}$ memory. Importantly, we have just established that any schedule $\sigma^{\star}$ can be transformed into sets of negative segments $\{\negativesegmentsstar\}$ and sets of nonnegative segments $\{\positivesegmentsstar\}$ without increasing its peak memory. 

Informally, the key insight that we use to analyze our algorithm is that compared to $\sigma^\star$, our algorithm will place every negative segment ``earlier'' and every nonnegative segment ``later''. By earlier/later, we mean where segments are placed relative to staff nodes.
To make this more precise, let $S = \bigcup_{i=1}^{p-1} (\negativesegments \cup \positivesegments)$ be the set of all segments. We can sort this set using \cref{thm:segment-merge}, breaking ties arbitrarily (but consistently); let the sorted sequence of segments be $\sorted$. Consider the sequence derived by flattening $(v_1, \sorted, v_2, \sorted, v_3, \cdots, v_{p-1}, \sorted, v_p)$ into a sequence of vertices; we can think of it as a pseudo-schedule where we have inserted $S$ in sorted order between every adjacent pair of staff nodes. We will refer to this flattened sequence as $A$. Next, we will produce a sequence $B$ of the same length as $A$ but which matches the current state of our algorithm's schedule, namely $(v_1, \negativesegments[1], \positivesegments[1], \dots)$. We do so by swapping some vertices in $A$ with a skip symbol $\epsilon$. Specifically, for the $S$ between $v_i$ and $v_{i+1}$ we will skip all vertices from segments in $S \setminus \negativesegments[i] \setminus \positivesegments[i]$; we do this for all $S$ in $A$. By construction, $B$ is just (the flattened version of) $(v_1, \negativesegments[1], \positivesegments[1], \dots)$ with additional skip $\epsilon$ symbols inserted. We do the same thing for $\sigma^\star$ to get sequence $C$.

Now that the algorithm's schedule and $\sigma^\star$ are properly aligned, we can prove the following:
\begin{claim}
\label{claim:greedy}
  At any point in the algorithm, for any negative segment $\negativesegment$, its nodes appear (weakly) earlier in $B$ than in $C$. For any nonnegative segment $\positivesegment$, its nodes appear (weakly) later in $B$ than in $C$. Furthermore, the algorithm never even attempts to move segments in a way that would violate this.
  
  Additionally, consider any staff node, or the first node of any appearance of a segment in $A$. Suppose it occurs at index $j$. Then the cut formed by whatever nodes are within the first $j-1$ entries of $B$ has less memory usage than the cut formed by whatever nodes are within the first $j-1$ entries of $C$.
\end{claim}

\begin{proof}
  First, we show that for any point in the algorithm for which the first half of the claim holds, the second half of the claim directly follows. Let $S_B$ be the nodes within the first $j-1$ entries of $B$ and $S_C$ be the nodes within the first $j-1$ entries of $C$.
  
  The first case is that we have some staff node $v_i$. Then we know:
  \begin{align*}
    S_B &= \bigcup_{k=1}^{i-1} \left[ \{v_k\} \cup \negativesegments[k] \cup \positivesegments[k] \right] \\
    S_C &= \bigcup_{k=1}^{i-1} \left[ \{v_k\} \cup \negativesegmentsstar[k] \cup \positivesegmentsstar[k] \right] \\
    \nodesumcost{S_B} &= \sum_{k=1}^{i-1} \left[
      w_{v_i} + \sum_{v \in \negativesegments[k]} w_v + \sum_{v \in \positivesegments[k]} w_v
    \right] \\
    \nodesumcost{S_C} &= \sum_{k=1}^{i-1} \left[
      w_{v_i} + \sum_{v \in \negativesegmentsstar[k]} w_v + \sum_{v \in \positivesegmentsstar[k]} w_v
    \right] \\
    \nodesumcost{S_B} - \nodesumcost{S_C} &= \sum_{k=1}^{i-1} \left[
      \sum_{v \in \negativesegments[k]} w_v +
      \sum_{v \in \positivesegments[k]} w_v -
      \sum_{v \in \negativesegmentsstar[k]} w_v -
      \sum_{v \in \positivesegmentsstar[k]} w_v
    \right] \\
      &= \underbrace{\sum_{k=1}^{i-1} \left[
           \sum_{v \in \negativesegments[k]} w_v -
           \sum_{v \in \negativesegmentsstar[k]} w_v
         \right]}_\text{negative segments difference} +
         \underbrace{\sum_{k=1}^{i-1} \left[
           \sum_{v \in \positivesegments[k]} w_v -
           \sum_{v \in \positivesegmentsstar[k]} w_v
         \right]}_\text{nonnegative segments difference}
  \end{align*}
  But the first half of the claim implies that this final quantity is negative! Negative segments affect the negative segments difference summation. Each negative segment can either (i) have none of its nodes show up in any $\negativesegments[k]$ or $\negativesegmentsstar[k]$, (ii) have all of its nodes show up in some $\negativesegments[k]$ and some $\negativesegmentsstar[k]$, or (iii) have all of its nodes show up in some $\negativesegments[k]$ but none of its nodes show up in any $\negativesegmentsstar[k]$. Together, the nodes of a negative segment have negative weight, so the ``negative segments difference'' double-summation is nonpositive.
  
  The reasoning for nonnegative segments is almost identical; each nonnegative segment can either (i) have none of its nodes show up in any $\positivesegments[k]$ or $\positivesegmentsstar[k]$, (ii) have all of its nodes show up in some $\positivesegments[k]$ and some $\positivesegmentsstar[k]$, or (iii) have none of its nodes show up in any $\positivesegments[k]$ but all of its nodes show up in some $\positivesegmentsstar[k]$. Together, the nodes of a positive segment have positive weight, so the ``positive segments difference'' double-summation is nonpositive as well.
  
  We have shown that $\nodesumcost{S_B} \le \nodesumcost{S_C}$ for staff nodes $v_i$. It remains to prove it for the first node of segments as well. The argument is essentially the same as the previous case, with some added complexity due to some segments falling before our segment of interest in sort order and some segments falling after our segment of interest in sort order.
  
  This time, let $v_i$ be the staff node preceding the appearance of our segment, and let $S_\leftarrow$ be the set of segments that come before our segment in sort order and $S_\rightarrow = S$ be the remainder of the segments (and remember that we broke ties arbitrarily but consistently). We know:
  
  \begin{align*}
    S_B &=
      \bigcup_{k=1}^{i-1} \left[ (\negativesegments[k] \cup \positivesegments[k]) \cap S_\rightarrow \right]
      \cup
      \bigcup_{k=1}^i \left[ \{v_k\} \cup ((\negativesegments[k] \cup \positivesegments[k]) \cap S_\leftarrow) \right] \\
    S_C &=
      \bigcup_{k=1}^{i-1} \left[ (\negativesegmentsstar[k] \cup \positivesegmentsstar[k]) \cap S_\rightarrow \right]
      \cup
      \bigcup_{k=1}^i \left[ \{v_k\} \cup ((\negativesegmentsstar[k] \cup \positivesegmentsstar[k]) \cap S_\leftarrow) \right] \\
    \nodesumcost{S_B} &=
      \sum_{k=1}^{i-1} \left[
        \sum_{v \in \negativesegments[k] \cap S_\rightarrow} w_v +
        \sum_{v \in \positivesegments[k] \cap S_\rightarrow} w_v
      \right] \\
    &\hphantom{{}={}} + \sum_{k=1}^i \left[
        w_{v_i} +
        \sum_{v \in \negativesegments[k] \cap S_\leftarrow} w_v +
        \sum_{v \in \positivesegments[k] \cap S_\leftarrow} w_v
      \right] \\
    \nodesumcost{S_C} &=
      \sum_{k=1}^{i-1} \left[
        \sum_{v \in \negativesegmentsstar[k] \cap S_\rightarrow} w_v +
        \sum_{v \in \positivesegmentsstar[k] \cap S_\rightarrow} w_v
      \right] \\
    &\hphantom{{}={}} + \sum_{k=1}^i \left[
        w_{v_i} +
        \sum_{v \in \negativesegmentsstar[k] \cap S_\leftarrow} w_v +
        \sum_{v \in \positivesegmentsstar[k] \cap S_\leftarrow} w_v
      \right] \\
    \nodesumcost{S_B} - \nodesumcost{S_C} &= 
      \underbrace{\sum_{k=1}^{i-1} \left[
        \sum_{v \in \negativesegments[k] \cap S_\rightarrow} w_v -
        \sum_{v \in \negativesegmentsstar[k] \cap S_\rightarrow} w_v
      \right] + \sum_{k=1}^i \left[
        \sum_{v \in \negativesegments[k] \cap S_\leftarrow} w_v -
        \sum_{v \in \negativesegmentsstar[k] \cap S_\leftarrow} w_v
      \right]}_\text{negative segments difference} \\
    &\hphantom{{}={}} +
      \underbrace{\sum_{k=1}^{i-1} \left[
        \sum_{v \in \positivesegments[k] \cap S_\rightarrow} w_v -
        \sum_{v \in \positivesegmentsstar[k] \cap S_\rightarrow} w_v
      \right] + \sum_{k=1}^i \left[
        \sum_{v \in \positivesegments[k] \cap S_\leftarrow} w_v -
        \sum_{v \in \positivesegmentsstar[k] \cap S_\leftarrow} w_v
      \right]}_\text{nonnegative segments difference}
  \end{align*}
  We finish by executing the previous argument separately for segments in $S_\leftarrow$ and $S_\rightarrow$. In the end, this final quantity is still nonpositive and we still get $\nodesumcost{S_B} \le \nodesumcost{S_C}$.
  
  All that remains is to prove the first half of the claim. We do so by induction. Initially, for each bend $\bend = (v_\ell, u_1, \dots, u_k, v_r)$, the algorithm places each of its negative segments into $\negativesegments[\ell]$ and each of its nonnegative segments into $\positivesegments[r]$. Since the schedule $\sigma^\star$ is feasible, it must also respect the two staff nodes $v_\ell, v_r$ bounding the possible placements of these segments.
  
  Now for the inductive step. We assume that the algorithm is currently in a state where both halves of the claim are satisfied, and then we consider how the algorithm attempts to move segments in either \cref{algline:move1} or \cref{algline:move2}. We show that the first half of the claim remains satisfied (which implies the second half the claim again).

  For the sake of contradiction, suppose \cref{algline:move1} attempts to move a negative segment $\negativesegment$ (and possibly subsequent negative segments from the same bend) from $\negativesegments$ to $\negativesegments[i+1]$. Since $\sigma^\star$ is a feasible schedule, it must keep all negative segments from this bend in the order they were produced by segmentation. Hence, if any negative segments in this move violate the claim, then the original one (whose peak currently exceeds $\guess$) violates the claim.
  
  In order for this to be the first time that the first half of the claim is violated, it must be that $\negativesegment \in \negativesegmentsstar$. In other words, this segment appears at the same spot in both $B$ and $C$. By our inductive hypothesis, the second half of the claim implies that that the memory used immediately before this segment in $B$ is at most the memory used immediately before this segment in $C$. Since the nodes of the segment are the same in both, if the segment exceeds $\guess$ memory usage in $B$ then it must do so in $C$ as well. But $\sigma^\star$ was defined to be a schedule that never exceeds $\guess$ memory usage, so we have a contradiction.
  
  The nonnegative segment argument is identical. Suppose \cref{algline:move2} attempts to move a positive segment $\positivesegment$ (and possibly preceding positive segments from the same bend) from $\positivesegments$ to $\positivesegments[i-1]$. We know $\sigma^\star$ is feasible so it must keep positive segments from this bend in the order they were produced by segmentation. Again, if any positive segments violate the claim, then the original one does.
  
  We can align again by observing that $\positivesegment \in \positivesegmentsstar$ for this to be the first violation. Since the segment appears at the same spot in both $B$ and $C$, we know that the memory used immediately before the segment in $B$ is at most the memory used immediately before the segment in $C$. Hence exceeding $\guess$ memory usage in $B$ means we exceed $\guess$ memory usage in $C$, which is still a contradiction.
  
  This completes the proof of the claim.
\end{proof}

\begin{algorithm}[htbp]
  \SetAlgoLined
  \KwData{Caduceus Graph $G = (V, E)$ with staff $\staff$, bends $\bends$ and node-sum weights $w_v \in \mathbb{R}$}
  \KwResult{A memory-optimal schedule $\sigma$}
  
  For $i$ from $1$ to $p - 1$, initialize the negative segments between $v_i$ and $v_{i+1}$ to be $\negativesegments = \varnothing$ and the nonnegative segments between them to be $\positivesegments = \varnothing$\;
  \For{bend $\bend \in \bends$} {
    Unpack the bend as $\bend = (v_\ell, u_1, \cdots, u_k, v_r)$\;
    Segment the nonstaff nodes $(u_1, \cdots, u_k)$ using \cref{defn:segments}\;
    Add all produced negative segments to $\negativesegments[\ell]$\;
    Add all produced nonnegative segments to $\positivesegments[r-1]$\;
  }
  \tcc{For the rest of the algorithm, our current schedule is assumed to be $(v_1, \negativesegments[1], \positivesegments[1], v_2, \negativesegments[2], \positivesegments[2], \dots)$ where the segments in each $\negativesegments, \positivesegments$ are sorted according to \cref{thm:segment-consecutive}}
  \While{true} {
    Save a checkpoint $\kappa$ of the current schedule\;
    Set $\guess$ to be the peak of the current schedule\;
    \tcc{We attempt to achieve peak memory strictly less than $\guess$.}
    \While{the peak of the current schedule is at least threshold $\guess$} {
      Choose an arbitrary peak\;
      \If{the peak occurs at a staff node $v_i$} {
        \Return checkpoint $\kappa$\;
      }
      \If{the peak occurs in a negative segment $\negativesegment$ in some $\negativesegments$} {
        Attempt to move $\negativesegment$ and any subsequent negative segments from the same bend in $\negativesegments$ from $\negativesegments$ to $\negativesegments[i+1]$. If this is not possible (either because the bend ends at $r = i + 1$ or the bend has a nonnegative segment in $\positivesegments$), then return checkpoint $\kappa$\;
      }
      \ElseIf{the peak occurs in a nonnegative segment $\positivesegment$ in some $\positivesegments$} {
        Attempt to move $\positivesegment$ and any preceding nonnegative segments from the same bend in $\positivesegments$ from $\positivesegments$ to $\positivesegments[i-1]$. If this is not possible (either because the bend begins at $\ell = i$ or the bend has a negative segment in $\negativesegments$), then return checkpoint $\kappa$\;
      }
    }
  }
  \caption{$\textsc{OptimalCaduceus}(G, w)$}
  \label{alg:caduceus2}
\end{algorithm}

With \cref{claim:greedy} proven, finishing the proof of \cref{thm:caduceus} is straightforward. We first prove correctness. Since the while loop on \cref{algline:while} only terminates when the schedule uses at most $\guess$ peak memory, the algorithm never returns such a schedule when none exists. Additionally, the algorithm is guaranteed to return such a schedule when one exists; \cref{claim:greedy} implies that if such a schedule exists then (i) all of algorithm's attempted moves are safe and (ii) the memory immediately before a staff node is at most the memory before the same staff node in any feasible schedule $\sigma^\star$ that uses at most $\guess$ memory, so the if-clause on \cref{algline:staff} will not trigger.

We conclude by accounting for running time. Since segmentation via \cref{defn:segments} takes linear time, segmenting all bends can be done in $O(n)$ time. In the rest of the algorithm, the total number of segment moves is bounded by $O(np)$; there are at most $O(n)$ segments (since each is associated with at least one node and they do not share nodes) and each moves at most $p-2$ times. We may have to locate the peak $O(np)$ times to do so. This can be done naively by performing a linear scan each time (for $O(n^2p)$ runtime), but can be sped up by creating a max-segment-tree over the pseudo-schedule $A$, which improves the runtime to $O(np \log n)$.

The final factor in the runtime is the issue of binary searching over the guess $\theta$. Naively, this incurs an additional $O(\log (nW))$, where $W$ is the largest node weight (and would require the weights to be integer). However, we can avoid this entirely by restructuring our algorithm into \cref{alg:caduceus2}. Our algorithm now brute-forces over possible $\guess$ values from high to low. Importantly, it does not need to guess that many values, since improving the guess is only possible by moving a segment, which we have already argued only occurs $O(np)$ times. In fact, since the algorithm avoids resetting the current schedule, the total number of segment moves over all values of $\guess$ is still $O(np)$ and again a max-segment-tree lets us achieve only a $O(np \log n)$ runtime. This completes the proof of \cref{thm:caduceus}.
\end{proof} 

\section{Pumpkins and Trees}
\label{sec:pumpkin}
In this section, we study the simpler case of pumpkin graphs, which are a special case of SP-graphs. These are called parallel-chain graphs by Kayaaslan et al.~\cite{kayaaslan2018scheduling}.

\begin{definition}[Pumpkin graph]
\label{def:pumpkin}
A chain is a directed path $(s,v_1,\dots,v_{l},t)$ with distinct nodes $(l\ge 0)$.
   A \emph{pumpkin graph} is an SP-graph formed by the parallel composition of $d$ ($d\ge 1$) chains $b_1,b_2,\dots,b_d$ sharing the source $s$ and sink $t$. Each chain $b_i$ is called a branch of the pumpkin graph.
\end{definition}

We will prove the following theorem. Recall the notion of internal memory profile (denoted as $\internalprofile{\cdot}$) from \cref{defn:internalmemoryprofile}, which is defined on single-source single-sink DAGs (which include pumpkin graphs).

\begin{restatable}{theorem}{mainpumpkin}
    \label{thm:main-pumpkin}
   Given a pumpkin graph $G$ with $d$ branches in the \truemem\  model, there is an $O(2^d dn)$-time algorithm that finds a schedule $\sigma$ such that  $\internalprofile{\sigma}\ledot \internalprofile{\tau}$ for all possible schedules $\tau$ of $G$.
 In particular, it finds a schedule that minimizes the peak memory. 
\end{restatable}

We remark that our algorithm for $d$-branch pumpkin graphs is faster than that for general SP-graphs of maximum out-degree $d$ described in \cref{sec:sp}. 
Our algorithm for pumpkin graphs is based on dynamic programming over all the $2^d$ subsets of the $d$ branches, and the correctness of the dynamic program crucially relies on the existence of a dominant schedule for each subproblem.

\begin{proof}[Proof of \cref{thm:main-pumpkin}]
Let $s,t$ be the source and sink of the pumpkin $G$, and $b_1, b_2, \dots , b_d$ be its $d$ branches.  

We first deal with the degenerate case where some branch is a single $(s,t)$ edge. In this case, for any schedule $\sigma$ of $G$, in $\internalprofile{\sigma}$ the the memory  of $s$ is never released until $t$ is finished in the end, so $w_s$ should contribute throughout the profile (except the last entry $w_t$).
Hence, by shifinting the profile by $w_s$, the task becomes equivalent to the following: without  loss of generality assume $w_s = 0$,  and find  a schedule $\sigma$ of $G$ that has the dominant $\internalprofile{\sigma}$ (excluding the final $w_t$ entry).   This is the problem of scheduling an in-tree (since $s$ has no effect on the memory profile now), which was already solved by Liu \cite{liu1987application} in linear time.\footnote{In more details, Liu \cite{liu1987application} works with the edge-weighted model where memory cost is associated with edges instead of nodes. However, for in-trees, this edge-weighted model is equivalent to our \truemem\  model, since we can move the weight $w_v$ of node $v$ to the edge connecting $v$ and its parent in the tree.} 
Hence, in the following, we assume this degenerate case does not happen, and that each branch of the pumpkin has at least one internal node.

For each non-empty subset $S \subseteq \{b_1, b_2, \dots , b_d\}$, define $G|_S$ as the pumpkin graph with only branches in $S$.  For a schedule $\sigma = (\sigma_1=s,\sigma_2,\sigma_3,\dots,\sigma_{n'-1},\sigma_{n'}=t)$ of $G|_S$, we focus on the internal memory profile of $\sigma$ on $G|_S$, denoted $\internalprofile{\sigma}$. Recall by definition that $\internalprofile{\sigma}$ starts with $w_s$ and ends at $w_t$.
We define $OPT[S]$ to be the schedule $\sigma$ of $G|_S$ that has a dominant internal memory profile.
Such schedule must exist, due to \cref{thm:existsdominant}.

We use a dynamic programming (DP) to compute $OPT[S]$ for all non-empty $S\subseteq \{b_1,\dots,b_d\}$ (in non-decreasing order of $|S|$), and the final answer is $OPT[\{b_1,b_2,\dots,b_d\}]$, where we use $\circ$ to denote concatenation.  To compute $OPT[S]$ by dynamic programming, there are two cases:
\begin{itemize}
    \item $|S|=1$. This case has a single branch, and $OPT[S]$ is simply the unique schedule of this branch.
Observe that in the internal memory profile of this schedule, the memory cost of  node $s$ is released once the first node from branch $b_i$ is finished.  

    \item $|S|\ge 2$. 
  In this case,  note that in the memory profile of any schedule of $G|_S$, the memory cost of node $s$ is released once all the first nodes from all the branches from $S$ are finished. 
    
    We guess which branch from $S$ goes first, that is, we guess that the node scheduled right after $s$ comes from branch $b_i$ for some $i\in S$. 
    Let branch $b_i$ be decomposed as $s \circ b_i^{(1)} \circ b_i^{(+)} \circ t$, where $b_i^{(1)}$ denotes the first (internal) node on branch $b_i$.
    The purpose of this guess is so that the releasing of node $s$ is no longer dependent on its out neighbor $b_i^{(1)}$, but only dependent on the remaining out neighbors of $s$ (because $b_i^{(1)}$ finishes earlier than them).  This observation allows us to use the optimal substructure of the subproblem  $OPT[S \setminus b_i]$.
    Now our plan is to optimally interleave $b_i^{(+)}$ with the dominant schedule of $OPT[S \setminus b_i]$, and the releasing of node $s$ is coveniently taken care of by the memory profile of $OPT[S \setminus b_i]$.

Let schedule    $OPT[S \setminus b_i]$ be decomposed as 
\[OPT[S \setminus b_i] = s\circ OPT^{-}[S \setminus b_i] \circ t. \]
    
    Then, we have the following candidate schedule for $G|_S$:
   \[ s \circ  b_i^{(1)} \circ merge(OPT^-[S \setminus b_i], b_i^{(+)}) \circ t.\]
   where $merge$ is implemented by the merge operation from \cref{thm:segment-merge} in $O(n)$ time.
   
   Finally, we take the dominant one among candidate schedules (defined above) over all $i\in S$. Such dominant schedule must exist due to \cref{thm:existsdominant}.
\end{itemize}

The overall running time of this dynamic programming algorithm is $O(2^d dn)$.
\end{proof}

Next, we show an algorithm for scheduling out-trees with bounded out-degrees, in the \truemem\  model.

\begin{restatable}{theorem}{mainouttree}
    \label{thm:main-out-tree}
   Given an out-tree $G$ with maximum out-degree $d$ in the \truemem\  model, there is an $O(2^d poly(n))$-time algorithm that finds a schedule $\sigma$ of $G$ that  minimizes the peak memory. 
\end{restatable}
\begin{proof}
For convenience, we can complete this out-tree $G$ into an SP graph without affecting the minimum possible peak memory, by the following tree-gluing operation: let $G'$ be a copy of $G$, but with all node-weights set to zero, and all edges reversed (so $G'$ is an in-tree). Then, for each leaf $v$ of the tree $G$, we connect an edge from $v$ in $G$ to the copy of $v$ in $G'$. 

It is clear that this new graph is a single-source single-sink SP graph, and this transformation does not affect the optimal solution: the new nodes in $G'$ do not contribute to the memory cost, and once a leaf node from $G$ finishes we can immediately finish its successor (its own copy in $G'$) and release it.

Now we can use the linearization technique \cref{subsec:linearization} to solve this SP-graph. Unlike the case of general SP-graphs treated in \cref{sec:sp}, our SP-graph here has a simpler structure: every time we attempt to linearize a isolated subgraph (\cref{defn:isolated}) (assuming it is the smallest isolated subgraph in the graph), we know that this isolated subgraph must be a pumpkin (which is a much simpler special case of the starborescence from \cref{sec:sp}).  This is due to the fact that this SP-graph resulted from gluing two identical trees. Hence, we can directly invoke the pumpkin scheduler (\cref{thm:main-pumpkin}) to perform  linearization. The number of branches in this pumpkin does not exceed the maximum out-degree $d$ of the original tree. Hence, the entire linearization procedure takes time $O(2^d poly(n))$.
\end{proof} 

\section{Hardness}
\label{sec:hardness}

In this section, we show that the weighted one-shot black pebbling problem is strongly NP-hard even when the computation graph is a \emph{pumpkin graph}. As a corollary, we note that the same reduction also shows strong NP-hardness on \emph{out trees}.
We will reduce from the 3-Partition problem, which previous work has shown to be strongly NP-hard \cite{garey1979computers}.

\begin{definition}[3-Partition]
  Let $a_1, \ldots, a_n$ be $n$ positive integers where $n = 3m$. The 3-Partition problem is to determine whether there is a partition of the numbers into $m$ disjoint groups such that the sum of each group is exactly equal. \footnote{The 3-Partition problem usually also requires each group to contain exactly 3 numbers, but the unrestricted version is also strongly NP-hard.}
\end{definition}

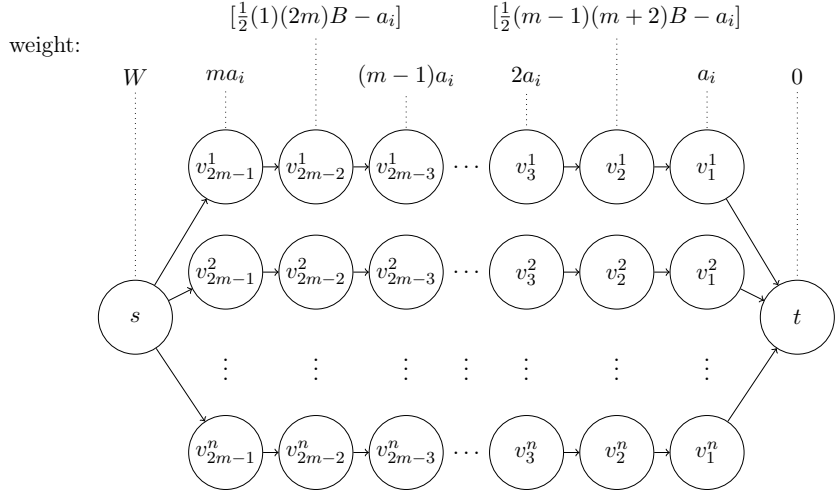
\begin{figure}
\centering
\scalebox{.8}{
\begin{tikzpicture}[%
  auto,
  scale=1.0,
  hollownode/.style={
    circle,
    draw=black,
    inner sep=1pt,
    minimum size=35pt,
  },
  ]
  \node[hollownode]   (s) at ( 0, 0) {$s$};

  \node[hollownode] (v11) at (1.5,   2.5) {$v^1_{2m-1}$};
  \node[hollownode] (v21) at (1.5,  0.75) {$v^2_{2m-1}$};
  \node             (v31) at (1.5, -0.75) {$\vdots$};
  \node[hollownode] (v41) at (1.5, -2.25) {$v^n_{2m-1}$};

  \node[hollownode] (v12) at (  3,   2.5) {$v^1_{2m-2}$};
  \node[hollownode] (v22) at (  3,  0.75) {$v^2_{2m-2}$};
  \node             (v32) at (  3, -0.75) {$\vdots$};
  \node[hollownode] (v42) at (  3, -2.25) {$v^n_{2m-2}$};

  \node[hollownode] (v13) at (4.5,   2.5) {$v^1_{2m-3}$};
  \node[hollownode] (v23) at (4.5,  0.75) {$v^2_{2m-3}$};
  \node             (v33) at (4.5, -0.75) {$\vdots$};
  \node[hollownode] (v43) at (4.5, -2.25) {$v^n_{2m-3}$};

  \node             (v14) at (5.5,   2.5) {$\hdots$};
  \node             (v24) at (5.5,  0.75) {$\hdots$};
  \node             (v34) at (5.5, -0.75) {$\vdots$};
  \node             (v44) at (5.5, -2.25) {$\hdots$};

  \node[hollownode] (v15) at (6.5,   2.5) {$v^1_3$};
  \node[hollownode] (v25) at (6.5,  0.75) {$v^2_3$};
  \node             (v35) at (6.5, -0.75) {$\vdots$};
  \node[hollownode] (v45) at (6.5, -2.25) {$v^n_3$};

  \node[hollownode] (v16) at (  8,   2.5) {$v^1_2$};
  \node[hollownode] (v26) at (  8,  0.75) {$v^2_2$};
  \node             (v36) at (  8, -0.75) {$\vdots$};
  \node[hollownode] (v46) at (  8, -2.25) {$v^n_2$};

  \node[hollownode] (v17) at (9.5,   2.5) {$v^1_1$};
  \node[hollownode] (v27) at (9.5,  0.75) {$v^2_1$};
  \node             (v37) at (9.5, -0.75) {$\vdots$};
  \node[hollownode] (v47) at (9.5, -2.25) {$v^n_1$};

  \node[hollownode] (t) at (11, 0) {$t$};
  
  \draw[->] (s) edge (v11) (v11) edge (v12) (v12) edge (v13) (v15) edge (v16) (v16) edge (v17) (v17) edge (t);
  \draw[->] (s) edge (v21) (v21) edge (v22) (v22) edge (v23) (v25) edge (v26) (v26) edge (v27) (v27) edge (t);
  \draw[->] (s) edge (v41) (v41) edge (v42) (v42) edge (v43) (v45) edge (v46) (v46) edge (v47) (v47) edge (t);

  \node       at (-1.5, 4.5) {weight:};
  \node (v00) at (   0,   4) {$W$};
  \node (v01) at ( 1.5,   4) {$m a_i$};
  \node (v02) at (   3,   5) {$[\frac12(1)(2m)B - a_i]$};
  \node (v03) at ( 4.5,   4) {$(m-1) a_i$};
  \node (v05) at ( 6.5,   4) {$2 a_i$};
  \node (v06) at (   8,   5) {$[\frac12(m-1)(m+2)B - a_i]$};
  \node (v07) at ( 9.5,   4) {$a_i$};
  \node (v08) at (  11,   4) {$0$};

  \draw[dotted] (v00) -- (s);
  \foreach \x in {1, 2, 3, 5, 6, 7}
    \draw[dotted] (v0\x) -- (v1\x);
  \draw[dotted] (v08) -- (t);
\end{tikzpicture}
}
\caption{Reduction from 3-Partition to Weighted One-Shot Black Pebbling on a Pumpkin.}
\label{fig:tikz-hardness}
\end{figure}

\pumpkinhardness

\subsection*{Reduction}
Suppose we have a 3-Partition instance $a_1, \ldots, a_n$. Let $B = \frac1m \sum_{i=1}^n a_i$. We are therefore looking for $m$ disjoint groups that each sum to $B$.

The reduction plan is as follows. We will construct a pumpkin with $n$ branches, one corresponding to each $a_i$. As we present the formal reduction, it may be useful to examine \Cref{fig:tikz-hardness}, which depicts this pumpkin. Our pumpkin has $n$ branches; the $i^{th}$ branch corresponding to 3-Partition element $a_i$. Each branch has $2m - 1$ nodes, excluding the shared source $s$ and sink $t$ which we use to produce a memory profile with $m$ valleys and $m-1$ peaks. The precise weights of each node are as follows:
\begin{enumerate}
  \item The shared source node $s$ is very large and has $\sz{s} = W$ that we will determine later.
  \item The shared sink node $t$ has $\sz{t} = 0$.
  \item The $i^{th}$ branch has $2m-1$ internal nodes labeled with reversed-indices $v^i_{2m-1}, \ldots, v^i_1$. Nodes $v^i_k$ with an odd index $k$ have $\sz{v^i_k} = \frac{k+1}{2} a_i$ and will correspond to valleys. Nodes $v^i_k$ with an even index $k$ have $\sz{v^i_k} = \frac12 (m - \frac{k}2) (m + \frac{k}2 + 1) B - a_i$.
\end{enumerate}

For this to be a valid problem, we need to ensure that all weights are nonnegative. This trivially holds for the source $s$, the sink $t$, and the branch nodes with odd $k$ since 3-Partition involves positive integers by definition. To show it for the branch nodes with even $k$, we observe that the original 3-Partition instance cannot possibly be solvable unless all elements satisfy $a_i \le B$ (if there is a violating element, the group it is in will always sum to more than $B$ since elements are positive). The smallest weight node in branch $i$ with even $k$ is the earliest one, $k = 2m - 2$, which has $\sz{v^i_{2m-2}} = mB - a_i \ge B - a_i \ge 0$.

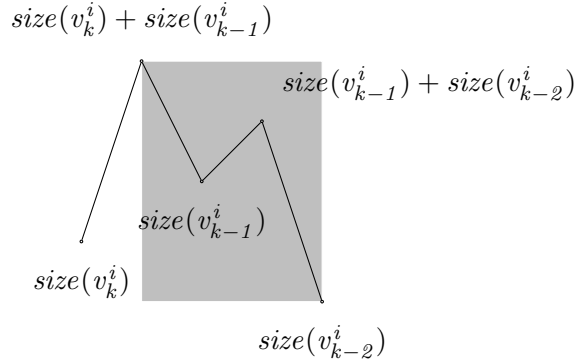
\begin{figure}
\centering
\begin{tikzpicture}[%
  auto,
  scale=0.8,
  void/.style={
    circle,
    draw=black,
    inner sep=0pt,
  },
  ]
  \draw[color=white, fill=black!25] (1, 5) rectangle (4, 1);
  
  \node[void] (v1) at (0, 2) {}; \node [below=0.15cm of v1] {$\sz{v^i_k}$};
  \node[void] (p1) at (1, 5) {}; \node [above=0.15cm of p1] {$\sz{v^i_k} + \sz{v^i_{k-1}}$};
  \node[void] (v2) at (2, 3) {}; \node [below=0.15cm of v2] {$\sz{v^i_{k-1}}$};
  \node[void] (p2) at (3, 4) {}; \node [above right=0.15cm of p2] {$\sz{v^i_{k-1}} + \sz{v^i_{k-2}}$};
  \node[void] (v3) at (4, 1) {}; \node [below=0.15cm of v3] {$\sz{v^i_{k-2}}$};

  \draw (v1) -- (p1) -- (v2) -- (p2) -- (v3);
\end{tikzpicture}
\caption{Memory profile when moving from a branch node with odd index $k$ to the next branch node with odd index $k-2$.}
\label{fig:tikz-hardness-segment}
\end{figure}

Next, we do some preliminary reasoning about the memory profile generated by branch $i$. In particular, let us consider starting with a pebble on branch node $v^i_k$ with odd index $k \in {3, 5, ..., 2m-1}$, and advancing to having a pebble on the next odd index. This involves five memory states:
\begin{enumerate}[label={(\arabic*)}] \itemsep0em
  \item Initially, we have a pebble on $v^i_k$ for $\sz{v^i_k}$ memory usage.
  \item Next, we add a pebble to $v^i_{k-1}$ for $\sz{v^i_k} + \sz{v^i_{k-1}}$ memory usage.
  \item Next, we remove a pebble from $v^i_k$ for $\sz{v^i_{k-1}}$ memory usage.
  \item Next, we add a pebble to $v^i_{k-2}$ for $\sz{v^i_{k-1}} + \sz{v^i_{k-2}}$ memory usage.
  \item Next, we remove a pebble from $v^i_{k-1}$ for $\sz{v^i_{k-2}}$ memory usage.
\end{enumerate}

We claim that these five memory states fall into the following relationships, which are depicted in \Cref{fig:tikz-hardness-segment}.

\begin{claim}
\label{clm:hardness0}
  The key memory states are (1), (2), and (5), whose memory usages falls into the relationship $(5) \le (1) \le (2)$. Furthermore, the memory usages of (3) and (4) are both between the memory usages of (2) and (5).
\end{claim}

We first prove that the memory usages of (1), (2), and (5) are ordered $(5) \le (1) \le (2)$. The first comparison is essentially by construction: $\sz{v^i_{k-2}} = \frac{k-1}{2} a_i < \frac{k+1}{2} a_i = \sz{v^i_k}$. The second comparison relies on $\sz{v^i_{k-1}} \ge 0$ which we have already proved.

Next, we prove that the memory usages of (3) and (4) are both between the memory usages of (2) and (5). The easy half is reasoning about the memory usage of (4); it is at least the memory usage of (5) because $\sz{v^i_{k-1}} \ge 0$ and it is at most the memory usage of (2) because $\sz{v^i_{k-2}} < \sz{v^i_k}$, both of which we just reasoned about to order $(5) \le (1) \le (2)$. The harder half is reasoning about the memory usage of (3). It is at least the memory usage of (5) because:
\begin{align*}
  \sz{v^i_{k-1}} &= \frac12 \left(m - \frac{k-1}{2}\right) \left(m + \frac{k-1}{2} + 1\right) B - a_i \\
                 &\ge mB - a_i \qquad \text{(worst case is $k = 2m-1$)} \\
                 &\ge ma_i - a_i = (m-1) a_i \\
                 &\ge \frac{k-1}{2} a_i \qquad \text{(worst case is $k = 2m-1$)} \\
                 &= \sz{v^i_{k-2}}
\end{align*}
It is at most the memory usage of (2) because trivially $\sz{v^i_k} \ge 0$.

The upshot of all this reasoning is that all valleys in a branch occur when a pebble is on a node with odd index $k$ and all peaks in a branch occur when pebbles are on a node with odd index $k$ and the subsequent node with even index $k-1$. Furthermore, the valleys decrease from $m a_i$ to $a_i$ while the peaks increase from $mB + (m-1) a_i$ to $\frac12 (m-1)(m+2)B + a_i$ (each subsequent peak gains at least one multiple of $B$ while losing exactly one multiple of $a_i$). When the branch $i$ is clear from context, we drop the superscript and refer to nodes $v_k$ of branch $i$.

Why did we go through the trouble of setting up this careful valley and peak structure for every branch/element? At a high level, the tension in our pebbling game is that we want to remove the pebble from the heavyweight source $s$. However, in order to do so, we must be somewhere in the middle of every branch since we can never re-add the pebble to $s$ again. We are incentivized to get as far along each branch as possible because the valleys decrease as we do so, but this is counterbalanced by the pressure of the increasing peaks. In the end, we will show that the best strategies correspond to taking a group of branches to their final valley, then another group of branches to their second-to-last valley, and so on until we take a group of branches to their first valley. For such a strategy, its overall performance will rest on ensuring all of these groups have the same sum (namely $B$), because any deviation from this sum will cause a part of this strategy to exceed the desired pebbling cost.

With this in mind, we will choose $W$ to be $\frac12 m(m+1)n B + mB$, which is at least than the sum of the final peaks from every branch in the construction. This will ensure that once we manage to remove the pebble from $s$, the remainder of the pebbling strategy is irrelevant since it will always be less than the initial cost of placing a pebble on $s$.

\begin{claim}
\label{clm:hardness1}
  If the 3-partition instance is feasible, then there exists a strategy with pebbling cost at most $W + \dfrac{1}{2} m(m+1) B$.
\end{claim}

\begin{proof}
  Suppose the input 3-partition instance is feasible and let $S_1, S_2, \ldots, S_m$ denote a feasible partition. Then we have $\sum_{a_i \in S_j} a_i = B$ for each $j \in [m]$. For notational convenience, we will let $S^{-1}(i)$ denote the $j$ such that $i \in S_j$ and without loss of generality reindex the elements so that all elements of $S_1$ come before all elements of $S_2$ and so on, i.e. $S^{-1}(\cdot)$ is an increasing function.
    
  Consider the strategy $\sigma$ constructed as follows. First, schedule the source node $s$, then schedule each branch corresponding to integers in $S_1$ up to their last odd index node $v_1$, then schedule each branch corresponding to integers in $S_2$ up to their second-to-last odd node $v_3$, and so on until we schedule each branch corresponding to integers in $S_m$ up to their first odd node $v_{2m-1}$. Formally, the schedule $\sigma \triangleq (s) \circ \sigma^1_1 \circ \cdots \circ \sigma^i_{2\cdot S^{-1}(i) - 1} \circ \cdots \circ \sigma^n_{2m-1} \circ *$ where $\sigma^i_k = (v^i_{2m-1},  \ldots, v^i_k)$ is the subschedule that executes branch $i$ until the node with (odd) index $k$ and $*$ indicates any feasible subschedule for the remaining nodes of the graph. Note that this schedule implies a corresponding pebbling strategy where we remove a pebble from a node as soon as all of its immediate successors have had a pebble at some point.

  Since the peaks are increasing in each branch, the pebbling cost while executing $\sigma^i_{2j - 1}$ (where $i \in S_j$) occurs in the peak right before arriving at the last node $v^i_{2j - 1}$. By \Cref{clm:hardness0}, this peak involves exactly the previous odd and even nodes: $v^i_{2j + 1}$ and $v^i_{2j}$. Additionally, we will have pebbles on $s$ and the previous branches $1, \ldots, i-1$. Hence the total pebbling cost is
  \begin{align*}
    &\hphantom{{}={}} \sz{s} +
                      \left( \sum_{j' < j} \sum_{i' \in S_{j'}} \sz{v^{i'}_{2j' - 1}} \right) +
                      \left( \sum_{i' \in S_j . i' < i} \sz{v^{i'}_{2j - 1}} \right) +
                      \sz{v^i_{2j + 1}} + \sz{v^i_{2j}} \\
    &= W +
       \left( \sum_{j' < j} \sum_{i' \in S_{j'}} j' a_{i'} \right) +
       \left( \sum_{i' \in S_j . i' < i} j a_{i'} \right) +
       \left( (j+1) a_i \right) +
       \left( \frac12 (m - j)(m + j + 1)B - a_i \right) \\
    &= W +
       \left( \sum_{j' < j} \sum_{i' \in S_{j'}} j' a_{i'} \right) +
       \left( \sum_{i' \in S_j . i' < i} j a_{i'} \right) +
       \left( j a_i \right) +
       \left( \frac12 (m - j)(m + j + 1)B \right) \\
    &= W +
       \left( \sum_{j' < j} \sum_{i' \in S_{j'}} j' a_{i'} \right) +
       \left( \sum_{i' \in S_j . i' \le i} j a_{i'} \right) +
       \left( \frac12 (m - j)(m + j + 1)B \right) \\
    &\le W +
       \left( \sum_{j' \le j} \sum_{i' \in S_{j'}} j' a_{i'} \right) +
       \left( \frac12 (m - j)(m + j + 1)B \right) \\
    &= W +
       \left( \sum_{j' \le j} j' B \right) +
       \left( \frac12 (m - j)(m + j + 1)B \right) \\
    &= W +
       \frac12 j(j+1) B +
       \frac12 (m - j)(m + j + 1)B \\
    &= W + \frac12 (m^2 + m) B
  \end{align*}
  Finally, once node $v^n_{2m-1}$ has been scheduled, the source node $s$ can have its pebble removed. We chose the weight of $s$ to be sufficiently large so that the remaining nodes of the graph can be arbitrarily scheduled without increasing the memory consumption beyond $W$. This completes the proof.
\end{proof}

\begin{claim}
\label{clm:hardness2}
  If there exists a strategy with pebbling cost at most $W + \dfrac{1}{2}m(m+1)B$, then there exists a feasible solution to the 3-partition instance.
\end{claim}

\begin{proof}
  Consider a strategy $\sigma$ with pebbling cost at most  $W + \dfrac{1}{2}m(m+1)B$. Without loss of generality, $\sigma$ executes each branch $i$ up to some node $v^i_k$ with odd index $k$ until the source node is released since interleaving the branches does not help to reduce memory consumption (find the first instance where we return to some branch $i$, observe that we could have done that when we initially made progress on branch $i$ because there were only fewer branches then and for those other branches we could have been sitting at a lower valley).

  We next claim that without loss of generality, the strategy first executes some set of branches up to node $v_1$, then executes some set of branches up to node $v_3$, and so on. Suppose not, then consider the first time that we make progress on some branch $j$ up to $v^j_{2k-1}$ such that we make progress on the next branch $j+1$ up to $v^{j+1}_{2k'-1}$ for some $k' < k$. Consider a new strategy $\sigma'$ that's exactly the same as the earlier one, but now executes branch $j$ also up to $v^j_{2k'-1}$. Consider the pebbling cost in $\sigma$ at the peak before node $v^{j+1}_{2k'-1}$: we have $W + P + k \cdot a_j + \frac12 (m-k')(m+k'+1)B + k'a_{j+1}$ where $P$ is the total size contributed by the valleys of all other branches (apart from $j$). On the other hand, in strategy $\sigma'$, consider the pebbling cost at the peak before executing node $v^j_{2k'-1}$: we have $W + P + \frac12 (m-k')(m+k'+1)B + k'a_{j}$ which is less than the pebbling cost in $\sigma$ at the peak before node $v^{j+1}_{2k'-1}$ so it must fit in our budget of $W + \frac12 m (m+1) B$ as well. Repeating this argument, we can assume that the strategy $\sigma$ has the assumed form.

  Let $S_1, S_2, \ldots, S_m$ be the sets of branches that are executed up to nodes $v_1, v_3, \ldots v_{2m-1}$ respectively. Let $a(S_k) \triangleq \sum_{j \in S_k} a_j$ for convenience. We focus on the pebbling cost attained during the last branch in the first set $S_1$, at the peak right before the valley it stops at. At that peak, we need to pay for the source ($W$), the valleys of all previous branches, and the peak itself. The total cost incurred is the following expression.
  \begin{align*}
    W + \dfrac{1}{2}(m-k)(m+k+1)B + \sum_{k' \leq k}k'a(S_{k'})
  \end{align*}
  But since the pebbling cost of the strategy is at most $W + \dfrac{1}{2}(m)(m+1)B$, we have
  \begin{align*}
    &\hphantom{{}\Rightarrow{}} W + \dfrac{1}{2}(m-k)(m+k+1)B + \sum_{k' \leq k}k'a(S_{k'}) \le W + \frac12 (m)(m+1)B\\
    &\Rightarrow \dfrac{1}{2}(m-1)(m+2)B + a(S_{1}) \leq \dfrac{1}{2}(m)(m+1)B\\
    &\Rightarrow a(S_1) \leq B
  \intertext{Similarly, by examining the other sets $S_k$, we get the following sets of inequalities:}
  \sum_{i=1}^k i a(S_i) &\leq \dfrac{1}{2}k(k+1)B, \ \forall 1\leq k\leq m
  \intertext{Further, by definition we also have:}
  \sum_{i=1}^m a(S_i) &= mB
  \end{align*}
  Setting $a(S_i) = B,\ \forall i$ is the unique solution and hence the 3-partition instance is feasible.
\end{proof}

Claims \ref{clm:hardness1} and \ref{clm:hardness2} together complete the proof of \Cref{thm:pumpkin-hardness}.
In fact, we can utilize a small modification of the same reduction to show that the problem remains strongly NP-hard even when the underlying computation graph is an out-tree.

\treehardness

\begin{proof}[Proof Sketch]
Consider the same reduction as above with just one change. For each branch $i$, the last internal node on the branch $(v^i_1)$ connects to a node $t_i$ (instead of a common sink node $t$). We set $\sz{t_i} = W$ for all $i$. The constructed graph is thus a very simple out-tree with only a single node that has out-degree greater than 1.

If the input 3-Partition instance is feasible, then we redo \Cref{clm:hardness1} by being more careful in the subschedule $*$ to exit one branch at a time. On the other hand, if we have a feasible strategy, then we patch \Cref{clm:hardness2} by observing that since $\scr{t_i} = W > \frac12 m(m+1)B$, no strategy can execute any node in $\{t_i\}$ until the large source node $s$ has been released. With this observation, the same proof as Claim \ref{clm:hardness2} continues to hold.
\end{proof} 

\bibliographystyle{alpha}
\nocite{*}  
\bibliography{scheduling}

\appendix
\section{Dominance (Deferred Proofs from \cref{sec:dominant-def})}
\label{sec:missing-proofs}

\pointerlemma*
\begin{proof}
~
 \paragraph*{$\Leftarrow$ direction:} Suppose there is a pointer-advancement sequence $\{(a_k, b_k)\}$ as in the lemma statement. Fix any $i \in [|A|]$. Since $a_1 = 1$, $a_{|A| + |B| - 1} = |A|$, and $a_i$ only increments by at most one at each step, there is some $k$ so that $a_k = i$. We claim $j = b_k$ satisfies the properties required by \cref{defn:algdef}.

   To show the first claim $\max_{i' \ge i} A[i'] \le \max_{j' \ge j} B[j']$, consider the index $i'\ge i$ that attains the maximum on the left-hand side. 
  Then, in the pointer-advancement sequence, there is some $k'\ge k$ such that $a_{k'} = i'$, and we know $b_{k'}\ge b_k = j$. Then the claim follows from 
  \[ A[i'] = A[a_{k'}] \le B[b_{k'}] \le \max_{j' \ge b_k}B[j'].\]
  The third claim  $\max_{i' \le i} A[i'] \le \max_{j' \le j} B[j']$ can be proved analogously.
  
   To show the second claim $\min_{i' \ge i} A[i'] \le \min_{j' \ge j} B[j']$, the proof is symmetric to above. Consider the index $j'\ge j$ that attains the minimum on the right-hand side.  Then, in the pointer-advancement sequence, there is some $k'\ge k$ such that $b_{k'} = j'$, and we know $a_{k'}\ge a_k = i$. Then the claim follows from 
  \[ B[j'] = B[b_{k'}] \ge A[a_{k'}] \ge \min_{i' \ge a_k}A[j'].\]
  The fourth claim  
  $\min_{i' \le i} A[i'] \le \min_{j' \le j} B[j']$  can be proved analogously.
  
  
  \paragraph*{$\Rightarrow$ direction:}  Suppose $A$ and $B$ satisfy \cref{defn:algdef}. 
  
  Let $i_1\le i_2$ be the leftmost and rightmost indices such that $A[i_1]=A[i_2] = \min_{i}A[i]$.
  Let $j_1\le j_2$ be the leftmost and rightmost indices such that $B[j_1]=B[j_2] = \max_{j}B[j]$.
  
  Note that the properties from \cref{defn:algdef} easily imply $\min_{i}A[i]\le \min_{j}B[j] \le \max_{j}B[j]$, and also $\max_{j}B[j] \ge \max_i A[i]$.
  
  We would like to construct a pointer-advancement process, where pointer $i$ initially points to $A[1]$ and pointer $j$ initially points to $B[1]$, and we want to advance them so that eventually they point to $A[|A|]$ and $B[|B|]$ respectively. We will do this in four stages:
  \begin{itemize}
      \item From $(i,j)=(1,1)$ to $(i,j)=(i_1,j_1)$.
      \item From $(i,j)=(i_1,j_1)$ to $(i,j)=(i_1,j_2)$.
      \item From $(i,j)=(i_1,j_2)$ to $(i,j)=(i_2,j_2)$.
      \item From $(i,j)=(i_2,j_2)$ to $(i,j)=(|A|,|B|)$.
  \end{itemize}
  Note that the second stage is obviously valid: since $i$ stays at $A[i_1] = \min_{i}A[i] \le \min_j B[j]$, we know that $i$ never points to a higher value than $j$ throughout this stage. Similarly, the third stage is also valid, since $j$ stays at $B[j_2] =\max_{j}B[j] \ge \max_i A[i]$. It remains to find ways to perform the first stage and the fourth stage. In the following, we explain the fourth stage, and the first stage follows from a symmetric proof.
  
Starting from $(i,j) = (i_2,j_2)$, we perform the pointer advancement using a while loop, and we maintain the following invariant at the begining/end of every loop iteration:
 \[  \min_{i'\ge i}A[i'] = A[i] \le \min_{j'\ge j}B[j'], \text{ and }  \max_{j' \ge j}B[j'] = B[j] \ge \max_{i'\ge i}A[i'].\]
 Then, each loop iteration runs as follows (it terminates until either $i=|A|$ or $j=|B|$): 
 \begin{itemize}
     \item Define $i_h = \arg \max_{i'\ge i} A[i']$, and $i_v = \arg \min_{i'\ge i_h} A[i']$. To break ties, we always pick $i'$ as large as possible.
     \item Define $j_v = \arg \min_{j'\ge j} B[j']$, and $j_h = \arg \max_{j'\ge j_v} B[j']$. To break ties, we always pick $j'$ as large as possible.
     \item If $A[i_h]>B[j_h]$, then advance pointer $i$ to $i_v$.
     \item Otherwise, advance pointer $j$ to $j_h$.
 \end{itemize}
 First, we observe the pointer movement is always valid. This is due to the loop invariant: $B[j] \ge \max_{i'\ge i} A[i']$, so moving $i$ while holding $j$ still will not violate the movement requirement $A[i]\le B[j]$. Similarly, moving $j$ while holding $i$ still also does not violate the requirement, due to $A[i] \le \min_{j'\ge j}B[j']$.
 
 Now we show the loop invariant is preserved after the movement.  We separately consider the two cases:
 \begin{itemize}
     \item Suppose $A[i_h]>B[j_h]$, and by definition we should advance $i$ to $i_v$. 
     
     We first claim $A[i_v] \le B[j_v]$ in this case. To show this, we apply \cref{defn:algdef} to $i_h$, and obtain $j^*$ such that $\max_{j'\ge j^*}B[j'] \ge \max_{i'\ge i_h}A[i']$ and $\min_{j'\ge j^*}B[j'] \ge \min_{i'\ge i_h}A[i']$.
      By definition of $i_h$, the former inequality gives $\max_{j'\ge j^*}B[j'] \ge \max_{i'\ge i_h}A[i'] =  A[i_h]>B[j_h]$, which then implies $j^*<j_v$ (due to the definitions of $j_h\ge j_v$). Then, the latter inequality gives $A[i_v] = \min_{i'\ge i_h}A[i'] \le \min_{j'\ge j^*}B[j'] = B[j_v]$.
     
     After advancing $i$ to $i_v$, the second invariant $\max_{j' \ge j}B[j'] = B[j] \ge \max_{i'\ge i_v}A[i']$ obviously still holds. The first invariant, $\min_{i'\ge i_v}A[i'] = A[i_v]\le \min_{j'\ge j}B[j']$ also holds, due to the property  $A[i_v]\le B[j_v]$ we have just shown.
     \item Suppose $A[i_h] \le B[j_h]$. Then we can analogously show the invariant still holds, similar to the last paragraph of the previous case.
 \end{itemize}
 Next, we make sure that each iteration must move one of the pointers by at least one step.
 Since the loop has not terminated, we know from the definition of the loop that $i<|A|$ and $j<|B|$ both hold.
In this case, by the way we break ties, we must have $i_v>i$, and also $j_h>j$. Hence we must make progress in either case.

Finally,  after the loop finishes when $i=|A|$ or $j=|B|$, we still have the loop invariant, which allows us to move the unfinished pointer to the end without violating the $A[i]\le B[j]$ requirement. This completes our pointer-advancement argument.
\end{proof}
\section{Every Graph has a Dominant Schedule (Deferred Proofs from \cref{sec:dominantexists})}
\label{sec:dominance-proofs}

    \subsection{The Prefix Lemma}
    \label{subsec:prefixlemma}
Now we state our main technical lemma. 
Informally, the lemma states that given any two schedules $\sigma$ and $\tau$ such that $\set(\tau_{1:i}) \subseteq \set(\sigma) \subseteq \set(\tau)$, one can always find a subschedule $\sigma'$ such that moving $\sigma'$ to be the prefix of $\tau$ only improves the resulting schedule. 

\begin{lemma}[Prefix Lemma] \label{lem:prefix-lemma}
Let $\sigma$ and $\tau$ be partial schedules of a DAG $G$. 
Suppose there exists an $i\geq 0$ such that $\set(\tau_{1:i}) \subseteq \set(\sigma) \subseteq \set(\tau)$. 
Then, there exists a valid partial schedule $\sigma'$  that satisfies all of the following properties.
\begin{enumerate}
    \item \label{prefix_lemma:prefix_valid} 
$\sigma'$ is a subschedule of $\sigma$ and $\set(\tau_{1:i}) \subseteq \set(\sigma')$.
    \item \label{prefix_lemma:s_better} 
For all $j$, $\nodesumcost{\sigma' \capdot \sigma_{1:j}} \le  \nodesumcost{\sigma_{1:j}}$, with strict inequality if $\sigma' \capdot \sigma_{1:j} \neq \sigma_{1:j}$.

\item \label{prefix_lemma:t_better} 
Let $\tau' = \sigma' \cupdot  \tau$, and let $n = |\tau| = |\tau'|$. 
Then 
 $\max_{j > i}  \nodesumcost{\tau_{1:j}} \ge  \max_{j > |\sigma'|} \nodesumcost{\tau'_{1:j}}$.

\item \label{prefix_lemma:t_dominates} 
Further, if $\sigma \ledot \tau_{1:i}$, then $\tau' \ledot \tau$.
\end{enumerate}
For convenience, we write $\sigma' = MOVE(\tau_{1:i}, \sigma, \tau)$.
\end{lemma}

\begin{proof}
We'll create $\sigma'$ algorithmically, proving properties based on the algorithm. To this end, we define the MOVE algorithm. It tries to improve $\tau$ by moving blocks of nodes with negative total weight toward the beginning of the schedule.

\begin{algorithm}[H] \label{alg:move}
\DontPrintSemicolon
\caption{$\textsc{MOVE}$ }
  \SetAlgoLined
  \KwData{Partial schedules $\sigma$, $\tau$, and a prefix schedule $\tau_{1:i}$, with $\set(\tau_{1:i}) \subseteq \set(\sigma) \subseteq \set(\tau)$.}
  \KwResult{A partial schedule $\sigma'$ satisfying the properties in Lemma \ref{lem:prefix-lemma}.}
  Let $\hat{\tau}$ be a copy of $\tau$. We'll modify $\hat{\tau}$.\;
  Designate some subsequences of $\hat{\tau}$ as \emph{blocks} (disjoint sets of nodes that are consecutive in $\hat{\tau}$): one block $B_0 = \set(\tau_{1:i})$, and a separate block for each $\sigma_j \in \set(\sigma) \setminus \set(\tau_{1:i})$. Let $\mathcal{B}$ be the collection of these blocks. (Some nodes of $\hat{\tau}$ may not be in any block.)\;
  \For{$j=1$ to $|\sigma|$} {
    Let $B \in \mathcal{B}$ be the block that contains $\sigma_j$. If $B=B_0$, continue to the next value of $j$. \;    
    \While{$\nodesumcost{B} \leq 0$} {
        \tcp{Note:\ $B$ can't be at the start of $\hat{\tau}$ (which always starts with $B_0$). }
        \uIf {$B$ is immediately preceded in $\hat{\tau}$ by $B_0$}{
          Set $B_0 \leftarrow B_0 \cup B$, remove $B$ from $\mathcal{B}$, and break (go to the next value of $j$).\label{line:break}\;
        }
        \uElseIf {$B$ is immediately preceded in $\hat{T}$ by a block $B'$ such that $\sigma_{j'}\in B'$ for some $j' < j$ \label{line:bprime} } {
          Merge $B'$ and $B$: set $B \leftarrow B' \cup B$ and remove $B'$ from $\mathcal{B}$.\; \label{line:merge}
        }
        \Else {
          Swap places of $B$ and its immediately preceding node or block in $\hat{\tau}$. \label{line:swap}
        }
    }
  }
  \Return a partial schedule $\sigma'$: nodes of $B_0$, ordered according to $\sigma$.
\end{algorithm} 

We first show some useful properties of Algorithm \ref{alg:move}.
\begin{claim} \label{claim:prefix}
Let $B\in \mathcal{B}$ be a block during an execution of Algorithm \ref{alg:move} such that $B\neq B_0$.
Let $B = \{\sigma_{u_1}, \sigma_{u_2}, ..., \sigma_{u_h}\}$, with $u_1 < u_2 <...< u_h$ being the positions in schedule $\sigma$. Then:
\begin{itemize}
    \item \emph{\textbf{PrefixPositive:}}  If $|B| > 1$, then $\nodesumcost{\sigma_{u_1}} + \nodesumcost{\sigma_{u_2}} + ... + \nodesumcost{\sigma_{u_j}} > 0$ for all $1 \le j < h$. Consequently, if $\nodesumcost{B}>0$, this is true for all $j\leq h$.
    \item \emph{\textbf{SuffixNegative:}} If $\nodesumcost{B} \le 0$, then $\nodesumcost{\sigma_{u_j}} + \nodesumcost{\sigma_{u_{j+1}}} + ... + \nodesumcost{\sigma_{u_h}} \le 0$ for all $j$.
\end{itemize}
\end{claim}

\begin{proof}
    We begin by showing two properties that hold during the $j$th iteration of the \emph{for} loop.
    
    First, any block $B\in \mathcal{B}$ with $B\neq B_0$ and $|B| > 1$ consists only of nodes with index in $\sigma$ less than or equal to $j$. This follows by induction on the steps of the algorithm. Initially there are no such blocks, and a new one can be created on line \ref{line:merge}. If $|B'|=1$, its node is before $j$ in $\sigma$ by the \emph{if} condition on line \ref{line:bprime}; if $|B'|>1$, its nodes are before $j$ by the induction hypothesis. Similarly for $B$, which is either the singleton $\{\sigma_j\}$, or has nodes earlier in $\sigma$ by the induction hypothesis. Thus, the merged block also satisfies this property.

    Second, any block $B\neq B_0$ such that $\max \{j' : \sigma_{j'} \in B\} < j$ has $\nodesumcost{B}>0$. Let $\ell=\max \{j' : \sigma_{j'} \in B\}$ and consider a past \emph{for} loop iteration when $j$ was equal to $\ell$. The \emph{while} loop couldn't have ended on line \ref{line:break}, as $B\neq B_0$, so it must have ended due to $\nodesumcost{B}>0$. The block $B$ couldn't have changed in later iterations, because then it would have acquired a node with index higher than $\ell$.

    To show the PrefixPositive part of the claim, we proceed by induction on the steps of the algorithm. Initially, there are no blocks satisfying the criteria of the claim. A new block satisfying the criteria can only be created on line \ref{line:merge}. Let $B'$ and $B$ be the two blocks that are merged into $B$ there. Since $\sigma_j\in B$, by the first property above, $j$ must be the highest node index in $B$. Let $\ell$ be the highest node index in $B'$, and by the same property it must be that $\ell<j$. Thus, $j$ will be the highest node index of the merged block. So any strict prefix of the merged block will be the union of some prefix of $B'$ and some prefix of $B$ (where a \emph{prefix} refers to an ordering of the block according to indices in $\sigma$). Now, by the second property above, $\nodesumcost{B'}>0$. So all non-empty prefixes of $B'$ (including the full $B'$) have positive cost (by induction hypothesis if $|B|>1$ or from $\nodesumcost{B'}>0$, if $|B'|=1$). If $B$ contributes a non-empty prefix, it must have positive cost by the induction hypothesis. Since both prefixes can't be empty, the total cost of the prefix of the merged block must be positive.

    For SuffixNegative, suppose that $\nodesumcost{B} \le 0$. If $u_j=1$, the claim is trivial. Otherwise, 
    $\nodesumcost{\sigma_{u_j}} + ... + \nodesumcost{\sigma_{u_h}} = \nodesumcost{B} - [\nodesumcost{\sigma_{u_1}} +... + \nodesumcost{\sigma_{u_{j-1}}}].$
    Since the expression in brackets is positive by PrefixPositive, the claim follows.    
\end{proof}

We now proceed with the proof of Lemma \ref{lem:prefix-lemma}.

\vspace{-3mm}
\paragraph{Property \ref{prefix_lemma:prefix_valid}.}
Note that $\set(\tau_{1:i})\subseteq B_0$ from the start and never leaves it, so $\set(\tau_{1:i}) \subseteq \set(\sigma')$. Further, only nodes of $\sigma$ are in any blocks and can join $B_0$, so $\set(\sigma') \subseteq \set(\sigma)$. Since $\sigma'$ is ordered according to $\sigma$, it is a subschedule of $\sigma$.

\vspace{-3mm}
\paragraph{$\sigma'$ is a valid partial schedule of $G$.}
To show that $\sigma'$ is a valid partial schedule of $G$, we need to show that there are no precedence violations. For the sake of contradiction, suppose $G$ has an edge from a node $u$ to a node $\sigma_j$, such that $\sigma_j$ is in $\sigma'$, but $u$ does not appear before it in $\sigma'$. It's not possible that $u \not\in \sigma$ or that $u$ appears after $\sigma_j$ in $\sigma$, otherwise $\sigma$ would not be a valid partial schedule of $G$. Therefore $u=\sigma_{j'}$ for some $j'<j$. Since $\sigma'$ is ordered like $\sigma$, $u$ can't be after $\sigma_j$ in $\sigma'$, so $u$ is not in $\sigma'$ at all. It's also not possible that $u$ appears after $\sigma_j$ in $\tau$, otherwise $\tau$ would not be a valid partial schedule of $G$. So, in Algorithm \ref{alg:move}, the block with $u=\sigma_{j'}$ started earlier in $\hat{\tau}$ than the block with $\sigma_j$. However, since a block with $\sigma_j$ joined $B_0$ at some point, and the block with $\sigma_{j'}$ did not, the two must have been swapped. But this is not possible, since the swaps on line \ref{line:swap} only happen when there is no smaller-indexed node in the preceding block.

\vspace{-3mm}
\paragraph{Property \ref{prefix_lemma:s_better}.}
Let $\mathcal{B}' = \mathcal{B} \setminus \{B_0\}$ for value of $\mathcal{B}$ at the very end of Algorithm \ref{alg:move}, i.e. $\mathcal{B}'$ is a collection of blocks that never got merged into $B_0$. Thus, $\set(\sigma) \setminus \set(\sigma') = \cup_{B\in \mathcal{B}'} B$.  Any $B \in \mathcal{B}'$ must have positive cost, and, by Claim \ref{claim:prefix}, all its non-empty prefixes (defined based on indices in $\sigma$) must have positive cost as well. 
The nodes in $\set(\sigma' \capdot \sigma_{1:j})$ are the nodes of $\sigma'$ with indices up to $j$. Nodes of $\sigma_{1:j}$ that are not in this set constitute a union of prefixes of blocks in $\mathcal{B}'$. If all these prefixes are empty, the compared sets are identical, and the claim is trivial. Otherwise, these prefixes have a positive cost in total, and  $\nodesumcost{\sigma' \capdot \sigma_{1:j}} <  \nodesumcost{\sigma_{1:j}}$.

\vspace{-3mm}
\paragraph{$\sigma'$ dominates $\sigma$.}
We also show that $\sigma' \ledot \sigma$, which is needed for the proof of property \ref{prefix_lemma:t_dominates} below.
To do this, we show the four inequalities in Definition \ref{defn:algdef}. For any $i\in [|\sigma'|]$, let $j\in [|\sigma|]$ be such that $\sigma'_i = \sigma_j$ (recall that $\sigma'$ is a subschedule of $\sigma$, so such $j$ exists, and, moreover, the value of $j$ increases monotonically with $i$). 
For the first inequality, 
$\max_{i'\geq i} \nodesumcost{\sigma'_{1:i'}} \leq \max_{j'\geq j} \nodesumcost{\sigma_{1:j'}}$,
let $i^*\ge i$ be the index that realizes the max on the left.
Let $j^* \ge j$ be such that $\sigma'_{i^*} = \sigma_{j^*}$. Then $\set(\sigma_{1:j^*})$ is a superset of $\set(\sigma'_{1:i^*})$, differing by some prefixes of blocks in $\mathcal{B}'$, which have positive cost. So $\max_{j'\geq j} \nodesumcost{\sigma_{1:j'}} \ge
\nodesumcost{\sigma_{1:j^*}} \geq \nodesumcost{\sigma'_{1:i^*}} = \max_{i'\geq i} \nodesumcost{\sigma'_{1:i'}}$. The proof of the third inequality is nearly identical. 
To show the second and fourth inequalities, we let $j^*$ be the value of $j'$ realizing the min on the right, and then define $i^* = \max \{i' : \set(\sigma'_{1:i'}) \subseteq \set(\sigma_{1:j^*})\}$ to be the latest node that $\sigma'$ has in common with $\sigma_{1:j^*}$. Since $\set(\sigma'_{1:i^*}) \subseteq \set(\sigma_{1:j^*})$, differing by some positive prefixes of blocks in $\mathcal{B}'$, the inequalities follow.

\vspace{-3mm}
\paragraph{Property \ref{prefix_lemma:t_better}.}
To show property \ref{prefix_lemma:t_better}, we actually show a stronger claim, which will be needed later. The actual property \ref{prefix_lemma:t_better} follows as a consequence.
We claim that the memory profile of $\tau'$ after $|\sigma'|$ dominates that of $\tau$ after $i$:
\[\big(\nodesumcost{\tau'_{1:|\sigma'|}}, 
\nodesumcost{\tau'_{1:|\sigma'|+1}}, ...,
\nodesumcost{\tau'_{1:n}} \big) ~\ledot~
\big(\nodesumcost{\tau_{1:i}}, 
\nodesumcost{\tau_{1:i+1}}, ...,
\nodesumcost{\tau_{1:n}}\big).\]
%
Again, we show the four inequalities in Definition \ref{defn:algdef}.
We let the indices $j$ for the portion of the memory profile of $\tau'$ range from $|\sigma'|$ to $n$, and the indices $k$ for the portion of the memory profile of $\tau$ range from $i$ to $n$, to better correspond with the expression in the current lemma. 
To meet the requirement of Definition \ref{defn:algdef}, for any $j\in \{|\sigma'|+1, ..., n\}$, let $k\in \{i+1, ..., n\}$ be such that $\tau'_j = \tau_k$. 
Since $\set(\tau') = \set(\tau)$,  $\set(\sigma')\supseteq \set(\tau_{1:i})$, and $\sigma'$ is a prefix of $\tau'$, we have $\set(\tau'_{|\sigma'|+1:n}) \subseteq \set(\tau_{i+1:n})$. Thus, the required value of $k$ in the correct range exists. Also, since $\tau'_{|\sigma'|+1:n}$ is a subschedule of $\tau$, these values of $k$ increase monotonically with $j$.

To show the first inequality of Definition \ref{defn:algdef}, 
$\max_{j'\geq j} \nodesumcost{\tau'_{1:j'}} \le \max_{k'\geq k} \nodesumcost{\tau_{1:k'}}$, let $j^*$ be the value of $j'$ that achieves the max on the left. Let $k^*$ be such that $\tau'_{j^*} = \tau_{k^*}$. By monotonicity discussed above, $j^*\geq j$ implies $k^*\geq k$. The set $\tau'_{1:j^*}$ is a superset of $\tau_{1:k^*}$, differing by some suffixes of blocks $B$ that started later in $\tau$, but merged with $B_0$ on line \ref{line:break} of Algorithm \ref{alg:move}. Since these blocks had $\nodesumcost{B} \le 0$, the SuffixNegative part of Claim \ref{claim:prefix} implies that the cost of these suffixes is non-positive, so $\nodesumcost{\tau'_{1:j^*}} \le \nodesumcost{\tau_{1:k^*}}$. Since $\max_{k'\geq k} \nodesumcost{\tau_{1:k'}} \ge \nodesumcost{\tau_{1:k^*}}$, the inequality follows. The proof of the third inequality is nearly identical (but we need to modify the range of $j'$ to $|\sigma'| \le j' \le j$ and the range of $k'$ to $i \le k' \le k$ in the subscripts of $\max$).
We also note that the inequality 
$\max_{j > i}  \nodesumcost{\tau_{1:j}} \ge  \max_{j > |\sigma'|} \nodesumcost{\tau'_{1:j}}$
in property \ref{prefix_lemma:t_better} follows by the same reasoning, by using $i$ and $|\sigma'|$ instead of $j$ and $k$, respectively.

For the second inequality, $\min_{j'\geq j} \nodesumcost{\tau'_{1:j'}} \le \min_{k'\geq k} \nodesumcost{\tau_{1:k'}}$, let $k^*$ be the value of $k'$ that achieves the min on the right. Let $j^* = \max \{j' : \set(\tau'_{|\sigma|+1:j'}) \subseteq \set(\tau_{i+1:k^*})\}$ be the latest node before $k^*$ that hasn't been merged into $B_0$. Then, as before, $\tau'_{1:j^*}$ is a superset of $\tau_{1:k^*}$ differing by some non-positive suffixes, which implies that 
$\min_{j'\geq j} \nodesumcost{\tau'_{1:j'}} \le  \nodesumcost{\tau'_{1:j^*}} \le \nodesumcost{\tau_{1:k^*}} = \min_{k'\geq k} \nodesumcost{\tau_{1:k'}}$. The fourth inequality has a similar proof.

\vspace{-3mm}
\paragraph{Property \ref{prefix_lemma:t_dominates}.}
Finally, we prove property \ref{prefix_lemma:t_dominates} by using Lemma \ref{lem:concatendominance} about concatenation.
To prove the result, we consider the memory profile of $\tau' = \sigma' \cupdot  \tau$ as consisting of two parts: the memory profile of $\sigma'$, followed by the remaining part. The memory profile of $\tau$ consists of one for $\tau_{1:i}$, followed by the remaining part.
Now, by transitivity, $\sigma' \ledot \sigma$ (proven above) and $\sigma \ledot \tau_{1:i}$ (by assumption) imply $\sigma' \ledot \tau_{1:i}$. Further, the remaining part of $\tau'$ dominates the remaining part of $\tau$ by the proof of property \ref{prefix_lemma:t_better} above. 
Thus, the conclusion follows.
\end{proof}

\subsection{Proof of Common Prefix Extension Lemmas}
\label{subsec:prefixextension}
Before proving \cref{lem:extendtopeak} and \cref{lem:extendtovalley} for \cref{alg:awesomizer},
we first introduce another auxiliary lemma.
Recall that in \cref{alg:awesomizer}, $\alpha$ is defined as a peak-minimizing schedule of $G$, and $\sigma$ is the schedule returned by \cref{alg:awesomizer}.
The following lemma roughly says that $\sigma$  has  the same leftmost peak as $\alpha$ does.

\begin{lemma}[$\sigma$ has the same leftmost peak as $\alpha$]
\label{lem:peaknotchanged}
In \cref{alg:awesomizer}, $S_h = A_h$, and hence $\nodesumcost{S_h} = \hat c$. Moreover, $\max \nodesumcost{\sigma} = \hat c$, and $h$ is the smallest $h'\in \{0,1,\dots,|V|\}$ such that $\nodesumcost{S_{h'}} = \hat c$.  
\end{lemma}
\begin{proof}
By \cref{line:recursion}, $\lambda$ is a prefix of $\sigma$.
Since $A_h \subseteq L$ (by \cref{line:defn-mincut-L}), and $\lambda = \alpha \capdot L$ consists of the nodes of $L$ arranged in the same order as $\alpha$ (by \cref{line:lambda}), we know $\alpha_{1:h}$ must be a prefix of $\lambda$. Hence, 
\[\alpha_{1:h} = \lambda_{1:h} = \sigma_{1:h}.\]
In particular, $S_h = A_h$ and $\nodesumcost{S_h} = \nodesumcost{A_h} =  \hat c$, which means this peak of $\alpha$ also occurs in $\sigma$ at the same position.
Moreover, by the definition of $h$ in \cref{line:defn-i}, for all $0\le h_0< h$, we have $\nodesumcost{S_{h_0}} = \nodesumcost{A_{h_0}} < \hat c$, so $\sigma$ cannot have a peak to the left of $S_h$. Hence, to prove the claimed statement, it remains to show the upper bound, $\max \nodesumcost{\sigma} \le \hat c$. 
Since the definition of $\sigma$ involves recursion, we prove the following claim as an intermediate step.
\begin{claim}
\label{claim:lambdaprime}
Extend partial schedule $\lambda$ to the full schedule $\lambda' = \lambda \cupdot \alpha$. Then, $\max \nodesumcost{\lambda'} \le \hat c$
\end{claim}
We first show how \cref{claim:lambdaprime} implies the desired upper bound $\max \nodesumcost{\sigma} \le \hat c$ by an induction.
Write $\lambda'$ as the concatenation $\lambda' = \lambda \circ \beta$, where $\beta$ is a schedule of $G[V\setminus L]$. Since $\max\nodesumcost{\lambda'} \le \hat c$ by \cref{claim:lambdaprime}, we know the peak cost of $\beta$ in $G[V\setminus L]$ is at most $\hat c - \nodesumcost{L}$. Hence, the peak-minimizing schedule of $G[V\setminus L]$ has peak cost at most $\hat c - \nodesumcost{L}$. We know as an inductive hypothesis that $\textsc{FindSchedule}(G[V\setminus L])$ returns a schedule whose peak cost is not higher than the peak-minimizing schedule of $G[V\setminus L]$, and is in turn at most $\hat c - \nodesumcost{L}$. Hence, 
\begin{align*}
\max    \nodesumcost{\sigma} &= \max \nodesumcost{\lambda \circ \textsc{FindSchedule}(G[V\setminus L])}\\
    & \le \max \{ \max \nodesumcost{\lambda}, \nodesumcost{L} + (\hat c - \nodesumcost{L}) \}\\
    & \le \hat c,
\end{align*} 
as desired.

Now, it remains to prove \cref{claim:lambdaprime}.
\begin{proof}[Proof of \cref{claim:lambdaprime}]
We need to prove
$\nodesumcost{L'_{h'}} \le \hat c$ for all $h'\in \{0,1,\dots,|V|\}$. For $h'\le h$, this immediately follows from $\nodesumcost{L'_{h'}} = \nodesumcost{L_{h'}} = \nodesumcost{A_{h'}}\le \hat c$. We consider the following two remaining cases:
\begin{itemize}
    \item \textbf{Case $h+1 \le h'\le |L|$:}
Note that we still have $L_{h'} = L'_{h'}$     in this case.
Suppose for contradiction that $\nodesumcost{L_{h'}}> \hat c$. Let the nodes in $L_{h'}$ be 
\[ \alpha_1,\alpha_2, \dots,\alpha_h,\, \alpha_{j_1},\alpha_{j_2},\dots,\alpha_{j_{h'-h}},\]
where $h<j_1<j_2<\dots<j_{h'-h}\le |V|$.
Note that $A_{j_{h'-h}} \cup L$ is a topological cut of $G$ (by \cref{lemma:capcuptopo}) that contains $A_h$, and we claim it has strictly smaller cost than $L$, contradicting the definition of $L$ at \cref{line:defn-mincut-L}. Indeed,
\begin{align*}
 \nodesumcost{A_{j_{h'-h}} \cup L} -   \nodesumcost{L}   &= 
    \nodesumcost{A_{j_{h'-h}}} -  \nodesumcost{L\cap A_{j_{h'-h}}}\\
   & = \nodesumcost{A_{j_{h'-h}}} -  \nodesumcost{L_{h'}}\\
   & < \nodesumcost{A_{j_{h'-h}}} - \hat c\\
   & \le \max \nodesumcost{\alpha} - \hat c \\
   & = 0.
\end{align*}
    \item \textbf{Case $h'>|L|$:}
 Since $\lambda' = \lambda \cupdot \alpha$ and $A_h \subseteq L$, we can write $L'_{h'} = L \cup A_j$ for some $j > h$. 
 Suppose for contradiction that $\nodesumcost{L'_{h'}}>\hat c$. Then, 
\begin{align*}
   \nodesumcost{L}  - \nodesumcost{L \cap A_j} & = \nodesumcost{L\cup A_j}  -\nodesumcost{A_j}\\
   & = \nodesumcost{L'_{h'}}  -\nodesumcost{A_j}\\
   & > \hat c -\nodesumcost{A_j}\\
   & \ge \hat c - \max \nodesumcost{\alpha}\\
   & = 0.
\end{align*}
Hence, $L \cap A_j$ is a topological cut of $G$ (by \cref{lemma:capcuptopo}) with strictly smaller cost than $L$, and contains $A_h$ (by $j>h$), contradicting the definition of $L$ at \cref{line:defn-mincut-L}. \qedhere
\end{itemize}
\end{proof}
\end{proof}

Now we are ready to prove \cref{lem:extendtopeak} and \cref{lem:extendtovalley} using the prefix lemma (\cref{lem:prefix-lemma}). The two proofs have a similar overall structure.

\begin{proof}[Proof of \cref{lem:extendtovalley}]
We are given some schedule $\tau$ of $G$ that is not dominated by $\sigma = \textsc{FindSchedule}(G)$, where $\tau_{1:h}=\sigma_{1:h}$. To construct another schedule $\tau'$, we  apply the prefix lemma (\cref{lem:prefix-lemma}) to obtain a partial schedule
    \[\sigma' = MOVE(\tau_{1:h}, \sigma_{1:|L|}, \tau),\] and extend it to the full schedule 
    \begin{equation}
    \label{eqn:case1defnofT'}
     \tau' = \sigma' \cupdot \tau   
    \end{equation}
    as in \cref{lem:prefix-lemma}.
     We have the following two claims.
\begin{claim}$ \tau' \ledot \tau$.
\label{claim:case1claim1}
\end{claim}
\begin{proof}[Proof of \cref{claim:case1claim1}]
By Property~\ref{prefix_lemma:t_dominates} of \cref{lem:prefix-lemma}, it suffices to establish $\sigma_{1:|L|}\ledot \tau_{1:h}$, using the following pointer-advancement argument (\cref{defn:cost-seq-superior}).

We start with two pointers at $S_0$ and $T_0$ respectively.
By reflexivity  of dominance (\cref{lem:superiority-reflexive-transitive}), we have $\sigma_{1:h}\ledot \tau_{1:h}$, so we can advance the two pointers to $S_h$ and $T_h$ respectively.  Then by \cref{lem:peaknotchanged}, $S_h$ is a peak of $\sigma$, so it holds for all $j$ that $\nodesumcost{S_j} \le \nodesumcost{S_h} = \nodesumcost{T_h}$. Hence, we can continue to advance the pointer on $\sigma$ from $S_h$ to $S_{|L|}$, while keeping the cost value along the way at most $\nodesumcost{T_h}$. This proves that $\sigma_{1:|L|} \ledot \tau_{1:h}$.
\end{proof}

\begin{claim} $\sigma' = \sigma_{1:|L|}$.
\label{claim:case1claim2}
\end{claim}
\begin{proof}[Proof of \cref{claim:case1claim2}]
Recall from Property~\ref{prefix_lemma:prefix_valid} of \cref{lem:prefix-lemma} that $\sigma'$ is a subschedule of $\sigma_{1:|L|}$.
Suppose for contradiction that $\sigma' \neq \sigma_{1:|L|}$, which means $\sigma'$ is strictly contained in $\sigma_{1:|L|}$. Then, we apply Property~\ref{prefix_lemma:s_better} of \cref{lem:prefix-lemma} to $S_{|L|}$, and obtain the strict inequality \[\nodesumcost{S_{|L|}} > \nodesumcost{\set(\sigma') \cap S_{|L|}} = \nodesumcost{\set(\sigma')}.\]
Since  $S_{|L|} = L$ (see \cref{line:recursion}),
 this means we have found a topological cut $\set(\sigma')$ of $G$ with strictly smaller cost than $L$. 
Note also that $\set(\sigma')$ contains $A_h$, which is a consequence of $T_h \subseteq \set(\sigma')$ (by Property~\ref{prefix_lemma:prefix_valid} of \cref{lem:prefix-lemma}), $T_h = S_h$ (by $\sigma_{1:h}=\tau_{1:h}$), and $S_h = A_h$ (by \cref{lem:peaknotchanged}). Together, this contradicts the definition of $L$ at \cref{line:defn-mincut-L} in \cref{alg:awesomizer}.
\end{proof}

Since $\sigma$ does not dominate $\tau$, we know from \cref{claim:case1claim1} that $\sigma$ does not dominate $\tau'$ either (due to the transitivity of dominance in  \cref{lem:superiority-reflexive-transitive}). 
Moreover, $\sigma_{1:|L|}$ is a prefix of $\tau'$ (due to \cref{eqn:case1defnofT'} and \cref{claim:case1claim2}). Hence $\tau'$ satisfies the desired properties.
\end{proof}

\vspace{0.5cm}

Now we prove \cref{lem:extendtopeak}.
\begin{proof}[Proof of \cref{lem:extendtopeak}]
We are given some schedule $\tau$ of $G$ that is not dominated by $\sigma = \textsc{FindSchedule}(G)$.


Recall from \cref{line:defnalpha,line:definecost} in  \cref{alg:awesomizer} that 
$\hat c$ is the peak cost of a peak-minimizing schedule of $G$.  This means schedule $\tau$ cannot have a lower peak, so there exists $i \in \{0,1,\dots,|V|\}$  such that 
\begin{equation}
    \nodesumcost{T_i} \ge \hat c,
    \label{eqn:defnti}
\end{equation}
and we pick the largest such $i $.
Then, we pick the smallest $j \in \{0,1,\dots,|V|\}$ such that \[\set(T_i) \subseteq \set(S_j).\]

Apply the prefix lemma (\cref{lem:prefix-lemma}) to obtain a partial schedule
\[\sigma' = MOVE(\tau_{1:i}, \sigma_{1:j}, \tau),\]
and extend it to the full schedule
    \begin{equation}
    \label{eqn:case2defnofT'}
     \tau' = \sigma' \cupdot \tau
    \end{equation}
    as in \cref{lem:prefix-lemma}.

To prove $\tau'$ satisfies the desired properties, we again have two claims.
\begin{claim}$ \tau' \ledot \tau$.
\label{claim:case2claim1}
\end{claim}
\begin{proof}[Proof of \cref{claim:case2claim1}]
By Property~\ref{prefix_lemma:t_dominates} of \cref{lem:prefix-lemma}, it suffices to establish $\sigma_{1:j} \ledot \tau_{1:i}$, using the following pointer-advancement argument (\cref{defn:cost-seq-superior}).

We start with two pointers at $S_0$ and $T_0$ respectively.

By the cut-nonnegativity of $G$, we have $\nodesumcost{T_k}\ge 0$ for all $k$.
Hence, we can move the pointer on $\tau$ from $T_0$ to $T_i$, without ever going below $\nodesumcost{S_0}=0$. 

Next, by \cref{lem:peaknotchanged} we have $\nodesumcost{S_k} \le \hat c$ for all $k$. Combined with \cref{eqn:defnti}, this means $\nodesumcost{S_k} \le \nodesumcost{T_i}$ for all $k$.
Hence,  we can move the pointer on $\sigma$  from $S_0$ to  $S_j$, without going above $\nodesumcost{T_i}$.  This proves that $\sigma_{1:j} \ledot \tau_{1:i}$.
\end{proof}
\begin{claim}
$\sigma_{1:h}$ is a prefix of $\sigma'$.
\label{claim:case2claim2}
\end{claim}
\begin{proof}[Proof of \cref{claim:case2claim2}]
By Property~\ref{prefix_lemma:t_better} 
of \cref{lem:prefix-lemma}, we have \[\max_{k > i} \nodesumcost{T_k} \ge \max_{k > |\sigma'|} \nodesumcost{T'_k},\]
and since we chose $i$ to be the largest one satisfying \cref{eqn:defnti}, we have \[\max_{k > i} \nodesumcost{T_k} < \hat c .\]
Combining the two inequalities gives 
\begin{equation}
\label{eqn:nopeakafterm}
 \max_{k > |\sigma'|} \nodesumcost{T'_k}< 
 \hat c.
\end{equation}
 
On the other hand, using again the fact that $\hat c$ is the peak cost of a peak-minimizing schedule of $G$, we have
\begin{equation}
\label{eqn:haspeak}
 \max\nodesumcost{\tau'} \ge   \hat c.
\end{equation}

Combining \cref{eqn:nopeakafterm} and \cref{eqn:haspeak}, we conclude that $\max \nodesumcost{\tau'}$ must be  attained by $T'_k$ for some $k\le |\sigma'|$.
Since $\sigma'$ is a prefix of $\tau'$ (by \cref{eqn:case2defnofT'}),  $T'_k$ is the same as $S'_k$, so we have 
\[
 \nodesumcost{S'_k} \ge \hat c
 \]
for some $k\le |\sigma'|$.
We then express this $S'_{k}$ as $S'_{k} = \set (\sigma' \capdot \sigma_\ell)$ for some $0\le \ell\le j$, which is possible because $\sigma'$ is a subschedule of $\sigma_{1:j}$ (by Property~\ref{prefix_lemma:prefix_valid} of \cref{lem:prefix-lemma}). Hence, 
\begin{equation}
\label{eqn:ineq1}
 \nodesumcost{\set(\sigma'\capdot \sigma_\ell)} \ge \hat c. 
\end{equation}

By Property~\ref{prefix_lemma:s_better} of \cref{lem:prefix-lemma} applied to $S_\ell$, 
\begin{equation}
\label{eqn:ineq2}
    \nodesumcost{\set(\sigma' \capdot \sigma_\ell)} \le \nodesumcost{S_\ell},
\end{equation} 
where the equality can hold only if $\set(\sigma'\capdot \sigma_\ell) = S_\ell$.

By \cref{lem:peaknotchanged}, $S_h$ is the leftmost peak of $\sigma$ with peak cost $\hat c$, so
\begin{equation}
\label{eqn:ineq3}
 \hat c \ge \nodesumcost{S_\ell},
\end{equation}
where the equality can hold only if $\ell\ge h$.

Combine the three inequalities \cref{eqn:ineq1,eqn:ineq2,eqn:ineq3}, we conclude that all equalities must hold, and the equality conditions $\ell\ge h$, $\set(\sigma'\capdot \sigma_\ell) = S_\ell$ are achieved.  Together they imply $S_h \subseteq S_\ell \subseteq \set(\sigma')$. Hence, $\sigma'$ must contain $\sigma_{1:h}$ as a prefix, since $\sigma'$ is a subschedule of $\sigma_{1:j}$ (by Property~\ref{prefix_lemma:prefix_valid} of \cref{lem:prefix-lemma}).
\end{proof}

Since $\sigma$ does not dominate $\tau$, we know by \cref{claim:case2claim1} that $\sigma$ does not dominate $\tau'$. 
Moreover, $\sigma_{1:h}$ is a prefix of $\tau'$ (due to \cref{eqn:case2defnofT'} and \cref{claim:case2claim2}).
Hence $\tau'$ satisfies the desired properties.
\end{proof}

\subsection{A Dominant Schedule Passes Through Min-Cut}
\label{subsec:mincut}
\mincutlemma*
\begin{proof}
Suppose $\max \nodesumcost{\tau}$ is attained by $\nodesumcost{T_i}$.
Apply the prefix lemma (\cref{lem:prefix-lemma}) and obtain
\[ \sigma' = MOVE(\tau_{1:i}, \sigma, \tau),\]
which is a subschedule of $\sigma$ (by Property~\ref{prefix_lemma:prefix_valid}), and let $\tau'=\sigma'\cupdot \tau$ as in \cref{lem:prefix-lemma}.
  We claim that $\sigma' = \sigma$. Otherwise, $\sigma'$ is strictly contained in $\sigma$, and  applying Property~\ref{prefix_lemma:s_better} of \cref{lem:prefix-lemma} to $\sigma$ gives the strict inequality $\nodesumcost{\set(\sigma')} < \nodesumcost{\set(\sigma)}$, contradicting the fact that $S = \set(\sigma)$ is a minimum topological cut. Hence, $\sigma'= \sigma$ and $\tau'=\sigma \cupdot \tau$.

For all $0\le j\le |V|$,   we have
\begin{align}
  \nodesumcost{S \cap T_j}   &=   \nodesumcost{T_j}-(\nodesumcost{S\cup T_j} - \nodesumcost{S}) \nonumber\\
  &\le  \nodesumcost{T_j}, \label{eqn:ktiti}
\end{align}
where we used the fact $S\cup T_j$ is a topological cut and cannot have cost lower than $\nodesumcost{S}$.

By Property~\ref{prefix_lemma:t_dominates} of \cref{lem:prefix-lemma}, in order to show the  desired statement $\tau'\ledot \tau $, it suffices to show $\sigma \ledot \tau_{1:i}$.
We now show $\sigma \ledot \tau_{1:i}$ using a pointer-advancement argument (\cref{defn:cost-seq-superior}) with the following two phases: 
\begin{itemize}
    \item 
Recall that $\sigma$ is a subschedule of $\tau$. Denote $i'=|S \cap T_i|$, so that $\set(\sigma_{1:i'})=S\cap T_i$. Applying \cref{eqn:ktiti} to all $0\le j\le i$ implies $\sigma_{1:i'} \ledot \tau_{1:i}$.
Hence, we can advance the pointer for $\sigma$ from the beginning to $S \cap T_i$, and advance the pointer for $\tau_{1:i}$ from $T_0$ to $T_i$.
\item 
Move the pointer for  $\sigma$ from $S \cap T_i$ to $S$.
The costs along the way can be expressed as $\nodesumcost{S\cap T_j}$ for $i\le j\le |V|$, since $\sigma$ is a subschedule of $\tau$.
These costs never exceed $\nodesumcost{T_i}$,  since $\nodesumcost{S\cap T_j}\le  \nodesumcost{T_j} \le \nodesumcost{T_i }$ for all $j$ (by \cref{eqn:ktiti} and the definition of $i$).
\end{itemize}
This proves that $\sigma \ledot \tau_{1:i}$, which implies $\sigma \cupdot \tau \ledot \tau$ by Property~\ref{prefix_lemma:t_dominates} of \cref{lem:prefix-lemma}.
\end{proof}
\section{Linearization (Deferred Proofs from \cref{sec:linearization})}
\label{sec:linearization-proofs}

In this section we prove \cref{lem:exchangelinearize} and \cref{lem:linearization}. 
First, we need an auxiliary lemma about memory profiles in the \truemem\  model.
\begin{lemma}
\label{lem:oddhigher}
Let $G=(V,E)$ be a DAG with $|V|=n$ in the \truemem\  model, and $\sigma$ a schedule of $G$. Then, for odd $i$, 
\[ 
\nodeweightedprofile{\sigma}_i \ge \max\{\nodeweightedprofile{\sigma}_{i-1},\nodeweightedprofile{\sigma}_{i+1}\}. \]
Similarly, when $\internalprofile{\sigma}$ is defined, it also holds that
\[ 
\internalprofile{\sigma}_i \ge \max\{\internalprofile{\sigma}_{i-1},\internalprofile{\sigma}_{i+1}\}. \]
\end{lemma}
\begin{proof}
In the following proof, we use ``odd'' to mean memory during the execution of a node (because they appear in odd indices of a schedule), and use ``even'' to mean memory after finishing a node. Recall from \cref{def:irl}, that for any odd index $i = 2j + 1$, $\profile{\sigma}_i = \memduring{\sigma, \sigma_{j+1}}$ and similarly for an even index, $\profile{\sigma}_{2j} = \memafter{\sigma, \sigma_j}$.

Thus, for $i = 2j+1$ we have the following:
\begin{align*}
    \profile{\sigma}_i - \profile{\sigma}_{i-1} = \scr{\sigma_{j+1}} + \sz{\sigma_{j+1}} \geq 0
\end{align*}
where the last inequality follows from the first part of \cref{eqn:scratchcondition}. Similarly, we have the following.
\begin{align*}
   \profile{\sigma}_i - \profile{\sigma}_{i+1} &= \memduring{\sigma, \sigma_{j+1}} - \memafter{\sigma, \sigma_{j+1}} \\
   &= \sz{\sigma_{j+1}} + \scr{\sigma_{j+1}} + \sum_{v \in \mathcal{R}} \sz{v} - \sz{\sigma_{j+1}}\cdot\mathbbm{1}_{\delta^+(\sigma_{j+1}) = \phi}
   \end{align*}
where $\mathcal{R} = \{\sigma_k : 1\leq k \leq j, \delta^+(\sigma_k) \cap \{\sigma_{j+1}, \ldots,\sigma_n\} = \{\sigma_{j+1}\}\}$ is the set of previously scheduled nodes whose outputs can now be deallocated and $\mathbbm{1}_{\delta^+(\sigma_{j+1}) = \phi}$ is an indicator function to indicate whether node $\sigma_{j+1}$ has any out-neighbors. However, note that $\{\sigma_k: \delta^+(\sigma_k) = \{\sigma_{j+1}\}\} \subset \mathcal{R}$ and hence, the second part of \cref{eqn:scratchcondition} implies that $\scr{\sigma_{j+1}} + \sum_{v \in \mathcal{R}} \sz{v} \geq 0$. Plugging this into the equation above yields, $\profile{\sigma}_i - \profile{\sigma}_{i+1} \geq 0$.
Hence we have shown that 
\[ 
\profile{\sigma}_i \ge \max\{\profile{\sigma}_{i-1},\profile{\sigma}_{i+1}\}. \]

The second inequality in the lemma statement follows from the same proof.
\end{proof}

\begin{corollary}
\label{cor:oddeven}
Let $G=(V,E)$ be a DAG with $|V|=n$ in the \truemem\  model, and $\sigma,\tau$ are schedules of $G$ such that $\nodeweightedprofile{\sigma}\ledot \nodeweightedprofile{\tau}$.  Recall that $\nodeweightedprofile{\cdot }$ is $0$-indexed.
Then there is a pointer-advancement process showing $\nodeweightedprofile{\sigma}\ledot \nodeweightedprofile{\tau}$ (see \cref{defn:cost-seq-superior}), such that: 
\begin{itemize}
    \item When the pointer on $\tau$ advances, the pointer on $\sigma$ must be at an even index.
    \item When the pointer on $\sigma$ advances, the pointer on $\tau$ must be at an odd index.
\end{itemize}
The same statement holds for $\internalprofile{\cdot }$.
\end{corollary}
\begin{proof}
From \cref{lem:oddhigher} we know an odd-indexed entry in $\nodeweightedprofile{\cdot}$ is no lower than its two neighbors. Hence, in any pointer-advancement process showing $\nodeweightedprofile{\sigma}\ledot \nodeweightedprofile{\tau}$, whenever the pointer on $\sigma$ is at an odd index, we can immediately advance it to the next even index which has a lower or equal value, without hurting the property that it points to a lower or equal value than the one pointed value on $\tau$. An analogous claim holds whenever the pointer on $\tau$ is at an even index. Hence we can make the pointer-advancement process have the two desired properties.
\end{proof}

Now we are ready to prove the exchange lemma. 
\exchangelemma*
\begin{proof}
Let $s,t$ be the unique source and unique sink of $G[U]$.  Then $\pi_1=\pi'_1=s, \pi_{|U|}=\pi_{|U|}'=t$. Consider a pointer-advancement process showing $\internalprofile{\pi'}\ledot \internalprofile{\pi}$ that satisfies the property in \cref{cor:oddeven}. For each node $\pi_i$ ($2\le i\le |U|$), consider the time window when the pointer on $\internalprofile{\pi}$ is at the value 
$\memduring{\pi, \pi_i}$
, and let $P_i$ denote the sequence of nodes $\pi_{k}'$ for which the pointer on $\internalprofile{\pi'}$ passes through 
$\memduring{\pi', \pi'_k}$
during this time window. Note that the concatenation $P_2\circ P_3\circ \dots P_{|U|}$ is exactly $(\pi'_2,\pi'_3,\dots,\pi'_{|U|})$.

We define the claimed schedule $\sigma'$ as follows: in schedule $\sigma$, simultaneously for all $2\le i\le |U|$, we replace the occurrence of $\pi_i$ by the sequence $P_i$.   Since $\pi$ is a subschedule of $\sigma$, and $\pi_1=\pi_1'=s$ appears before $P_2$ in $\pi'$, we immediately know that $\pi'$ appears as a subsequence of $\sigma'$.

First we show that $\sigma'$ is indeed a valid schedule of $G$. Suppose it is not, then $G$ contains a directed edge $(\sigma'_i,\sigma'_j)$ (for some $i>j$). There are a few cases:
\begin{itemize}
    \item $\sigma'_i\notin U$ and $\sigma'_j\notin U$. Then these two nodes are not involved by the replacement, and should have the same order as in the original schedule $\sigma$. We know this is impossible because $\sigma$ is a valid schedule of $G$.
    \item $\sigma'_i\in U$ and $\sigma'_j\in U$. Then by definition of our replacement, these two nodes should be ordered by $\sigma'$ in the same way as they are ordered by $\pi'$. We know this is impossible because $\pi'$ is a valid schedule of $G[U]$.
     \item $\sigma'_i\in U$ and $\sigma'_j\notin U$. Since $G[U]$ is a closed subgraph of $G$, by definition, the path $(\sigma'_i,\sigma'_j)$ must go through $t$, meaning that $\sigma'_i = t$. Suppose $\sigma'_i \in P_{k}$, which means $\sigma'$ is generated from the replacement of $\pi_k$. Then, since $\sigma'_j\notin U$ is not involved in the replacement, and it occurs to the left of $\sigma'_i$ in schedule $\sigma'$ (because $j<i$),  we know that before the replacement, $\sigma'_j$ must occur to the left of $\pi_k$ in schedule $\sigma$. However, in $G$ there is a path from $\pi_k$ to $t=\sigma'_i$ (because $t$ is the unique sink of $G[U]$), and then to $\sigma'_j$,  contradicting the fact that $\sigma$ is a valid schedule of $G$.
     \item $\sigma'_i\notin U$ and $\sigma'_j\in U$. This case can be proved in the same way as the previous case.
\end{itemize}

Next, we show $\internalprofile{\sigma'}\ledot \internalprofile{\sigma}$ via a pointer-advancement argument (\cref{defn:cost-seq-superior}). 
Since the prefixes of $\internalprofile{\sigma'}$ and $\internalprofile{\sigma}$ before $s$ are identical,  we  can first advance both pointers to the point where $s$ has just finished, i.e., in both schedules $\sigma$ and $\sigma'$ currently incur cost $\memafter{\sigma, s}$. Then, because of the way we constructed $\sigma'$, it is possible to run the following loop (with the invariant that before/after each loop iteration, both pointers are at even indices):
\begin{itemize}
    \item If the next step in $\internalprofile{\sigma}$ is $\memduring{\sigma, \pi_i}$ for some node $\pi_i \in U$: Then we first advance the pointer on $\internalprofile{\sigma}$ to $\memduring{\sigma, \pi_i}$, and then keep advancing the pointer on $\internalprofile{\sigma'}$ until all nodes in $P_i$ are completed, i.e., it should now point to $\memafter{\sigma', v}$ where $v$ is the last node in $P_i$. Finally we advance the pointer on $\internalprofile{\sigma}$ by one more step to $\memafter{\sigma, \pi_i}$.

    We now argue that this pointer increment sequence maintains the required invariant that the pointer on $\internalprofile{\sigma'}$ always points to a lower memory value than the pointer on $\internalprofile{\sigma}$.
     We have
$\memduring{\sigma, \pi_i} = \memduring{\pi, \pi_i} + \sum_{u \in \mathcal{S}} \sz{u}$ where $\mathcal{S} \subseteq V \setminus U$ is the set of nodes not in $U$ that have been scheduled before $\pi_i$ in $\sigma$ and have at least one unscheduled out-neighbor. Similarly in the schedule $\sigma'$, for any node $v \in P_i$, we have $\memduring{\sigma', v} = \memduring{\pi', v} + \sum_{u \in \mathcal{S}} \sz{u}$ (Note the use of the same set $\mathcal{S}$ in both cases. This is because by induction, the two schedules have scheduled exactly the same set of vertices from outside $U$.) and similarly $\memafter{\sigma', v} = \memafter{\pi', v} + \sum_{u \in \mathcal{S}} \sz{u}$. However, since $\internalprofile{\pi'} \ledot \internalprofile{\pi}$, by construction of $P_i$, we must have that $\memduring{\pi, \pi_i} \geq \memduring{\pi', v} \geq \memafter{\pi', v}$ and the invariant is maintained.
        
    \item Otherwise, if the next step in $\internalprofile{\sigma}$ is $\memduring{\sigma, x}$ for some $x \notin U$: Then we first advance the pointer on $\internalprofile{\sigma}$ to $\memduring{\sigma, x}$. Note that by construction, the next step in $\internalprofile{\sigma'}$ must be $\memduring{\sigma', x}$. So, we then advance the pointer on $\internalprofile{\sigma'}$ by two steps to $\memafter{\sigma', x}$. Finally, we advance the pointer on $\internalprofile{\sigma}$ by one more step to $\memafter{\sigma, x}$.

    Once again, it can be easily verified that the required invariant from dominance is maintained in this case.
\end{itemize}
After $t$ is finished by both  pointers, both sequences are identical, so we can advance both pointers to the very end. The constructed pointer advancement sequence satisfies the requirements of \cref{defn:pointerdef} and thus we have shown that $\internalprofile{\sigma'} \ledot \internalprofile{\sigma}$.

\end{proof}

Next, we show that for any memory profile (simply an odd length sequence of positive numbers), there exists a path graph such that in the computation graph model, the memory profile of the unique schedule equals the given memory profile.

\begin{lemma}
\label{lem:chain}
Let $A = (0=a_0, a_1, ... a_{2n})$ be a memory profile with $a_i \ge 0$ for all $i$. Then there is a computation graph $\fancyg$ whose underlying graph $G$ is a path with the property that the unique valid schedule $\sigma$ on $\fancyg$ is such that $\nodeweightedprofile{\sigma} = A$.
\end{lemma}
\begin{proof}
    Let $G = (V,E)$ be a directed path with $V = \{v_1, v_2, \ldots, v_{n}\}$ as the set of vertices and let $E = \{(v_i, v_{i+1})\}_{i=1}^{n-1}$ as the set of edges. For all $i\leq i \leq n$, let $\sz{v_i} := a_{2i}$ and $\scr{v_i} := a_{2i-1} - \sz{v_i} - \sz{v_{i-1}}$ where we define $\sz{v_{0}} = 0$ for notational convenience.

    Let $B = (0=b_0, b_1, \ldots, b_{2n}) = \profile{\sigma}$ denote the memory profile of the unique schedule $\sigma$ of $G$. For any even $j = 2i$, by definition, we have $b_j =  \memafter{\sigma, v_i}$. However, since $G$ is a path graph, we have $\memafter{\sigma, v_i}= \sz{v_i}$ and hence $b_{2i} = a_{2i}$. Similarly, for and odd $j = 2i-1$, we have $b_j = \memduring{\sigma, v_i}$. Again, since $G$ is a path graph, we have $\memduring{\sigma, v_i} = \sz{v_{i-1}} + \scr{v_i} + \sz{v_i} = a_{2i-1}$ and the lemma follows. 
\end{proof}

Finally, we can show the necessary linearization lemma.

\linearization*
\begin{proof}
Given a dominant schedule $\pi$ of the isolated subgraph $G[U]$, we use \cref{lem:chain} to construct a path graph whose unique schedule has memory profile identical to that of $\pi$. We replace the isolated subgraph $G[U]$ in $G$ by the newly constructed path graph.
It is clear that any schedule of $G'$ is also a valid schedule of $G$ (since the unique ordering on nodes on $U$ enforced by the path is a valid topological ordering). 

Moreover, take any schedule $\sigma$ of $G$. We can use \cref{lem:exchangelinearize} to replace the subschedule of $\sigma$ on $G[U]$ by the dominant schedule $\pi$ of $G[U]$, so that the obtained schedule $\sigma'$ dominates $\sigma$. Moreover, $\sigma'$ is also a valid schedule of $G'$ since it contains $\pi$ as a subschedule. So the dominant schedule of $G'$ also dominates $\sigma'$. By transitivity \cref{lem:superiority-reflexive-transitive}, we know the dominant schedule of $G'$ dominates $\sigma$. Since $\sigma$ is arbitrary, we have proved that the dominant schedule of $G'$ is a dominant schedule of $G$.
\end{proof}

\section{Segmentation (Deferred Proofs from \cref{sec:segments})}
\label{sec:appendix:segments}

We first prove \cref{lem:segment-peak-valley}.
\segmentpeakvalley*
\begin{proof}
  We prove the equalities one-by-one. Since the left-hand-sides are always the appropriate endpoint of the region, it suffices to prove that the given index minimizes (for valleys) or maximizes (for peaks) the function $x \to \nodesumcost{\pi_{1:x}}$ within the appropriate range.
  
  We now consider the first inequality. If $i = 0$, then $\valley{i} = \mincut$ and it is a minimizer over all of $1:(k-1)$. If $i > 0$, then it is a minimizer over $\peak{i-1}:(k-1)$ which is a superset of $\valley{i}:\peak{i}$. If $i < 0$, then it is a minimizer over $1:\peak{i}$ which is a superset of $\valley{i}:\peak{i}$.
  
  We now consider the second and third inequalities. If $i \ge 0$, then it is a maximizer over $\valley{i}:(k-1)$ which is a superset of $\valley{i}:\peak{i}$ and $\peak{i}:\valley{i+1}$. If $i < 0$, then it is a maximizer over $1:\valley{i+1}$ which is a superset of $\valley{i}:\peak{i}$ and $\peak{i}:\valley{i+1}$.
  
  We now consider the fourth inequality. If $i = -1$, then $\valley{i+1} = \mincut$ and it is a maximizer over all of $1:(k-1)$. If $i > -1$, then $\valley{i+1}$ is a maximizer over $\peak{i}:(k-1)$ which is a superset of $\peak{i}:\valley{i+1}$. If $i < -1$, then $\valley{i+1}$ is a maximizer over $1:\peak{i+1}$ which is a superset of of $\peak{i}:\valley{i+1}$.
\end{proof}

Using this technical lemma, now we can prove \cref{thm:segment-consecutive}.
\segmentconsecutive*
\begin{proof}
   The high-level proof plan is as follows. We construct $\sigma'$ from $\sigma$ by going segment-by-segment and moving all of the nodes of that segment to the location its peak currently is. This will dominate the original schedule because for all the nodes between the left endpoint of the segment (a valley) and the peak, the contribution of this path to memory will drop to the value of the left endpoint. Similarly, for all the nodes between the peak and the right endpoint of the segment (also a valley), the contribution of this path to memory will drop to the value of the right endpoint.
  
  Let $\segment$ be the any segment that does not appear consecutively in $\sigma$. We will make it appear consecutively while only moving its nodes within $\sigma$ to arrive at some $\sigma''$ such that $\sigma'' \ledot \sigma$. Each time we do this, an additional segment appears consecutively within $\sigma$, so eventually all segments must appear consecutively within $\sigma$.
  
  By definition, $\segment = \valley{i} + 1:\valley{i+1}$ for some $i$. Let $\sigma'''$ be $\sigma$ with every node of $\pi_{\segment}$ removed except $\pi_{\peak{i}}$. Let $\sigma''$ be $\sigma'''$ with $\pi_{\peak{i}}$ replaced with $\pi_{\segment}$ (consecutively). Note that this does not violate any precedence constraints because $\pi$ is an isolated subgraph for node sum.
  
  To prove dominance, we appeal to \cref{defn:cost-seq-superior} and construct a sequence of pairs of indices. As long as neither index refers to a vertex in $\pi_{\segment}$, they must refer to the same vertex (since $\sigma$ and $\sigma''$ agree on everything except the order of $\pi_{\segment}$). In this case, we can safely increment both by examining the weight of the vertex referred to: if it is positive, then we first increment the dominated schedule, $\sigma$; negative, then we first increment the dominating schedule, $\sigma'$.
  
  On the other hand, if at least one of them points to a vertex in $\pi_{\segment}$, we handle it as follows. If the index of $\sigma$ points to such a node, increment it immediately until such a time that it points to $\pi_{\peak{i}}$. Since this is where we inserted $\pi_{\segment}$ in $\sigma''$, we can now advance the index of $\sigma''$ past all of $\pi_{\segment}$. We then resume increment the index of $\sigma$ every time we see a node from $\pi_{\segment}$.
  
  To finish, we need to show that the memory usage of $\sigma$ at its pointer is at least the memory usage of $\sigma''$ at its pointer. Consider two paired indices produced by this construction; $j$ which indexes into $\sigma$ and $j''$ which indexes into $\sigma''$. The situation becomes more clear when we split memory usage into the contribution from the $\set(\pi_{\segment})$ and the contribution from $V \setminus \set(\pi_{\segment})$:
  \begin{align*}
    \nodesumcost{\set(\sigma_{1:j})}
      &= \sum_{v \in \set(\sigma_{1:j})} w_v \\
      &= \sum_{v \in \set(\sigma_{1:j}) \cap \set(\pi_{\segment})} w_v +
         \sum_{v \in \set(\sigma_{1:j}) \setminus \set(\pi_{\segment})} w_v \\
    \nodesumcost{\set(\sigma''_{1:j''})}
      &= \sum_{v \in \set(\sigma''_{1:j''})} w_v \\
      &= \sum_{v \in \set(\sigma''_{1:j''}) \cap \set(\pi_{\segment})} w_v +
         \sum_{v \in \set(\sigma''_{1:j''}) \setminus \set(\pi_{\segment})} w_v
  \end{align*}
  First, we consider $\set(\sigma_{1:j}) \cap \set(\pi_{\segment})$ versus $\set(\sigma''_{1:j''}) \cap \set(\pi_{\segment})$, i.e. the contribution from $\set(\pi_{\segment})$. We break into cases: (i) if $\set(\sigma_{1:j})$ does not contain $\pi_{\peak{i}}$, then $\set(\sigma''_{1:j''}) \cap \set(\pi_{\segment})$ is empty, (ii) if $\set(\sigma_{1:j})$ does contain $\pi_{\peak{i}}$ and $\sigma_j \ne \pi_{\peak{i}}$, then $\set(\sigma''_{1:j''}) \cap \set(\pi_{\segment})$ is $\set(\pi_{\segment})$, and (iii) if $\set(\sigma_{1:j})$ does contain $\pi_{\peak{i}}$ and $\sigma_j \ne \pi_{\peak{i}}$, then $\set(\sigma''_{1:j''}) \cap \set(\pi_{\segment})$ could be anything from empty to all of $\set(\pi_{\segment})$.
  
  We invoke \cref{lem:segment-peak-valley} to assert that the memory contribution to $\nodesumcost{\set(\sigma''_{1:j''})}$ is smaller than the memory contribution to $\nodesumcost{\set(\sigma_{1:j})}$ in all cases. Case (i) is covered by the first equality: $\sigma$ is between $\valley{i}$ and $\peak{i}$ but $\sigma''$ is still at $\valley{i}$. Case (ii) is covered by the fourth equality: $\sigma$ is between $\peak{i}$ and $\valley{i}$ but $\sigma''$ has made it to $\valley{i+1}$. Case (iii) is covered by the second and third equalities: $\sigma$ is at $\peak{i}$ and $\sigma''$ is between $\valley{i}$ and $\valley{i+1}$.
  
  Next, we consider $\set(\sigma_{1:j}) \setminus \set(\pi_{\segment})$ versus $\set(\sigma''_{1:j''}) \setminus \set(\pi_{\segment})$, i.e. the contribution from outside $\set(\pi_{\segment})$. This is easy because we advance the indices together for such nodes; the memory contribution to $\nodesumcost{\set(\sigma''_{1:j''})}$ is smaller than the memory contribution to $\nodesumcost{\set(\sigma_{1:j})}$ due to our decision to favor the former by giving it positive nodes second and negative nodes first.
  
  Hence $\sigma''$ dominates $\sigma$ and by repeatedly applying this we can generate the desired $\sigma'$, completing the proof.
\end{proof}

Finally, we prove \cref{thm:segment-merge}.
\segmentmerge*
\begin{proof}
  We begin by applying \cref{thm:segment-consecutive} to $\sigma$ to assume without loss of generality that $\sigma$ arranges the nodes of each segment consecutively. From here, to get from $\sigma$ to $\sigmasort$, we can just move entire segments as scheduling units. In particular, we just need to be able to exchange adjacent segments that are either in the wrong order according to their sort value or which have equal sort value. That is enough to get from $\sigma$ to $\sigmasort$ by simulating bubble-sort (and using $\sigmasort$ to break ties).
  
  When proving that such an exchange produces a schedule that dominates the original one, it suffices to only consider the two adjacent segments being swapped and none of the other segments. This is due to \cref{lem:concatendominance}: everything before the two segments identical and hence dominates due to reflexivity of dominance (\cref{lem:superiority-reflexive-transitive}), as well as everything after the two segments. Hence if we can prove that the two segments swapped dominates the two segments unswapped, we can concatenate all these pieces and get that the resulting schedule dominates the original.
  
  At this point, we break into three cases. Either (i) one segment is nonnegative and one segment is negative, (ii) both segments are nonnegative, or (iii) both segments are negative.
  
  \textbf{Case (i): One segment is nonnegative and one segment is negative}. In this case, the original order must be the nonnegative segment followed by the negative segment. We prove dominance using \cref{defn:cost-seq-superior} and construct a sequence of pairs of indices. We advance the index for the original order to its peak, then completely advance the index for the new order, then finish advancing the index for the original order.
  
  We just need to prove that at all times, the index for the original order incurs more memory usage than the index for the new order. Since everything before these two segments was the same, both memory profiles begin at the same memory usage. When we advance the index for the original order to its peak, we definitely (step 1) advance from the left valley of the nonnegative segment its peak and may additionally (step 2) advance from peak of the nonnegative segment to its right valley and (step 3) from the left valley of the negative segment to its peak. All steps are okay due to \cref{lem:segment-peak-valley}. As we move along step 1, we know that memory usage remains higher than the left valley due to the first equality of the lemma. As we move along step 2, we know that memory usage remains higher than the right valley due to the fourth equality of the lemma, which is higher than the left valley due to being a nonnegative segment. As we move along step 3, we know that memory usage remains higher than the left valley due to the first equality of the lemma.
  
  When we next completely advance the index for the new order, we invoke the second and third inequalities of \cref{lem:segment-peak-valley}. In particular, the peak of a segment dominates all other points in the segment. Since the index into the original schedule is currently sitting at an original peak, we only need to show the memory usage at that original peak is larger than the memory usage at either of our new peaks, which are larger than the memory usage of their respective segments. Here is where we use the fact that one of the segments is nonnegative and one is negative. In the original order, the nonnegative segment is relative to the initial memory and the negative segment is relative to inital memory plus the difference between the nonnegative segment's valleys (which is nonnegative). In the new, swapped order, the negative segment is relative to the initial memory, and the nonnegative segment is relative to the initial memory plus the difference between the negative segment's valleys (which is negative). This implies that the original nonnegative peak is higher than the new nonnegative peak and the original negative peak is at least as high as the new negative peak. Hence an original peak must be higher, and this pointer advancement is safe.
  
  Finally, we finish advancing the index for the old order. This is just a mirror case to advancing the index to the old peak. We may (step 1) advance from the peak of the nonnegative segment to its right valley and (step 2) from the left valley of the negative segment to its peak. Then we definitely (step 3) advance from the peak of the negative segment to its right valley. We again invoke \cref{lem:segment-peak-valley}. As we move along step 1, we know that memory usage remains higher than the right valley due to the fourth equality, and we are missing the total contribution of the negative segment, which is negative. As we move along step 2, we know that memory usage remains higher than the left valley due to the first equality, which is higher than the right valley due to being a negative segment. As we move along step 3, we know that the memory usage remains higher than the right valley due to the fourth equality. That completes the analysis of this case.
  
  \textbf{Case (ii): Both segments are nonnegative} We again prove dominance using \cref{defn:cost-seq-superior} and construct a sequence of pairs of indices. We advance the index for the original order to the peak of the second segment, then completely advance the index for the new order, then finish advancing the index for the original order.
  
  We just need to prove that at all times, the index for the original order incurs more memory usage than the index for the new order. Since everything before these two segments was the same, both memory profiles begin at the same memory usage. When we advance the index for the original order to the peak of the second segment, we definitely (step 1) advance from the left valley of the first segment to its peak, (step 2) advance from peak of the first segment to its right valley, and (step 3) advance from the left valley of the second segment to its peak. All steps are okay due to \cref{lem:segment-peak-valley}. As we move along step 1, we know that memory usage remains higher than the left valley due to the first equality of the lemma. As we move along step 2, we know that memory usage remains higher than the right valley due to the fourth equality of the lemma, which is higher than the left valley due to being a nonnegative segment. As we move along step 3, we know that memory usage remains higher than the left valley due to the first equality of the lemma.
  
  When we next completely advance the index for the new order, we invoke the second and third inequalities of \cref{lem:segment-peak-valley}. In particular, the peak of a segment dominates all other points in the segment. We claim that the index into the original schedule is currently sitting at an original peak (i.e. the second original peak is at least as high as the first original peak), and that this peak is higher than either of the new peaks, which are in turn higher than their respective segments. Here is where we use the definition of sort value for nonnegative segments. In the original order, the first segment has equal or higher sort value. This means it has an equal or smaller gap between its peak and its right valley (due to the fact we inverted the positive numbers). In the original order, the first peak is initial memory plus the difference between the first segment's valleys (which is nonnegative) plus the first segment's gap. In the original order, the second peak is initial memory plus the difference between the first segment's valleys (which is nonnegative) plus the second segment's valleys (which is nonnegative) plus the second segment's gap. Again, since the second segment has a equal or larger gap, the second peak is at least as high as the first one, in the original order.
  
  In the new, swapped order, the first peak to occur is initial memory plus the difference between the originally-second segment's valleys plus the originally-second segment's gap. This is dominated by the original second peak. The second peak to occur is initial memory plus the both differences between the segment's valleys plus the first segment's gap. This is also dominated by the original second peak. Hence this part of the pointer advanacement is safe.
  
  Finally, we finish advancing the index for the old order. We advance from the peak of the negative segment to its right valley. We again invoke \cref{lem:segment-peak-valley}. We know that the memory usage remains higher than the right valley due to the fourth equality. That completes the analysis of this case.
  
  \textbf{Case (iii): Both segments are negative} This is essentially a mirror case to case (ii). We again prove dominance using \cref{defn:cost-seq-superior} and construct a sequence of pairs of indices. We advance the index for the original order to the peak of the first segment, then completely advance the index for the new order, then finish advancing the index for the original order.
  
  We just need to prove that at all times, the index for the original order incurs more memory usage than the index for the new order. Since everything before these two segments was the same, both memory profiles begin at the same memory usage. When we advance the index for the original order to the peak of the first segment, we advance from the left valley of the first segment to its peak. This step is okay due to \cref{lem:segment-peak-valley}. We know that memory usage remains higher than the left valley due to the first equality of the lemma.
  
  When we next completely advance the index for the new order, we invoke the second and third inequalities of \cref{lem:segment-peak-valley}. In particular, the peak of a segment dominates all other points in the segment. We claim that the index into the original schedule is currently sitting at an original peak (i.e. the first original peak is at least as high as the second original peak), and that this peak is higher than either of the new peaks, which are in turn higher than their respective segments. Here is where we use the definition of sort value for negative segments. In the original order, the first segment has equal or higher sort value. This means it has an equal or larger gap between its peak and its left valley (due to the fact we inverted the negative numbers). In the original order, the first peak is initial memory plus the first segment's gap. In the original order, the second peak is initial memory plus the difference between the first segment's valleys (which is negative) plus the second segment's gap. Again, since the first segment has a equal or larger gap, the first peak is at least as high as the second one, in the original order.
  
  In the new, swapped order, the first peak to occur is initial memory plus the plus the originally-second segment's gap. This is dominated by the original first peak. The second peak to occur is initial memory plus the differences between the originally-second segment's valleys (negative) plus the originally-first segment's gap. This is also dominated by the original first peak. Hence this part of the pointer advanacement is safe.
  
  Finally, we finish advancing the index for the old order. We (step 1) advance from the peak of the first segment to its right valley, (step 2) advance from the left valley of the second segment to its peak, and (step 3) advance from the peak of the second segment to its right valley. We again invoke \cref{lem:segment-peak-valley}. As we move along step 1, we know that memory usage remains higher than the right valley due to the fourth equality, and we are missing the total contribution of the negative segment, which is negative. As we move along step 2, we know that memory usage remains higher than the left valley due to the first equality, which is higher than the right valley due to being a negative segment. As we move along step 3, we know that the memory usage remains higher than the right valley due to the fourth equality. That completes the analysis of this case.
  
  We have now finished the proof of all cases, showing that exchanging segments out of sort order or with equal sort value is okay. Hence we can produce the desired $\sigmasort$ by conducting a bubble-sort. This completes the proof.
\end{proof}

\end{document}